\documentclass[11pt,a4paper,oneside]{article}
\pdfoutput=1

\usepackage[]{authblk}

\usepackage[a4paper,textwidth=16cm,textheight=24cm,centering]{geometry}

\usepackage[]{amsmath}
	\numberwithin{equation}{section}
\usepackage[]{amssymb}
\usepackage{amsfonts}
\usepackage{mathrsfs}

\DeclareMathOperator{\R}{\mathbb{R}}
\DeclareMathOperator{\C}{\mathbb{C}}
\DeclareMathOperator{\N}{\mathbb{N}}
\DeclareMathOperator{\Z}{\mathbb{Z}}

\DeclareMathOperator{\cs}{\mathbb{S}}
\DeclareMathOperator{\ct}{\mathbb{T}}
\newcommand{\Th}{\mathscr{Q}}
\newcommand{\scH}{\mathscr{H}}
\newcommand{\mz}{\mathcal{Z}}

\newcommand{\mN}{\mathcal{N}}
\newcommand{\mM}{\mathcal{M}}

\newcommand{\mI}{\mathcal{I}}

\newcommand{\mC}{\mathcal{C}}
\newcommand{\mO}{\mathcal{O}}

\newcommand{\mH}{\mathcal{H}}
\newcommand{\mA}{\mathcal{A}}
\newcommand{\scA}{\mathscr{A}}
\newcommand{\scD}{\mathscr{D}}

\newcommand{\sQ}{\mathsf{Q}}
\newcommand{\sF}{\mathsf{F}}

\newcommand{\sV}{\mathsf{V}}
\newcommand{\sU}{\mathsf{U}}

\newcommand{\tws}{\mathbf{s}}
\newcommand{\twt}{\mathbf{t}}

\newcommand{\dd}{\mathrm{d}}
\newcommand{\ii}{\mathtt{i}}

\newcommand{\tr}{\mathrm{Tr}}
\DeclareMathOperator{\sh}{sh}
\DeclareMathOperator{\ch}{ch}
\DeclareMathOperator{\hs}{HS}
\newcommand{\keff}{\kappa^{\mathrm{eff}}}
\newcommand{\kk}{\mathsf{k}}
\newcommand{\sQred}{\sQ^{\mathrm{red}}}
\newcommand{\pair}[2]{\langle #1,#2 \rangle_{\alpha}}
\newcommand{\gnode}[1]{\overset{#1}{\bigcirc}}
\newcommand{\fnode}[1]{\overset{#1}{\Box}}
\newcommand{\csnode}[2]{\overunderset{#1}{#2}{\bigcirc}}
\newcommand{\multe}[1]{\overset{#1}{\text{-----}}}

\makeatletter
\newcommand\xleftrightarrow[2][]{%
  \ext@arrow 9999{\longleftrightarrowfill@}{#1}{#2}}
\newcommand\longleftrightarrowfill@{%
  \arrowfill@\leftarrow\relbar\rightarrow}
\makeatother

\usepackage[utf8]{inputenc}
\usepackage[english]{babel}
\usepackage[pdftex]{graphicx}
\usepackage[dvipsnames]{xcolor}

\usepackage{amsthm}
	\newtheorem{conjnum}{Conjecture}

	\newtheorem{thm}{Theorem}[subsection]
	\newtheorem{lem}[thm]{Lemma}
	\newtheorem{prop}[thm]{Proposition}
	\newtheorem{cor}[thm]{Corollary}
	\newtheorem{conj}[thm]{Conjecture}
	
	\theoremstyle{definition}
	
	\newtheorem{defin}[thm]{Definition}

	\newtheorem{rmk}[thm]{Remark}
	\newtheorem{hypt}[thm]{Hypothesis}
	
	\newtheorem{hypnum}[conjnum]{Hypothesis}

\usepackage[linktocpage=true,colorlinks=true,citecolor=Violet,linkcolor=OliveGreen,urlcolor=Violet]{hyperref}

\usepackage[toc,page]{appendix}
\usepackage[nottoc]{tocbibind}
\usepackage{caption}
\usepackage{enumerate}

\usepackage{tikz}
		\usetikzlibrary{calc,er,positioning,arrows.meta}

\usepackage[numbers,compress,square,comma]{natbib}	
\newcommand\aafootnote[1]{%
  \begingroup
  \renewcommand\thefootnote{}\footnote{\texttt{#1}}%
  \addtocounter{footnote}{-1}%
  \endgroup
}

\newenvironment{tenumerate}{
\begin{enumerate}[(i)]
  \setlength{\itemsep}{0pt}
}{\end{enumerate}}

\begin{document}
\bibliographystyle{myJHEP}
\captionsetup[figure]{labelfont={sc,small},labelformat={default},labelsep=period,font=small}
\captionsetup[table]{labelfont={sc,small},labelformat={default},labelsep=period,font=small}
{
\pagenumbering{roman} 
	\title{\huge\textbf{Twisted traces and quantization of moduli stacks of 3d \texorpdfstring{$\mathcal{N}=4$}{N=4} Chern--Simons-matter theories}}
	\author{{\Large Leonardo Santilli}}
	\affil{\it Center for Mathematics and Interdisciplinary Sciences,\protect\\ Fudan University, Shanghai, 200433, China.}
	\affil{\it Shanghai Institute for Mathematics and Interdisciplinary Sciences (SIMIS),\protect\\ Shanghai, 200433, China.}
	\date{  }
	
	\maketitle
	\thispagestyle{empty}
	
\aafootnote{santilli at simis dot cn}

\begin{abstract}
	We conjecture, and show in a plethora of examples, that the sphere partition function of 3d $\mathcal{N}=4$ Chern--Simons-matter theories equals a sum of twisted traces on tensor products of Verma modules over the quantization of the moduli spaces of vacua. This extends a conjecture of Gaiotto--Okazaki to Chern--Simons-matter theories. We also show that the partition function of every Abelian gauge theory with higher charges has such twisted trace decomposition, and uncover new Abelian dualities between theories with and without Chern--Simons couplings.
\end{abstract}

\clearpage
{
\tableofcontents}
\thispagestyle{empty}
}

\clearpage
\pagenumbering{arabic}
\setcounter{page}{1}
	\renewcommand*{\thefootnote}{\fnsymbol{footnote}}
	\setcounter{footnote}{0}

\section{Introduction}
We initiate the study of the quantized moduli spaces of vacua of 3d $\mN=4$ Chern--Simons-matter theories, through the lens of the so-called sphere quantization \cite{Gaiotto:2023hda,Gaiotto:2023kuc}.\par

\subsubsection*{Background: 3d \texorpdfstring{$\mN=4$}{N=4} gauge theories and symplectic singularities}
Every 3d $\mN=4$ quiver gauge theory yields two distinguished rings, whose spectra are called Coulomb and Higgs branch. Both the Coulomb branch \cite{Weekes:2020rgb,Bellamy:2023lqf} and the Higgs branch \cite{Bellamy:2016uju} are symplectic singularities. 
A pivotal fact about 3d $\mN=4$ Coulomb and Higgs branches is that they admit a deformation quantization \cite{Yagi:2014toa,Chester:2014fya,Beem:2016cbd}. These quantized algebras will be denoted respectively 
\begin{equation}
\label{eq:ACAH}
	\mA_{\hbar}^{\mC}, \qquad \mA_{\hbar}^{\mH} .
\end{equation}\par
\medskip
Adding supersymmetric Chern--Simons couplings generically breaks the supersymmetry down to $\mN=3$; however, for suitable choices of Chern--Simons levels and matter content, $\mN=4$ supersymmetry is preserved \cite{Gaiotto:2008sd,Assel:2022row}. 
Every 3d $\mN=4$ Chern--Simons-matter theory also yields two distinguished moduli spaces, henceforth referred to as the A-branch $\mM_A$ and the B-branch $\mM_B$, which are expected to be symplectic singularities.\par
\medskip
This work aims at studying \emph{quantized} A- and B-branches of 3d $\mN=4$ Chern--Simons-matter theories. To achieve this, the description of $\mM_A$ and $\mM_B$ is refined to distinguish among inequivalent moduli stacks with the same coarse moduli space.

\subsubsection*{Background: Sphere quantization}
The sphere partition function is a powerful tool to study the quantized Coulomb and Higgs branch algebras \eqref{eq:ACAH} \cite{Dedushenko:2016jxl,Dedushenko:2017avn,Dedushenko:2018icp,Dedushenko:2019mzv,Fan:2019jii,Dedushenko:2025qxe}. Consider a 3d $\mN=4$ quiver gauge theory subject to: 
\begin{hypnum}\label{hyp:massive}
The Coulomb branch is fully resolved and has isolated fixed points under the action of a torus isometry.
\end{hypnum}
Let us clarify from the outset that the present work studies theories for which this \emph{fails}.\par
Nonetheless, under Hypothesis \ref{hyp:massive}, the protected sphere correlation functions \cite{Dedushenko:2016jxl} provide twisted traces (Definition \ref{def:twtrace}) on the quantized algebras \cite{Gaiotto:2023kuc}. Furthermore, the authors of \cite{Bullimore:2016nji} associate Verma modules over $\mA_{\hbar}^{\mC}$ and $\mA_{\hbar}^{\mH}$ to each torus fixed point. Gaiotto--Okazaki \cite{Gaiotto:2019mmf} conjecture that the sphere partition function encodes twisted traces on these Verma modules.\par
\begin{conjnum}[{\cite{Gaiotto:2019mmf}}]\label{c:GO}
	Let $\mz$ be the sphere partition function of a 3d $\mN=4$ theory satisfying Hypothesis \ref{hyp:massive}, and denote by $\zeta$ and $m$ the equivariant parameters of the torus action on the Coulomb and Higgs branch, respectively. In an open subset of the parameter space, 
	\begin{equation}
		\mz = \sum_{\alpha \ : \ \text{\rm fixed pt} } \ii^{- \tilde{\kappa} } e^{\ii 2 \pi \pair{m}{\zeta}} \hat{\chi}_{\alpha}^{\mC}  \left(e^{-2\pi \zeta}\right)  \hat{\chi}_{\alpha}^{\mH} \left( e^{-2\pi m} \right) ,
	\end{equation}
	where $\ii=\sqrt{-1}$, $\tilde{\kappa}\in \Z/4\Z$, $\pair{\cdot}{\cdot}$ is a bilinear pairing defined in \eqref{eq:pairmzeta}, and $\hat{\chi}_{\alpha}^{\mC}$ and $\hat{\chi}_{\alpha}^{\mH}$ are, respectively, twisted traces on the Verma modules over $\mA_{\hbar}^{\mC}$ and $ \mA_{\hbar}^{\mH}$ attached to the fixed point $\alpha$.
\end{conjnum}
In a nutshell, the sphere partition function is expressible as a sum of pairwise products of twisted traces of $\mA_{\hbar}^{\mC}$- and $ \mA_{\hbar}^{\mH}$-modules. The conjecture was further elucidated and supported in \cite{Bullimore:2020jdq}. It is shown for Abelian theories in \cite{Bullimore:2020jdq}, and argued in general \cite{Gaiotto:2023hda}, that the twist is $(-1)^{R}$, where $R$ is the R-charge, i.e. a $\Z$-grading induced by the contracting $\C^{\ast}$-action on the branch.\par
\medskip
These claims rely crucially on Hypothesis \ref{hyp:massive}, but the latter generically fails in 3d $\mN=4$ Chern--Simons-matter theories. The goal of the current work is to study how to adapt them in presence of Chern--Simons couplings.\par

\subsection{Main results}
The present work investigates whether a version of Conjecture \ref{c:GO} holds for 3d $\mN=4$ Chern--Simons-matter theories. These theories never satisfy Hypothesis \ref{hyp:massive}, unless the Chern--Simons levels are restricted to $\{0, \pm 1\}$.\par
The first main result is Conjecture \ref{myconj1}.
\begin{itemize}
	\item The first novelty is to consider the quantization $\mA^{A}_{\hbar},\mA^{B}_{\hbar}$ of $\mM_A,\mM_B$, and the Verma modules over them.
		\begin{itemize}
			\item[$\triangleright$] The distinction between the moduli stacks $\mM_A,\mM_B$ and their underlying algebraic varieties $M_A,M_B$ is emphasized.
			\item[$\triangleright$] The quantization of $\mM_A,\mM_B$ is studied in \S\ref{sec:VermaCS}. Pairs of Verma modules $\scH^{A}_{\alpha} ,\scH^{B}_{\alpha} $ over $\mA^{A}_{\hbar},\mA^{B}_{\hbar}$ are associated with the torus fixed points $p_{\alpha} \in \left(M_A^{\ct_A}\right)_{\mathrm{red}}$.
		\end{itemize}
	\item The sphere partition function $\mz$ is written as a sum of twisted traces:
		\begin{equation}
		\label{intro:ZCS}
			\mz = \sum_{p_{\alpha} \in \left(M_A^{\ct_A}\right)_{\mathrm{red}}}  S_{\alpha}~e^{\ii 2 \pi \pair{ \zeta}{ m} }\hat{\chi}_{\alpha} \left(e^{-2\pi \zeta} , e^{-2\pi m} \right) ,
		\end{equation}
		where $\hat{\chi}_{\alpha}$ is a twisted trace on $\scH^{A}_{\alpha} \otimes \scH^{B}_{\alpha} $. In general, there is no factorization, and the complete expression for $\hat{\chi}_{\alpha}$ is \eqref{eq:conj1Z}.
		\begin{itemize}
			\item[$\triangleright$] $S_{\alpha}$ in \eqref{intro:ZCS} is a product of supersymmetric pure Chern--Simons partition functions.
			\item[$\triangleright$] The twist in $\hat{\chi}_{\alpha}$ is
			\begin{equation}\label{intro:twist}
				\exp \left\{ \ii \frac{2 \pi}{\kappa_{\alpha}} \left( \frac{1}{2} R_A \otimes 1_B + \frac{1}{2} 1_A \otimes R_B +R_A \otimes R_B\right) \right\},  
			\end{equation}
			where $R_A$ and $R_B$ are, respectively, the $\Z$-grading operators on the A- and B-branch, i.e. the R-charges, and $\kappa_{\alpha}$ is a positive integer which depends on the Chern--Simons levels. For a theory with Chern--Simons levels in the set $\{ 0, \pm \kappa\}$, $\kappa_{\alpha} \in \{1, \kappa\} \ \forall \alpha$.\par
			The lack of factorization of each summand into a product of twisted traces, one on $\scH^{A}_{\alpha} $ and one on $ \scH^{B}_{\alpha} $, is due to the explicit dependence of the twist on the Chern--Simons levels through \eqref{intro:twist}.\par
			\item[$\triangleright$] In some instances, the twisted traces do factorize: 
				\begin{equation}
					\hat{\chi}_{\alpha} \left(e^{-2\pi \zeta} , e^{-2\pi m} \right) =  \hat{\chi}_{\alpha}^{A} (e^{-2\pi \zeta})  \hat{\chi}_{\alpha}^{B} \left( e^{-2\pi m} \right) ,
				\end{equation}
				where $\hat{\chi}_{\alpha}^{A}$ and $\hat{\chi}_{\alpha}^{B}$ are, respectively, twisted traces on the Verma modules $\scH^{A}_{\alpha}$ and $\scH^{B}_{\alpha} $, twisted by $(-1)^{R_A}$ and $(-1)^{R_B}$. 
		\end{itemize}
		\end{itemize}\par
To elucidate these statements in a simple example, consider $U(1)_{\kappa} \times U(1)_{-\kappa}$ gauge theory with one hypermultiplet in the bi-fundamental representation. For this model, analyzed in \S\ref{sec:ex2node}, $\mM_A = \left[ \C^2 / \Z_{\kappa} \right]$, and $\mM_B$ is a point with a $\Z_{\kappa}$-action. There are no mass parameters to resolve the du Val singularity $M_A=\C^2 /\Z_{\kappa}$, and its quantization $\mA^{A}_{\hbar}$ is given explicitly in \cite{Etingof:2020fls}. The category $\mA^{A}_{\hbar}\text{-mod}$ contains only one Verma module, unlike the quantization of a crepant resolution of the du Val singularity. The partition function of this theory equals a twisted trace on this Verma module, with twist $(-1)^{\frac{1}{\kappa} R_A}$ in agreement with \eqref{intro:twist}, multiplied by the three-sphere partition functions of $U(1)_{\kappa}$ and $U(1)_{-\kappa}$ supersymmetric Chern--Simons theories.\par
\medskip
The second main result is a proof of a formula generalizing Conjecture \ref{c:GO} to Abelian gauge theories \emph{without} Chern--Simons couplings, but with non-minimal charges. These theories belong to the familiar class of 3d $\mN=4$ quiver gauge theories, but they do not satisfy Hypothesis \ref{hyp:massive}. Therefore, their partition function is beyond the scope of Conjecture \ref{c:GO}.
\begin{itemize}
	\item Again, the distinction between Coulomb and Higgs stacks and the underlying algebraic varieties is emphasized; their quantization is discussed in Theorem \ref{thm:quantumhyper}.
	\item Corollary \ref{cor:Zhighertrace} evaluates the sphere partition function as a sum of twisted traces, in full generality for Abelian gauge theories. It is a sum of as many terms as the number of torus fixed points on the Coulomb branch.
		\begin{itemize}
			\item[$\triangleright$] In some orbifolds of theories that satisfy Hypothesis \ref{hyp:massive}, each summand factorizes as a product of two twisted traces (Corollary \ref{cor:hypertwtr}), attaining the form of Conjecture \ref{c:GO}.
			\item[$\triangleright$] In general, however, there is no factorization, and each piece is expressed as a single trace on the tensor product of two Verma modules, one over the Coulomb and one over the Higgs branch algebras. 
		\end{itemize}
\end{itemize}\par
\medskip
The third main result is that, given an Abelian Chern--Simons-matter A-type quiver, there exists an Abelian quiver gauge theory $\sQ^{\prime}$, without Chern--Simons couplings but with non-minimal charges, such that
\begin{equation}\label{eq:CStoQprime}
	\mM_A \cong \mC (\sQ^{\prime}) , \qquad \mM_B \cong \mH (\sQ^{\prime}) , \qquad \mz = \mz_{\sQ^{\prime}} ,
\end{equation}
where the first two are isomorphisms of \emph{stacks}, and the third is the equality of the sphere partition function $\mz$ of the Chern--Simons-matter quiver and that of $\sQ^{\prime}$. This fact is extensively tested in examples in \S\ref{sec:ex} and proposed in general (Conjecture \ref{myconjmag}).\par
What is significant about this result is that, given an Abelian Chern--Simons-matter theory, the algorithm of \S\ref{sec:magQprime} explicitly provides $\sQ^{\prime}$ by requiring that $\mM_A$ and $\mC (\sQ^{\prime})$ have the same coarse moduli space, and that the two theories have the same number of resolution parameters. Then \eqref{eq:CStoQprime} turns out to hold. In other words, given an A-branch or Coulomb branch as a hypertoric stack, the B-branch or Higgs branch appears to be uniquely determined as a stack, and so is the sphere partition function, possibly up to an overall normalization. Schematically:
\begin{equation}
	\mM_A \cong \mC (\sQ^{\prime}) \text{ as stacks} \quad \Longrightarrow \quad \mM_B \cong \mH (\sQ^{\prime})  \text{ as stacks, and } \  \mz = \mz_{\sQ^{\prime}} .
\end{equation}\par
In field-theoretic terms, requiring not only that the Coulomb branch of $\sQ^{\prime}$ matches $M_A$, but also that the two theories have the same number of mass parameters, appears to automatically determine the match of both branches, of the 1-form symmetry, of the sphere partition function, and of the superconformal index.\par

\subsection{Outline}
The rest of this work is organized as follows. \S\ref{sec:prelim} is an overview section, introducing the main objects of study and setting the notation. \S\ref{sec:review3dN4} briefly recalls the salient features of 3d $\mN=4$ quiver gauge theories, while Chern--Simons-matter theories are introduced in \S\ref{sec:review3dCS}.\par
\S\ref{sec:highercharge} is devoted to Abelian quiver gauge theories with representations of non-minimal highest weight. \S\ref{sec:higherquant} discusses the quantization of the Coulomb and Higgs moduli stacks when the greatest common divisor of the highest weights is not necessarily one. Then, the sphere partition function is used in \S\ref{sec:HigherChargeSphere} to produce twisted traces on Verma modules over the quantized Coulomb and Higgs branch. An example in \S\ref{sec:higersqed} highlights the differences from the familiar case.\par
The core section of this work is \S\ref{sec:CSquantized}. The quantization of the moduli spaces of vacua of 3d $\mN=4$ Chern--Simons-matter theories is discussed in \S\ref{sec:VermaCS}, which also studies Verma modules over the resulting algebras. A relation between twisted traces and the sphere partition function is analyzed in \S\ref{sec:sphereCS}, with Conjecture \ref{myconj1} being the main outcome. It is a substitute for the Gaiotto--Okazaki conjecture in Chern--Simons-matter theories.\par
An extensive set of examples, including Chern--Simons-matter quivers with non-Abelian gauge groups and quivers not of Dynkin type, is given in \S\ref{sec:ex}. Definitions of the partition function and lengthy computations are collected in the appendix \S\ref{app:SUSYZ}.

\subsubsection*{Acknowledgements}
I am indebted to Dongmin Gang, Marcus Sperling, and Yehao Zhou for sharing their insight. I also thank Antonio Amariti, Marc Besson, Alicia Lamarche, Satoshi Nawata, Tomoki Nosaka, and Zhenghao Zhong for discussions and comments.
This work was partly supported by the Natural Science Foundation of China under grant W2433005 during its early stages.

\section{Preliminaries: 3d \texorpdfstring{$\mN=4$}{N=4} theories}
\label{sec:prelim}

\subsection{Basics on 3d \texorpdfstring{$\mN=4$}{N=4} theories without Chern--Simons couplings}
\label{sec:review3dN4}
This section contains a brief summary of the salient properties of 3d $\mN=4$ quiver theories \emph{without} Chern--Simons couplings, setting the notation for the quantized Coulomb and Higgs branch algebras and their modules.\footnote{An explicit relation between brane setups and slices in the affine Grassmannian has been spelled out in \cite{Grimminger:2020dmg,Bourget:2021siw,Santilli:2021rlf} and subsequent works. Those techniques apply to the present discussion, but are not reviewed here.}

\subsubsection{Quiver gauge theories}
\label{sec:quiverdef}
Let $\sQ=(\sQ_0 \sqcup \sF_0, \sQ_1)$ be a doubled, framed quiver in the sense of Nakajima. $\sQ_0$ denotes the vertex set of the unframed quiver, $\sF_0$ the framing, and $\sQ_1$ includes all edges --- which replace pairwise-opposite arrows. An orientation of the edges is assumed for ease of exposition, and $\mathsf{h}(e)$ (respectively $\mathsf{t}(e)$) denotes the head (respectively tail) of the edge.\par
A representation of $\sQ$ specifies a 3d $\mN=4$ gauge theory: 
\begin{itemize}
	\item The vertices determine the gauge algebra 
	\begin{equation}\label{eq:UNgauge}
		\mathfrak{g}_{\text{\tiny gauge}} = \bigoplus_{v\in \sQ_0} \mathfrak{u}(N_v) ,
	\end{equation}
	whose Cartan subalgebra is denoted $\mathfrak{t}_{\text{\tiny gauge}} \cong \R^{\sum_{v \in \sQ_0}N_v}$.
	\item The edges specify hypermultiplets transforming in the bi-fundamental representation of 
	\begin{equation}
		\mathfrak{u}(N_{\mathsf{h}(e)}) \oplus \mathfrak{u}(N_{\mathsf{t}(e)}) , \qquad e\in \sQ_1 ,
	\end{equation}
	including edges connecting to a framing node. Loop edges correspond to hypermultiplets in the adjoint representation.
\end{itemize}
The drawing conventions are summarized in Table \ref{tab:quiver}. The following is assumed throughout:
\begin{hypt}\label{hyp:goodQ} For every $v \in \sQ_0$ with vanishing Chern--Simons level, 
\begin{equation}
		\sum_{ e \in \sQ_1 \ : \ v = \mathsf{h} (e) } N_{ \mathsf{t} (e)}  + \sum_{ e \in \sQ_1 \ : \ v= \mathsf{t} (e) } N_{ \mathsf{h} (e)} \ge 2 N_v -1  .
	\end{equation}
\end{hypt}

\begin{table}
\centering
\begin{tabular}{|c|c|}
	\hline
		symbol & meaning \\
		\hline
		$\gnode{N}$ & gauge $U(N)$ \\
		\hline
		$\csnode{N}{\kappa}$ & gauge $U(N)_{\kappa}$ \\
		\hline
		$\fnode{k}$ & flavor $\mathfrak{su}(k)$ \\
		\hline
		--- & bi-fundamental hyper \\
		\hline
	\end{tabular}
\caption{Notation for drawing quiver gauge theories.}
\label{tab:quiver}
\end{table}\par

\subsubsection{Quantized Coulomb and Higgs branch algebras}
\label{sec:QuantumCH}
Consider a 3d $\mN=4$ quiver gauge theory $\sQ$. Its Coulomb and Higgs branch are denoted, respectively, 
\begin{equation}
	\mC ( \sQ ) , \qquad \mH ( \sQ ) .
\end{equation}
There exist flat families over $\C[\hbar]$ of non-commutative algebras $\mA_{\hbar}^{\mC( \sQ )}$ and $\mA_{\hbar}^{\mH( \sQ )}$ yielding a deformation quantization of these moduli spaces \cite{Braverman:2016wma}, 
\begin{equation}
	\mA_{\hbar}^{\mC( \sQ )} / (\hbar) \cong \C \left[ \mC ( \sQ )  \right] , \qquad \mA_{\hbar}^{\mH( \sQ )} / (\hbar) \cong \C \left[ \mH ( \sQ )  \right] .
\end{equation}\par 

\subsubsection{Parameters}
Three-dimensional quantum field theories with $\mN=4$ supersymmetry have an $\mathfrak{su}(2)_{C} \oplus \mathfrak{su}(2)_{H}$ R-symmetry. Fixing a Cartan subalgebra of $\mathfrak{su}(2)_{C} $ selects a complex structure on $\mC ( \sQ ) $; the corresponding complex torus acts on $\mC ( \sQ ) $ by scaling, and induces a $\Z$-grading on $\C \left[ \mC ( \sQ )  \right]$. The analogous statements hold for $\mathfrak{su}(2)_{H} $ and $\mH ( \sQ ) $.\par 
Additionally, the complexification of the topological and flavor symmetries $\mathfrak{g}_C \oplus \mathfrak{g}_H $, with Cartan subalgebra $\mathfrak{t}_C \oplus \mathfrak{t}_H $, are isometries of the two branches. There are Hamiltonian actions $\ct_{C} \curvearrowright \mC ( \sQ ) $ and $ \ct_{H} \curvearrowright \mH ( \sQ ) $ of the maximal tori $\ct_{C} = (\C^{\ast})^{\dim \mathfrak{t}_C}$ and $\ct_{H} = (\C^{\ast})^{\dim \mathfrak{t}_H}$. More explicitly, with gauge algebra \eqref{eq:UNgauge}, 
\begin{subequations}
\begin{align}
	\ct_C &= \mathrm{Hom} \left( \pi_{1} \left( \prod_{v \in \sQ_0} GL_{N_v} (\C) \right) , \C^{\ast} \right) \cong \left( \C^{\ast} \right)^{\lvert \sQ_0\rvert} , \\
	\ct_H &= \prod_{e \in \sQ_1} \left( \C^{\ast} \right)^{N_{\mathsf{h} (e)} N_{\mathsf{t} (e)}} \big/  \prod_{v \in \sQ_0} \left( \C^{\ast} \right)^{N_v} .
\end{align}
\end{subequations}
From now on a Cartan subalgebra of $\mathfrak{su}(2)_{C} \oplus \mathfrak{su}(2)_{H}$ is fixed. The parameters of the theory split into real and complex component in the given complex structure.
\begin{itemize}
	\item The real mass parameters $m \in \mathfrak{t}_H$ are the equivariant parameters for the $\ct_H$-action on $\mH (\sQ)$; and are also identified with K\"ahler parameters of a smoothing 
	\begin{subequations}
	\begin{equation}
		\widetilde{\mC} (\sQ) \longrightarrow\mC (\sQ) .
	\end{equation}
	\item The real FI parameters $\zeta \in \mathfrak{t}_C$ are the equivariant parameters for the $\ct_C$-action on $\mC (\sQ)$; and are also identified with K\"ahler parameters of a smoothing 
	\begin{equation}
		\widetilde{\mH} (\sQ) \longrightarrow \mH (\sQ).
	\end{equation}
	\end{subequations}
\end{itemize}
Hypothesis \ref{hyp:massive} in the introduction therefore requires the gauge theory to have `enough' parameters to fully resolve $\mC (\sQ)$ and $\mH (\sQ)$. It is rephrased here for completeness.
\begin{hypt}\label{hyp:mava}
There exists an open subset of $\mathfrak{t}_C \times \mathfrak{t}_H$ such that, for $(\zeta, m)$ in that subset,
\begin{itemize}
\item $\widetilde{\mC} (\sQ)$ is smooth and the $\ct_C$-fixed locus $\widetilde{\mC} (\sQ)^{\ct_C}$ consists of isolated points; or  
\item $\widetilde{\mH} (\sQ)$ is smooth and the $\ct_H$-fixed locus $\widetilde{\mH} (\sQ)^{\ct_H}$ consists of isolated points.
\end{itemize}
\end{hypt}
The two conditions are equivalent, and state that 
\begin{equation}
	\widetilde{\mC} (\sQ)^{\ct_C} = \left( \widetilde{\mC} (\sQ)^{\ct_C}\right)_{\mathrm{red}} = \left( \widetilde{\mH} (\sQ)^{\ct_H}\right)_{\mathrm{red}}=\widetilde{\mH} (\sQ)^{\ct_H} ,
\end{equation}
all being the disjoint union of finitely many points. The fixed points (also known as massive vacua) will be denoted $p_{\alpha} \in \widetilde{\mC} (\sQ)^{\ct_C}$, and the label $\alpha$ is used throughout for the structures associated with a choice of fixed point.\par 
For each vacuum $p_{\alpha} \in \widetilde{\mC} (\sQ)^{\ct_C}$, there exists a pairing 
\begin{equation}\label{eq:pairmzeta}
\pair{\zeta}{m} := \sum_{i,v} \zeta^{v} {(\keff_{\alpha})_{v}}^{i} m_i
\end{equation}
where $\keff_{\alpha}$ is a matrix of integers, which depends on $\sQ$.\par

\subsubsection{Modules over quantized Coulomb and Higgs branch algebras}
Assume the 3d $\mN=4$ theory $\sQ$ satisfies Hypothesis \ref{hyp:mava}. Let $\widetilde{M}^{(\xi)}$ be either $\widetilde{\mC} (\sQ)$ or $\widetilde{\mH} (\sQ)$, where $\xi=m$ in the former case and $\xi=\zeta$ in the latter, so that $\widetilde{M}^{(0)}\cong M$ is a symplectic singularity. Moreover, let $\ct \curvearrowright \widetilde{M}^{(\xi)}$ be the torus with Hamiltonian action, and $\mA_{\hbar}^{(\xi)}$ the quantization of $\widetilde{M}^{(\xi)}$.
Verma modules over $\mA_{\hbar}^{(\xi)}$ have been studied extensively in \cite{Braden:2014iea,Webster:2014,Losev:2015}.\par

\begin{lem}\label{lem:Vermatofp}
\begin{enumerate}[(i)]
\item\label{Vermatofp1} For generic $\xi$, the Verma modules over $\mA_{\hbar}^{(\xi)}$ are in one-to-one correspondence with the fixed points $p_{\alpha} \in \left[ \widetilde{M}^{(\xi)} \right]^{\ct}$.
\item\label{Vermatofp2} At non-generic values of $\xi$ at which two fixed points $p_{\alpha}, p_{\beta}$ coalesce, the corresponding $\mA_{\hbar}^{(\xi)}$-modules become isomorphic.
\end{enumerate}
\end{lem}
\begin{proof}The statement follows from \cite[Sec.5]{Hilburn:2020aau}. In particular, (\ref{Vermatofp2}) is a consequence of \cite[Lem.5.10]{Hilburn:2020aau}. 
\end{proof}
Therefore, each $p_{\alpha} \in \left[ \widetilde{M}^{(\xi)} \right]^{\ct}$ provides a pair of graded vector spaces, 
\begin{equation}\label{eq:scH}
	\scH^{A}_{\alpha} , \qquad  \scH^{B}_{\alpha} ,
\end{equation}
endowed, respectively, with the structure of a Verma module over $\mA_{\hbar}^{\mC (\sQ)},\mA_{\hbar}^{\mH (\sQ)}$.\par
\begin{rmk}The gauge theory derivation of these Verma modules was put forward in \cite{Bullimore:2016nji}, and the mathematical counterpart of this construction is given in \cite{Hilburn:2020aau}. A review may be found in \cite[Sec.2]{Bullimore:2020jdq}. The notation for the Verma modules is as in \cite{Bullimore:2020jdq}.
\end{rmk}

\subsubsection{Twisted traces on Verma modules}
\label{sec:Verma}
It is possible to define twisted traces on the modules \eqref{eq:scH} \cite{Etingof:2019guc}. 
\begin{defin}[Twisted trace]
\label{def:twtrace}
A trace on an algebra $\mA$ \emph{twisted} by an automorphism $\tws \in \mathrm{Aut}(\mA)$ is a linear map $\tr^{(\tws)} : \mA \longrightarrow \C$ satisfying 
\begin{equation}
	\tr^{(\tws)} (ab) = \tr^{(\tws)} \left( \tws(b) a \right) , \qquad \forall a,b \in \mA .
\end{equation}
Besides, let $\scH$ be a Verma module over $\mA$, with grading operator $J$, and represent $\tws$ by $\tws(a)= \twt^{-1} a \twt$. A \emph{twisted trace} on this module is the map $ \mA \longrightarrow \C$ given by
\begin{equation}
	\tr_{\scH} \left( \twt^{-1} a x^J \right) , \qquad \forall a \in \mA .
\end{equation}
\end{defin}
The twisted traces on \eqref{eq:scH} are introduced presently according to \cite[Eq.(2.36)]{Bullimore:2020jdq}.
\begin{defin}Let $J_C, J_H$ denote the grading operators induced on, respectively, $\mA_{\hbar}^{\mC (\sQ)}$ and $\mA_{\hbar}^{\mH (\sQ)}$ by the Hamiltonian torus action on the underlying symplectic variety.
The untwisted traces of the identity on the Verma modules \eqref{eq:scH} are 
\begin{subequations}
\begin{align}
	\chi_{\alpha} ^{\mC(\sQ)} (x) &:= \tr_{\scH^{A}_{\alpha} } x^{J_C} , \label{eq:chiC}\\
	\chi_{\alpha} ^{\mH(\sQ)} (y) &:= \tr_{\scH^{B}_{\alpha} } y^{J_H} . \label{eq:chiH}
\end{align}
\end{subequations}
Besides, the \emph{twisted traces} of the identity on the Verma modules \eqref{eq:scH} are 
\begin{subequations}
\begin{align}
	\hat{\chi}_{\alpha} ^{\mC(\sQ)} (x) &:= \tr_{\scH^{A}_{\alpha}} (-1)^{R_A} x^{J_C} , \label{eq:tchiC} \\
	\hat{\chi}_{\alpha} ^{\mH(\sQ)} (y) &:= \tr_{\scH^{B}_{\alpha}} (-1)^{R_B} y^{J_H}, \label{eq:tchiH}
\end{align}
\end{subequations}
with twists $(-1)^{R_A}$ and $(-1)^{R_B}$, where $R_A, R_B$ are the $\Z$-grading operators induced by the contracting $\C^{\ast}$-actions on the respective branch.
\end{defin}

\subsection{Example review: SQED}
\label{sec:SQED}
An explicit example is worked out to illustrate how the twisted traces emerge from the partition function. Consider a $U(1)$ gauge theory with $k$ hypermultiplets, 
\begin{equation}\label{eq:SQED}
	\gnode{1} \text{---} \fnode{k} ,
\end{equation}
known as 3d $\mN=4$ SQED. In this example, $\mathfrak{t}_C=\mathfrak{u}(1)$ and $\mathfrak{g}_H=\mathfrak{su}(k)$, hence \eqref{eq:SQED} admits an FI parameter $\zeta$ and $k$ mass parameters, subject to $\sum_{i=1}^{k} m_i =0$. For comparison with \cite{Bullimore:2020jdq}, the parameters are taken in the chamber $\zeta>0$ and $m_1<m_2<\cdots <m_{k}$.\par
The Coulomb branch is 
\begin{equation}
	\mC \left( \gnode{1} \text{---} \fnode{k} \right) = \C^2 / \Z_k ,
\end{equation}
while the Higgs branch is the closure of the minimal nilpotent orbit of $\mathfrak{sl}_{k}$. There exist resolutions
\begin{equation}
	\widetilde{\C^2 / \Z_k} \longrightarrow \mC \left( \gnode{1} \text{---} \fnode{k} \right)  , \qquad T^{\ast} \mathbb{P}^{k-1} \longrightarrow \mH \left( \gnode{1} \text{---} \fnode{k} \right) .
\end{equation}
\begin{itemize}
	\item From the Coulomb branch perspective, the K\"ahler moduli of the $(k-1)$ $\mathbb{P}^1$s in $\widetilde{\C^2 / \Z_k}$ are $m_k, \dots, \sum_{i=2}^{k} m_i$. There is a $\C^{\ast}$-action on $\widetilde{\C^2 / \Z_k}$ that preserves the holomorphic symplectic structure, which admits $k$ isolated fixed points.
	\item From the Higgs branch perspective, the K\"ahler modulus of $T^{\ast}\mathbb{P}^{k-1}$ is $\zeta$. There is a $(\C^{\ast})^{k-1}$-action that preserves the holomorphic symplectic structure, which admits $k$ isolated fixed points.
\end{itemize}
Either way, \eqref{eq:SQED} satisfies Hypothesis \ref{hyp:mava}; there are $k$ fixed points labeled $\alpha=1, \dots, k$. The pairing \eqref{eq:pairmzeta} reads 
\begin{equation}
	\pair{\zeta}{m}^{\text{\eqref{eq:SQED}}} = \zeta m_i \delta_{\alpha}^{i},
\end{equation}
namely the $(1\times k)$-matrix $\keff_{\alpha}$ has entries $\delta_{\alpha}^{i}$ ($\alpha$ fixed).\par
\medskip
The sphere partition function is 
\begin{equation}
	\mz_{\gnode{1} \text{---} \fnode{k}} = \int_{-\infty}^{\infty} \dd \sigma \frac{e^{\ii 2 \pi \zeta \sigma}}{\prod_{i=1}^{k} \ch (\sigma - m_i)} .
\end{equation}
Closing the integration contour in the upper half-plane picks residues from $k$ towers of poles, located at 
\begin{equation}
	\left\{ \sigma = m_{i} + \ii \left( \frac{1}{2} + n \right ) , \ n\in \N \right\} , \qquad \forall i=1, \dots, k .
\end{equation}
Simplifying, the resulting expression is 
\begin{align}\label{eq:ZSQEDres}
	\mz_{\gnode{1} \text{---} \fnode{k}} & = \sum_{\alpha=1}^{k} e^{\ii 2 \pi \pair{\zeta}{m}^{\text{\eqref{eq:SQED}}} } \left( \prod_{\substack{i=1, \dots, k \\ i \ne \alpha }} \frac{(-\ii)}{\sh (m_{\alpha}-m_i)} \right) \sum_{n=0}^{\infty} (-1)^{kn} x^{n+\frac{1}{2}} 
\end{align}
where $x:=e^{-2\pi \zeta}$. As pointed out in \cite{Gaiotto:2019mmf,Bullimore:2020jdq}, one reads off the twisted trace 
\begin{equation}\label{eq:chiC2Zn}
	\hat{\chi} ^{\C^2/\Z_k} (x) = \sum_{n=0}^{\infty} (-1)^{kn} x^{n+\frac{1}{2}}  = \frac{x^{\frac{1}{2}}}{ 1- (-1)^k x}  
\end{equation}
on the Verma modules over the quantized Coulomb branch algebra. By analogous manipulations on the other contribution, one reads off the twisted trace on Verma modules of the quantized Higgs branch algebra \cite[Eq.(3.34)]{Bullimore:2020jdq}.\par 
From \eqref{eq:ZSQEDres}, each summand has a coefficient $(-\ii)^{k-1} (-1)^{k-\alpha}$. The sign comes from reordering the masses according to the chosen chamber before expanding the trace. This gives $\tilde{\kappa}=2\alpha- k-1 \mod 4$, in agreement with the first-principles derivation in \cite[Eq.(3.8)]{Bullimore:2020jdq}.

\subsection{Basics on 3d \texorpdfstring{$\mN\ge 3$}{N=3} Chern--Simons-matter theories}
\label{sec:review3dCS}

\subsubsection{3d \texorpdfstring{$\mN\ge 3$}{N=3} Chern--Simons-matter quivers}
\label{sec:CSdef}

Chern--Simons-matter theories are specified by a 3d $\mN=4$ quiver theory as in \S\ref{sec:quiverdef} together with additional data, namely the Chern--Simons levels $\{ \kappa_v \}_{v \in \sQ_0}$ for each gauge group.\footnote{Mixed gauge Chern--Simons levels are not considered in this work.}
\begin{defin}
	An $\mN =4$ Chern--Simons-matter theory is characterized by a pair $\Th_{\underline{\kappa}} =\left( \sQ , \underline{\kappa} \right)$, where $\sQ$ is a 3d $\mN=4$ quiver gauge theory, and $\underline{\kappa} \ : \ \sQ_0  \longrightarrow \Z$, subject to the condition in \cite[Eq.(3.56)]{Gaiotto:2008sd}.
\end{defin}
\begin{rmk}
More precisely, one should allow both \emph{hypermultiplets} and \emph{twisted hypermultiplets} \cite{Hosomichi:2008jd}. In this work, hypermultiplets and twisted hypermultiplets appear exactly as in \cite{Li:2023ffx} and \cite{Marino:2025uub}. The difference between the two will not play a significant role in what follows, thus the distinction is taken into account but not explicitly mentioned throughout. 
\end{rmk}
Several examples in \S\ref{sec:ex} are quiver Chern--Simons theories based on A-type Dynkin diagrams. In this case, a sufficient condition to have $\mN=4$ supersymmetry is that $\sF_0=\emptyset$, $\underline{\kappa} \ne 0$, and the allowed Chern--Simons levels are, schematically, 
\begin{equation}
	\underline{\kappa} = \pm (0, \dots, 0 , \kappa , 0, \dots, 0 , - \kappa , 0, \dots, 0 , \kappa, \cdots ) .
\end{equation}
The general results of \S\ref{sec:CSquantized} hold regardless of the A-type hypothesis, as showcased in \S\ref{sec:nonlin1}-\S\ref{sec:trinion}.\par
\medskip
Moduli spaces of vacua of 3d $\mN=4$ Chern--Simons-matter theories have been analyzed early on in \cite{Hosomichi:2008jd,Imamura:2008nn,Jafferis:2008em}, and using Type IIB string theory in \cite{Assel:2017eun} and more recently in \cite{Marino:2025uub,Marino:2025ihk}.\par
By the $\mN=4$ supersymmetry, $\Th_{\underline{\kappa}}$ carries two distinguished holomorphic-symplectic varieties, together with isometries acting on them. In this work, $\mM_A$ and $\mM_B$ are treated as moduli stacks, with coarse moduli spaces $M_A$ and $M_B$ respectively. The Hamiltonian actions of the maximal tori are denoted $\ct_A \curvearrowright \mM_A$ and $\ct_B \curvearrowright \mM_B$.\par
The Cartan subalgebra of the global symmetries is isomorphic to $\mathfrak{t}_A \oplus \mathfrak{t}_B$; the equivariant parameters for the torus action on the moduli spaces are $\mathbb{Q}$-linear combinations of masses and FI parameters of the gauge theory. Crucially, the Chern--Simons-matter theories $\Th_{\underline{\kappa}}$ generically \emph{do not} satisfy Hypothesis \ref{hyp:mava}: The parameters in the gauge theory are typically not enough to smoothen the singularities of the algebraic varieties $M_A$ and $M_B$.

\subsubsection{Chern--Simons-matter quivers from 5-branes}
In Type IIB string theory, it is possible to realize 3d Chern--Simons-matter theories on the world-volume of D3-branes suspended between $(p,q)$-branes \cite{Kitao:1998mf}. 
To guarantee $\mN=4$ supersymmetry, only configurations with D3-branes suspended between NS5-branes and $(1,\kappa)$-branes, for fixed $\kappa \in \Z$, are considered.
Alternating sequences of NS5-branes and ($1,\kappa$)-branes give rise to linear quivers with Chern--Simons levels
\begin{tenumerate}
	\item $0$ if the D3-brane segments are suspended between branes of the same type;
	\item $+\kappa$ if the D3-brane segments end on a NS5-brane on the left and on a ($1,\kappa$)-brane on the right;
	\item $-\kappa$ if the D3-brane segments end on a ($1,\kappa$)-brane on the left and on a NS5-brane on the right.
\end{tenumerate}
Two drawings in Figure \ref{fig:samplebrane} exemplify the conventions.\par

\begin{figure}[htb]
\centering
\begin{tikzpicture}
	\draw[red,thick] (-4,0) -- (-4,2);
	\draw[red,thick] (-2,0) -- (-2,2);
	\node at (-3,1.5) {$\begin{color}{blue}\underbrace{ \bullet \cdots \bullet}_{k \text{ D5}}\end{color}$};
	\draw (-4,0.95) -- (-2,0.95);
	\draw (-4,0.85) -- (-2,0.85);
	\draw (-4,0.75) -- (-2,0.75);
	\node[anchor=south] at (-4,2) {\scriptsize NS5};
	\node[anchor=south] at (-2,2) {\scriptsize NS5};
	\node[anchor=north] at (-3,0.75) {\scriptsize $N$ D3};

	\draw[red,thick] (2,0) -- (2,2);
	\draw[red,thick] (6,0) -- (6,2);
	\draw[Violet,dashed,thick] (3.5,1.5) -- (2.5,0.5);
	\draw[Violet,dashed,thick] (5.5,1.5) -- (4.5,0.5);
	\draw (2,1.1) -- (6,1.1);
	\draw (2,1) -- (6,1);
	\draw (2,0.9) -- (6,0.9);
	\node[anchor=south] at (2,2) {\scriptsize NS5};
	\node[anchor=south] at (6,2) {\scriptsize NS5};
	\node[anchor=north] at (4,0.9) {\scriptsize $N$ D3};
	\node[anchor=south] at (3.5,1.5) {\scriptsize ($1,\kappa$)};
	\node[anchor=south] at (5.5,1.5) {\scriptsize ($1,\kappa$)};
	
	\node (q1) at (-3,-0.5) {$\gnode{N}\text{---}\fnode{k}$};	
	\node (c1) at (4,-0.5) {$\csnode{N}{\kappa}$-----$\csnode{N}{0}$-----$\csnode{N}{-\kappa}$};
\end{tikzpicture}
\caption{Brane realization of 3d $\mN=4$ theories, and the corresponding quiver, without (left) and with (right) Chern--Simons couplings.}
\label{fig:samplebrane}
\end{figure}
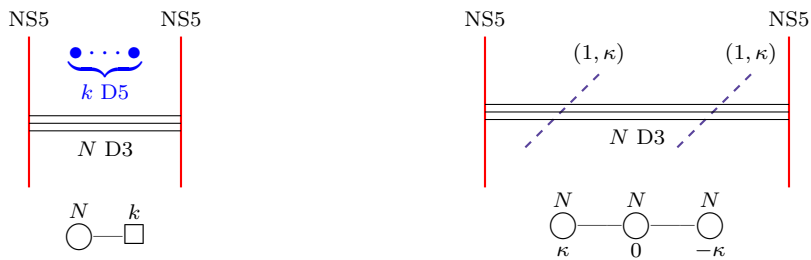

\subsubsection{Magnetic quivers}
\label{sec:QAQB}
Let $\Th_{\underline{\kappa}} = (\sQ, \underline{\kappa} )$ where $\sQ$ is an A-type Dynkin diagram, and denote by $M_A$ (respectively $M_B$) the coarse moduli space of $\mM_A$ (respectively $\mM_B$).\par
The recent work \cite{Marino:2025uub} proposes a way to characterize $M_A,M_B$ by comparison with ordinary Coulomb branches. The authors of \cite{Marino:2025uub} provide a constructive prescription to write down two `magnetic' quivers $\sQ_A,\sQ_B$ giving rise to 3d $\mN=4$ theories without Chern--Simons couplings, and conjecture the isomorphisms of singular varieties
\begin{equation}\label{eq:MAQAMBQB}
	M_A \cong \mC( \sQ_A ) , \qquad M_B \cong \mC( \sQ_B ) .
\end{equation}
The algorithm, inspired by brane constructions, is based on a sequence of local replacements along the quiver. The prescription is summarized in Table \ref{tab:MSAb} for the especially simple case of Abelian A-type quivers. More details can be found in \cite[Sec.3]{Marino:2025uub}.\par
All the results of this work are independent of those in \cite{Marino:2025uub}; the latter will be used as a cross-check throughout the examples.\par

\begin{table}
\centering
\begin{tabular}{|c|c|}
	\hline
		local quiver $\Th_{\underline{\kappa}}$ & local magnetic quiver $\sQ_A$ \\
		\hline
		$\cdots \underbrace{\csnode{1}{\kappa} \text{---}\csnode{1}{0} \text{---}\cdots  \text{---}\csnode{1}{0} \text{---} \csnode{1}{-\kappa}}_{k+1} \cdots $ & \begin{tikzpicture}[baseline=-1] \node (g) at (0,0) {$\cdots \gnode{1} \cdots $}; \node (fl) at (-1.25,-0.5) {$\fnode{\kappa}$};  \node (fc) at (0,-0.55) {$\cdots$}; \node (fr) at (1.25,-0.5) {$\fnode{\kappa}$}; \draw (-1.05,-0.5) -- (-0.25,-0.25);\draw (1.05,-0.5) -- (0.25,-0.25); \node[anchor=north] (u) at (0,-0.65) {$\underbrace{\hspace{2.75cm}}_{k}$};\end{tikzpicture}\\
		\hline
		$\csnode{1}{-\kappa} \text{---} \cdots $ & $\fnode{1} \text{---} \cdots$ \\
		\hline
		$ \cdots \text{---}\csnode{1}{\kappa} $ & $ \cdots \text{---} \fnode{1}$ \\
		\hline
	\end{tabular}
\caption{Recipe for the magnetic quivers $\sQ_A$ of Abelian A-type Chern--Simons-matter theories.}
\label{tab:MSAb}
\end{table}\par

\section{Quantized hypertoric stacks and twisted traces}
\label{sec:highercharge}

The present section is devoted to 3d $\mN=4$ quiver gauge theories \emph{without} Chern--Simons couplings, but carrying representations of non-minimal highest weight. This is a situation in which Hypothesis \ref{hyp:mava} is not satisfied, yet is more tractable than Chern--Simons-matter theories. Thus, this section is a first step toward the generalization of the sphere quantization prescription; it will be instrumental later but is also of independent interest.\par
The main results of this section are: 
\begin{itemize}
	\item Theorem \ref{thm:quantumhyper}: Quantization of the Coulomb and Higgs hypertoric stacks;
	\item Corollary \ref{cor:hypertwtr}: Twisted traces on modules over the quantized algebras are extracted from the sphere partition function.
\end{itemize}

\subsection{Quantized hypertoric stacks}
\label{sec:higherquant}

The focus of this section is on gauge theories whose Coulomb and Higgs branches are hypertoric varieties \cite{Braden:2008kny}. These are Abelian theories, $N_v =1 $ $\forall v\in \sQ_0$. For ease of exposition, framing nodes are also taken to have rank one, without loss of generality since an arbitrary number of framing nodes can be connected to each $v \in \sQ_0$.\par
\begin{defin}\label{def:Qprime}
An Abelian quiver gauge theory with arbitrary charges consists of the pair $\sQ^{\prime} =\left( \sQ, \kk \right)$, where $\sQ$ is as in \S\ref{sec:quiverdef}, and $\kk \in \mathrm{Hom} \left( U(1)^{\sQ_0}, U(1)^{\sQ_1}\right)$ is represented by a matrix $\kk= \left[ {\kk_{e}}^{v}\right]_{\substack{e \in \sQ_1 \\ v \in \sQ_0}}$, with 
\begin{equation}
	{\kk_{e}}^{v} \in \Z, \qquad {\kk_{e}}^{v} =0 \ \text{ if } \ v \notin \{ \mathsf{h}(e), \mathsf{t}(e) \}.
\end{equation}
$\sQ^{\prime}$ is said to have \emph{non-minimal charges} if some $\lvert {\kk_{e}}^{v}\rvert >1$. The notation 
\begin{equation}
	\kappa:=\lvert \mathrm{gcd} \left( \{ {\kk_{e}}^{v} \} \right)\rvert
\end{equation}
is adopted.
\end{defin}
Without loss of generality, it is henceforth assumed that $\sQ$ is connected and $\sF_0 \ne \emptyset$. 
\begin{lem}\label{lem:rankminorK} $\mathrm{rk} (\kk) =\lvert \sQ_0\rvert$.
\end{lem}
\begin{proof}
With the above assumptions, $\lvert \sQ_1 \rvert \ge \lvert \sQ_0 \rvert$. The submatrix of $\kk$ including only rows corresponding to the edges with $\mathsf{h}(e), \mathsf{t}(e) \in \sQ_0$ has rank $\lvert \sQ_0\rvert -1$ by the connectedness hypothesis. Adding back the edges with $\mathsf{h}(e)\in\sF_0$ introduces the rows with ${\kk_{e}}^{v} \ne 0$ at $v=\mathsf{t}(e)$ and $0$ elsewhere; likewise, the edges with $ \mathsf{t}(e) \in \sF_0$ introduce rows with only one non-zero entry ${\kk_{e}}^{v} \ne 0$ at $v=\mathsf{h}(e)$. These are linearly independent of the previous rows and increase the rank by $1$.
\end{proof}
\begin{rmk}The $\lvert \sQ_0\rvert \times\lvert \sQ_0\rvert$ minors of $\kk$ with non-vanishing determinant are invertible in $GL_{\lvert \sQ_0\rvert} (\mathbb{Q})$, but not necessarily in $SL_{\lvert \sQ_0\rvert} (\Z)$.
\end{rmk}\par
If $\sQ^{\prime}$ has non-minimal charges with $\kappa >1$, it is possible to define an Abelian theory with reduced charges, $\sQred := \left( \sQ, \frac{1}{\kappa}\kk \right)$, in which the entires of $\kk$ have been rescaled.\par
\begin{itemize}
	\item It has long been understood that the appropriate language for Higgs branches of gauge theories with non-minimal charges is that of gerbes \cite{Pantev:2005zs}. 
	\begin{prop}\label{prop:HBgerbe}
		$\mH \left( \sQ^{\prime} \right)$ is a $\Z_{\kappa}$-gerbe over $\mH \left( \sQred \right)$. 
	\end{prop}
	\begin{proof}[Sketch of proof]
	$\mH \left( \sQ^{\prime} \right)$ is a toric symplectic quotient $\C^{2 \lvert \sQ_1 \rvert} /\!\!/\!\!/_{\zeta} (\C^{\ast})^{\lvert \sQ_0\rvert}$ where the affine coordinates have weight divisible by $\kappa$ under the toric action.
	\end{proof}
	\item For the Coulomb branch:
	\begin{prop}\label{prop:Cisquot}
		$\mC \left( \sQ^{\prime} \right)$ is a quotient stack with coarse moduli space $\mC \left( \sQred \right)/\Gamma_{\kappa}$, for a finite group $\Gamma_{\kappa}$ of order $\kappa^{\lvert \sQ_0 \rvert}$.
	\end{prop}
	The coarse moduli space claim was observed in \cite{Nawata:2023rdx}, and established in \cite{Hanany:2023uzn}. A direct proof is given below for completeness, along the lines of \cite{Nawata:2023rdx}. The sole novelty here is to emphasize that $\mC \left( \sQ^{\prime} \right)$ should be thought of as a stack. 
\end{itemize}
It is possible to quantize these hypertoric stacks, and will be done in \S\ref{sec:quantizedhypertoric}.\par

\subsubsection{Gauging a subgroup of the topological symmetry}
One way to produce a quiver $\sQ^{\prime}$ with non-minimal charges is to gauge a finite Abelian subgroup of the Coulomb branch isometry of $\sQred$ \cite{Nawata:2023rdx,Hanany:2023uzn}.\par
Consider an Abelian quiver gauge theory $\sQ^{\prime \prime}=(\sQ, \lambda)$. Let $\sV\subset \sQ_0$ be defined by the property that the quiver with all $v \in \sV$ erased has charges multiple of $\kappa$, but the greatest common divisor is 1 if any $v \in \sV$ is added back in. That is, $\exists \kappa > 1$ such that
\begin{equation}
	\lvert \mathrm{gcd}_{\substack{ e \in \sQ_1 \\ v \in \sQ_0 \setminus \sU}} \left( \{ {\lambda_{e}}^{v} \} \right)\rvert = \begin{cases}  \kappa  , & \sV \subseteq \sU , \\ 1, & \sU \subsetneq \sV .\end{cases}
\end{equation}
There is a finite Abelian subgroup $\Gamma_{\kappa} := \prod_{v \in \sV} \Z_{\kappa} \subset \prod_{v \in \sV} \C^{\ast}_v \subseteq \ct_C$ acting on $\mC \left( \sQ^{\prime \prime} \right)$. To gauge this symmetry corresponds to project onto the invariant subalgebra $\C [\mC \left( \sQ^{\prime \prime} \right)]^{\Gamma_{\kappa}}$, whose spectrum is the quotient stack $\left[ \mC \left( \sQ^{\prime \prime} \right) / \Gamma_{\kappa} \right]$. Besides, the net effect of the orbifold is to rescale ${\lambda_e}^{v} \mapsto {\kk_e}^{v}=\kappa {\lambda_e}^{v}$ for all $v \in \sV$ \cite{Nawata:2023rdx}.\par
In conclusion, this outputs a new quiver $\sQ^{\prime}$ with non-minimal charges as claimed.
\begin{proof}[Proof of {Proposition \ref{prop:Cisquot}}]
	Given $\sQ^{\prime}$, it is always possible to define $\sQred$, and apply the above argument with $\sQ^{\prime \prime}=\sQred$ and $\sV=\sQ_0$.
\end{proof}

\subsubsection{Quantized hypertoric stacks}
\label{sec:quantizedhypertoric}
The following technical claim shows that it is possible to quantize a gerbe over a Higgs or Coulomb branch $M$. Lemma \ref{lem:quantumgerbe} is not necessary for the results in the rest of the work but is included for completeness.\par
\begin{lem}\label{lem:quantumgerbe}
	Let $M$ be a smooth affine scheme with $H^3(M, \Z_{\kappa})=0$, and $\mM \longrightarrow M$ a $\Z_{\kappa}$-gerbe over $M$. Assume in addition that there exists a deformation quantization $\mA_{\hbar}$ of $M$, flat over $\C [\hbar]$. Then, there exists a family of non-commutative rings $\mA_{\hbar}^{\prime}$ over $\C [\hbar]$ quantizing $\mM$.
\end{lem}
\begin{proof}
This proof requires working in \'etale topology. The subsequent statements on \'etale cohomology and Azumaya algebras are classical, and can be found in the book \cite{Milne:1980}.\par
The gerbe $\mM$ is characterized by a class $\beta \in H^{2}_{\text{\rm \'et}} \left( M  , \underline{\Z}_{\kappa}\right)$. On the one hand, \cite[Thm.XI.4.4]{SGA4} gives 
\begin{equation}\label{eq:SGA4H}
	H^{\bullet}_{\text{\rm \'et}} \left( M  , \underline{\Z}_{\kappa}\right) \cong H^{\bullet} \left( M, \Z_{\kappa} \right) .
\end{equation}
On the other hand, Kummer's short exact sequence of sheaves (in \'etale topology) 
\begin{equation}
	1 \longrightarrow \underline{\Z}_{\kappa} \longrightarrow  \mO_M^{\ast} \longrightarrow  \mO_M^{\ast}  \longrightarrow  1 
\end{equation}
induces a map $\delta : H^{2}_{\text{\rm \'et}} \left( M  , \underline{\Z}_{\kappa}\right) \longrightarrow H^{2}_{\text{\rm \'et}} \left( M  , \mO_M^{\ast}\right)$,
with the image given by the $\kappa$-torsion subgroup of $H^{2}_{\text{\rm \'et}} \left( M  , \mO_M^\ast\right)_{\mathrm{tor}}$. In this way, $\delta (\beta)$ determines a Morita-equivalence class of sheaves $\scA$ of Azumaya algebras over $M$. Fix an open cover $\{U_i \}$ of $M$ such that, locally, $\scA\vert_{U_i} \cong \mathrm{Mat}_{\kappa \times \kappa} \left( \mathcal{O}_{U_i} \right)$.\par
Consider next the deformation quantization. For every fixed $\hbar$ it is possible to define a sheaf of algebras $\scD_{\hbar}$ over $M$ such that $\Gamma \left(M, \scD_{\hbar} \right)= \mA_{\hbar}$. Explicitly, $\scD_{\hbar} (U) $ is the vector space $\mathcal{O}_X (U)[ \hbar]$ with multiplication induced by $\mA_{\hbar}$ on the ring of functions on the open $U \subset M$.\par
Locally it is possible to take $\scA_{\hbar} (U) =  \mathrm{Mat}_{\kappa \times \kappa} \left( \scD_{\hbar} (U) \right)$. The obstruction to the global existence is controlled by $H^{3}_{\text{\rm \'et}} \left( M  , \underline{\Z}_{\kappa}\right)$, which, by virtue of \eqref{eq:SGA4H} and the assumption on $M$, vanishes. One thus constructs a sheaf of non-commutative algebras $\scA_{\hbar}$.\par 
To conclude the proof, one takes sections and defines $\mA_{\hbar}^{\prime} := \Gamma \left( M, \scA_{\hbar} \right)$. It satisfies: $\mA_{\hbar}^{\prime}/ (\hbar) \cong \Gamma \left( M, \scA \right)$ by construction; $\mA_{\hbar}^{\prime} \cong \mA_{\hbar}$ if $\kappa=1$; and $\mA_{\hbar}^{\prime} $ is Morita-equivalent to $\mA_{\hbar}$ if $\beta=0$.
\end{proof}

The first new observation of the present paper is the quantization of the Higgs and Coulomb branches of theories with non-minimal charges.\par
\begin{thm}\label{thm:quantumhyper}
Let $\sQ^{\prime}$ describe an Abelian gauge theory with charges with greatest common divisor $\kappa$, and $\sQred$ describe the Abelian theory with analogous field content and charges divided by $\kappa$. Assume there exists an open region of parameter space such that the resolution $\widetilde{\mH} \left( \sQred \right) \longrightarrow \mH \left( \sQred \right)$ is smooth. Then:
\begin{enumerate}[(i)]
	\item\label{quanthyperC} There exists a family of non-commutative rings $\mA_{\hbar}^{\mC \left( \sQ^\prime \right)}$ flat over $\C [\hbar]$, quantizing the Coulomb branch $\mC \left( \sQ^{\prime} \right)$;
	\item\label{quanthyperH} There exists a family of non-commutative rings $\mA_{\hbar}^{\mH \left( \sQ^\prime \right)}$ flat over $\C [\hbar]$, quantizing the Higgs branch $\mH \left( \sQ^{\prime} \right)$.
\end{enumerate}
\end{thm}
\begin{proof}
(\ref{quanthyperC}) The Coulomb statement is essentially contained in \cite{Braverman:2016wma}.\par 
By construction, Proposition \ref{prop:Cisquot} gives $\C \left[ \mC \left( \sQ^{\prime} \right)\right] = \C \left[ \mC \left( \sQred\right)\right]^{\Gamma_{\kappa}}$. To quantize it, one uses that the authors of \cite{Braverman:2016wma} construct $\mA_{\hbar}^{\mC  \left( \sQred \right)} = H_{\bullet}^{\mathrm{BM}} \left( \mathcal{R} \right)^{G_{\C} [\![ z]\!] \rtimes \C^{\ast}}$ as a certain equivariant Borel--Moore homology, where $G_{\C} [\![ z]\!]$ is the complexification of the gauge group with entries formal power series (see \cite{Braverman:2016wma} for the details and for the definition of $\mathcal{R}$), and the $\C^{\ast}$-action is as discussed in \S\ref{sec:QuantumCH}.\par
Gauging $\Gamma_{\kappa}$ modifies the action of the gauge group. Denoting by $\hat{G}_{\C} \curvearrowright \mathcal{R}$ the new action, it is possible to define 
\begin{equation}
	\mA_{\hbar}^{\mC \left( \sQ^\prime \right)} := H_{\bullet}^{\mathrm{BM}} \left( \mathcal{R} \right)^{\hat{G}_{\C} [\![ z]\!] \rtimes \C^{\ast}} .
\end{equation}
This is true for any gauge group and representation of cotangent type \cite{Braverman:2016wma}, and is applied here to $G_{\C}\cong  (\C^{\ast})^{\lvert \sQ_0\rvert}$, and $\hat{G}_{\C}\cong G_{\C}$ acting on $\mathcal{R}$ with weights rescaled by $\kappa$.\par
\medskip
(\ref{quanthyperH}) The Higgs branch statement is an application of Lemma \ref{lem:quantumgerbe} to $M=\widetilde{\mH} \left( \sQred \right)$.
\end{proof}

\subsubsection{Torus fixed points on hypertoric stacks}
Let $\sQ^{\prime} =\left( \sQ, \kk \right)$ be a quiver with non-minimal charges, $\kappa>1$, and $\sQred$ be as above. 
\begin{prop}\label{prop:stack123}
\begin{tenumerate}
	\item\label{stack1} If ${\kk_{e}}^{v}= \kappa \left( \delta^{v}_{\mathsf{h}(e)}+\delta_{\mathsf{t}(e)}^{v} \right)$, then $\sQred$ satisfies Hypothesis \ref{hyp:mava}.
	\item\label{stack2} The torus fixed points $p_{\alpha} \in \widetilde{\mC} \left( \sQred \right)^{\ct_C}$ are in one-to-one correspondence with immersion maps 
	\begin{equation}\label{eq:alphaembed}
		\alpha : \sQ_0 \hookrightarrow \sQ_1 \quad \text{ such that } \quad \det_{v,w\in\sQ_0} ({\kk_{\alpha(w)}}^v) \ne 0 .
	\end{equation}
	\item\label{stack3} The quotient map sends each $p_{\alpha} \in \widetilde{\mC} \left( \sQred \right)^{\ct_C}$ to a fat point in $\mC \left( \sQ^\prime \right)^{\ct_C}$.
\end{tenumerate}
\end{prop}
\begin{proof}
Properties (\ref{stack1}) and (\ref{stack2}) are known, and have been used for example in \cite[Sec.6]{Bullimore:2016nji} and \cite[Sec.3]{Gaiotto:2023kuc}. They can be proven in the formalism of hyperplane arrangements \cite{Braden:2008kny}.\par
A self-contained proof of (\ref{stack2}) is given. From (\ref{stack1}), $p_{\alpha} \in \widetilde{\mC} \left( \sQred \right)^{\ct_C}$ is isolated. It is possible to choose local coordinates centered at $p_{\alpha}$ such that $\ct_C$ acts linearly on them. Each such change of variables is in one-to-one correspondence with an invertible $\lvert \sQ_0\rvert \times \lvert \sQ_0\rvert$ minor of $\kk$ and, by Lemma \ref{lem:rankminorK}, such minors exist. Finally, $\lvert \sQ_0\rvert \times \lvert \sQ_0\rvert$ minors of a $\lvert \sQ_1\rvert \times \lvert \sQ_0\rvert$ matrix are in one-to-one correspondence with subsets of $\sQ_1$ of size $\lvert \sQ_0\rvert$. This shows \eqref{eq:alphaembed}.\par
(\ref{stack3}) uses Proposition \ref{prop:Cisquot}. Because $\Gamma_{\kappa} \subset \ct_C$, then $\left[  \widetilde{\mC} \left( \sQred \right)^{\ct_C} \right]^{\gamma} = \widetilde{\mC} \left( \sQred \right)^{\ct_C}$ for all $\gamma \in\Gamma_{\kappa}$. Therefore, each torus fixed points $p_{\alpha} \in \widetilde{\mC} \left( \sQred \right)^{\ct_C}$ is fixed by the whole $\Gamma_{\kappa}$ and descends with multiplicity to the quotient stack $\mC \left( \sQ^\prime \right)^{\ct_C}$.
\end{proof}

\subsection{Twisted traces and sphere quantization}
\label{sec:HigherChargeSphere}
Let $\sQ^{\prime}$ be as in Definition \ref{def:Qprime}, and introduce FI parameters $\zeta=\{\zeta^{v}\}_{v\in\sQ_0}$ and redundant mass parameters $m=\{m_e\}_{e\in \sQ_1}$. 
For each $p_{\alpha} \in \widetilde{\mC} \left( \sQred \right)^{\ct_C}$, let 
\begin{equation}
	\pair{\zeta}{m} = \sum_{v,w \in \sQ_0} \zeta^{v} {(\kk^{-1})_v}^{\alpha(w)} m_{\alpha(w)} ,
\end{equation}
or, in the notation of \eqref{eq:pairmzeta}, ${(\keff_{\alpha})_{v}}^{i} = \sum_{w\in\sQ_0} \delta^{i}_{\alpha (w)} {(\kk^{-1})_v}^{\alpha(w)}$. Here, with a slight abuse of notation, ${(\kk^{-1})_v}^{\alpha(w)}$ indicate the entries of the inverse of the $\lvert \sQ_0\rvert \times \lvert \sQ_0\rvert$ minor of $\kk$ determined by the map $\alpha$ in \eqref{eq:alphaembed}. For later convenience, define also
\begin{equation}\label{eq:hatkappachargek}
	{(\hat{\kappa}_{\alpha})_{e}}^{v} :=\sum_{w\in\sQ_0} {\kk_e}^{w} {(\kk^{-1})_w}^{\alpha(v)} .
\end{equation}\par
The partition function of $\sQ^{\prime}$ is
\begin{equation}\label{eq:ZQpgAb}
	\mz_{\sQ^{\prime}} (\zeta, m) = \int_{\R^{\lvert\sQ_0\rvert}} \prod_{v\in\sQ_0} \dd \sigma_v e^{\ii 2\pi \zeta^{v}\sigma_v} \prod_{e\in\sQ_1} \ch \left( \sum_{v\in\sQ_0}{\kk_e}^{v}  \sigma_v  -m_e\right)^{-1}  .
\end{equation}
\begin{thm}
It holds that
\begin{equation}
\label{eq:ZQprime}
\begin{aligned}
	\mz_{\sQ^{\prime}} (\zeta, m) = &\sum_{ p_{\alpha} \in\widetilde{\mC} \left( \sQred \right)^{\ct_C}}  \frac{e^{\ii 2 \pi \pair{\zeta}{m}}}{\underset{v,w\in\sQ_0}{\det} ({\kk_{\alpha(w)}}^v) } e^{-\ii \frac{\pi}{4}\tilde{\kappa}_{\alpha}} \sum_{n\in \N^{\sQ_0}}\sum_{\nu \in \N^{\sQ_1 \setminus \alpha (\sQ_0)}} e^{ -\ii 2 \pi \hat{\kappa}_{\alpha} (n , \nu) } \\
	\times & \left[ \prod_{v\in \sQ_0} e^{- \ii \pi (\kappa^C_{\alpha})^{v} n_v} (x_{v;\alpha})^{\frac{1}{2} +n_{v}} \right] ~\left[ \prod_{e\in \sQ_1 \setminus \alpha (\sQ_0)} e^{- \ii \pi  (\kappa^H_{\alpha})_{e} \nu^e}  (y_{e;\alpha})^{\frac{1}{2}+\nu^e} \right] 
\end{aligned}
\end{equation}
where $\hat{\kappa}_{\alpha} (n , \nu) := \sum_{v\in\sQ_0 , e\in \sQ_1 \setminus \alpha (\sQ_0)} {(\hat{\kappa}_{\alpha})_{e}}^{v} n_{v}\nu^{e}$, 
\begin{subequations}\label{eq:kappacoefhigher}
\begin{align}
	\tilde{\kappa}_{\alpha}& := \sum_{v\in\sQ_0}\sum_{e\in \sQ_1 \setminus \alpha (\sQ_0)} {(\hat{\kappa}_{\alpha})_{e}}^{v} , \label{eq:kappacoefhighera}\\ 
	(\kappa^C_{\alpha})^{v} & := 1+\sum_{e\in \sQ_1 \setminus \alpha (\sQ_0)} {(\hat{\kappa}_{\alpha})_{e}}^{v} , \qquad (\kappa^H_{\alpha})_{e}:= 1+\sum_{v\in \sQ_0} {(\hat{\kappa}_{\alpha})_{e}}^{v} , \label{eq:kappacoefhigherb}
\end{align}
\end{subequations}
$x_{v;\alpha}$ are exponential functions of $\zeta \in \R^{\sQ_0}$ and $y_{e;\alpha}$ are exponential functions of $m \in \R^{\sQ_1}$. Explicitly, for $\zeta \in (0,\infty)^{\lvert \sQ_0\rvert}$ and generic $m$ they are given by
\begin{subequations}\label{eq:xyhigher}
\begin{align}
	x_{v;\alpha}&:= \prod_{w\in \sQ_0}e^{-2\pi \zeta^{w} {(\kk^{-1})_w}^{\alpha(v)}}, \quad & \forall v \in \sQ_0 , \label{eq:xva}\\
	y_{e;\alpha}&:= \left( e^{2\pi m_e} \prod_{v\in \sQ_0} e^{-2\pi {(\hat{\kappa}_{\alpha})_{e}}^{v}  m_{\alpha(v)}} \right)^{\mathrm{sgn}_{e;\alpha}} , \quad &\forall e \in \sQ_1 \setminus \alpha (\sQ_0), \label{eq:yea}
\end{align}
\end{subequations}
with, in \eqref{eq:yea}, $\mathrm{sgn}_{e;\alpha} \in \{\pm 1 \}$ is such that $\lvert y_{e;\alpha}\rvert <1$.
\end{thm}
\begin{proof}
Fix the chamber $\zeta \in (0,\infty)^{\lvert \sQ_0\rvert} \subset \mathfrak{t}_C$. If there exists a chamber such that $\lvert y_{e;\alpha}\rvert <1$, the proof is given in such a chamber for definiteness; if this is not the case for some pairs $(e;\alpha)$, the final expression is simply obtained by replacing ${(\hat{\kappa}_{\alpha})_{e}}^{v} \mapsto - {(\hat{\kappa}_{\alpha})_{e}}^{v}$ and $m_e \mapsto -m_e$ for the corresponding pair.\par
Closing the integration contour of \eqref{eq:ZQpgAb} in the upper quadrant of $\R^{\lvert \sQ_0\rvert}$ picks the poles at 
\begin{equation}
	\sum_{v\in\sQ_0}{\kk_e}^{v}  \sigma_v  = m_e + \ii \left( \frac{1}{2} + n \right) , \qquad n \in \N ,
\end{equation}
for $e$ in a subset of $\sQ_1$ of size $\lvert \sQ_0\rvert$. Selecting $\lvert \sQ_0\rvert$ edges $e \in \sQ_1$ such that the corresponding minor of $\kk$ can be inverted allows one to solve for $\{\sigma_v\}_{v\in\sQ_0}$, and pick their contribution to \eqref{eq:ZQpgAb}. This leads to the sum over $\alpha$ as in \eqref{eq:alphaembed}.\par 
Computing the residues yields:
\begin{equation}
\begin{aligned}
	\mz_{\sQ^{\prime}} (\zeta, m) = &  \sum_{\alpha \text{ as in \eqref{eq:alphaembed}}} \frac{e^{\ii 2 \pi \pair{\zeta}{m}}}{\underset{v,w\in\sQ_0}{\det} ({\kk_{\alpha(w)}}^v) }\sum_{n\in \N^{\sQ_0}} \prod_{w\in\sQ_0} e^{-2\pi \sum_{v\in \sQ_0} \zeta^{v} {(\kk^{-1})_v}^{\alpha(w)} \left( \frac{1}{2} + n_w \right)} (-1)^{n_w} \\
	&\times \prod_{\substack{e \in \sQ_1\\ e \notin \alpha (\sQ_0)}} \ch \left[ -m_e +\sum_{v,w\in\sQ_0} {\kk_e}^{v} {(\kk^{-1})_v}^{\alpha(w)} \left( m_{\alpha (w)} + \ii \left( \frac{1}{2} + n_w \right)  \right) \right]^{-1} .
\end{aligned}
\end{equation}
Expanding the $1/\ch$ in a geometric series and using the definition \eqref{eq:xyhigher} gives:
\begin{equation}
\begin{aligned}
	\mz_{\sQ^{\prime}} (\zeta, m) &= \sum_{\alpha \text{ as in \eqref{eq:alphaembed}}}  \frac{e^{\ii 2 \pi \pair{\zeta}{m}}}{\underset{v,w\in\sQ_0}{\det} ({\kk_{\alpha(w)}}^v) }  \sum_{n\in \N^{\sQ_0}}\sum_{\nu \in \N^{\sQ_1 \setminus \alpha (\sQ_0)}}  \\
	& \times \prod_{v\in \sQ_0} (-1)^{n_v} (x_{v;\alpha})^{\frac{1}{2} +n_{v}}\prod_{e\in \sQ_1 \setminus \alpha (\sQ_0)} (-1)^{\nu^e}  (y_{e;\alpha})^{\frac{1}{2}+\nu^e}  ~e^{-\ii 2 \pi {(\hat{\kappa}_{\alpha})_{e}}^{v}  \left( \frac{1}{2} +n_{v}\right)\left(\frac{1}{2} +\nu^{e} \right)}
\end{aligned}
\end{equation}
Rearranging the terms independent of $x_{v;\alpha}, y_{e;\alpha}$ in the second line concludes the proof.
\end{proof}
\begin{rmk}The coefficients \eqref{eq:kappacoefhigher} are in general rational.
\end{rmk}\par
\medskip
With $\sQ^{\prime}$ as above, let $\sQ$ be the underlying quiver. For each $p_{\alpha} \in \mC (\sQ)^{\ct_C}$, let $\scH^{A}_{\alpha}$ and $\scH^{B}_{\alpha}$ be the Verma modules as prescribed in \S\ref{sec:Verma}. Define a new operator on $\scH^{A}_{\alpha}\otimes \scH^{B}_{\alpha}$:
\begin{equation}\label{eq:twistop}
	\Delta_{\alpha} := \frac{1}{2} \kappa^C_{\alpha} \left( J_C - \frac{1}{2} \right) + \frac{1}{2} \kappa^H_{\alpha} \left( J_H - \frac{1}{2} \right) + \hat{\kappa}_{\alpha} \left( J_C -\frac{1}{2} \right)\otimes  \left( J_H -\frac{1}{2} \right) . 
\end{equation}
A rewriting of the partition function relates it to twisted traces of modules over the quantized algebras of the underlying quiver.\par
\begin{center}
\noindent\fbox{%
\parbox{0.94\linewidth}{
\begin{cor}\label{cor:Zhighertrace}
	 The sphere partition function of $\sQ^{\prime}$ is a sum over twisted traces:
	\begin{equation}
	\label{eq:hypertoricnotfactor}
		\mz_{\sQ^{\prime}} (\zeta, m) = \sum_{p_{\alpha}\in \left( \mC \left( \sQ^{\prime}\right)^{\ct_C} \right)_{\mathrm{red}} } \frac{e^{\ii 2 \pi \pair{\zeta}{m}}}{\underset{v,w\in\sQ_0}{\det} ({\kk_{\alpha(w)}}^v) } e^{- \ii \frac{\pi}{4} \tilde{\kappa}_{\alpha}} ~\tr_{\scH^{A}_{\alpha}\otimes\scH^{B}_{\alpha}} \left(  e^{-\ii 2 \pi \Delta_{\alpha}} x^{J_C} y^{J_H} \right) ,
	\end{equation}
	with $\Delta_{\alpha}$ as defined in \eqref{eq:twistop} and $\tilde{\kappa}_{\alpha}$ in \eqref{eq:kappacoefhighera}.
\end{cor}
}}\end{center}
\begin{proof}
First, using $ \left( \mC \left( \sQred\right)^{\ct_C} \right)_{\mathrm{red}}= \mC (\sQ)^{\ct_C} = \left( \mC \left( \sQ^{\prime}\right)^{\ct_C} \right)_{\mathrm{red}}$ one replaces the sum over $\alpha$ as in \eqref{eq:alphaembed} with a sum over fixed points on $\mC (\sQ)$, and then one may think of the Verma modules $\scH^{A}_{\alpha}$ and $\scH^{B}_{\alpha}$ as associated with the collection of fixed points $ \left( \mC \left( \sQ^{\prime}\right)^{\ct_C} \right)_{\mathrm{red}}$.\par
The trace is expanded in the tensor product basis 
\begin{equation}
	\left\{ \lvert n \rangle_A\otimes \lvert \nu \rangle_B \  : \ n \in \N^{\sQ_0} ,  \nu \in \N^{\sQ_1  \setminus \alpha (\sQ_0) }  \right\} 
\end{equation}
chosen as in \cite{Bullimore:2020jdq}, where $ \lvert n \rangle_A$ are eigenvectors of $J_C$ with eigenvalues $\left\{n_v + \frac{1}{2} \right\}_{v\in \sQ0}$ and likewise $ \lvert \nu \rangle_B$ are eigenvectors of $J_H$. Using 
\begin{equation}
	\Delta_{\alpha} \lvert n \rangle_A\otimes \lvert \nu \rangle_B  = \frac{1}{2} \sum_{v\in \sQ_0}(\kappa^C_{\alpha})^{v} n_v + \frac{1}{2} \sum_{e\in \sQ_1\setminus \alpha (\sQ_0) } (\kappa^H_{\alpha})_{e} \nu^{e}  +  \sum_{v\in\sQ_0 , e\in \sQ_1 \setminus \alpha (\sQ_0)} {(\hat{\kappa}_{\alpha})_{e}}^{v} n_{v}\nu^{e}
\end{equation}
shows that \eqref{eq:hypertoricnotfactor} equals \eqref{eq:ZQprime}.
\end{proof}
Therefore $\mz_{\sQ^{\prime}}$ is, in general, a sum of traces of modules over $\mA^{\mC (\sQ)}_{\hbar} \otimes \mA^{\mH (\sQ)}_{\hbar}$, which not necessarily factorize into products of twisted traces.\par

\subsubsection{Examples of twists for hypertoric stacks}
\begin{enumerate}[(i)]
	\item As a consistency check, if ${\kk_e}^{v} \in \{ 0,\pm 1\}$, then ${(\hat{\kappa}_{\alpha})_{e}}^{v} \in \Z$ and \eqref{eq:hypertoricnotfactor} recovers the result of \cite{Gaiotto:2023kuc}, with $\tilde{\kappa}_{\alpha} \in \Z$ in \eqref{eq:kappacoefhighera}, $(\kappa^C_{\alpha})^{v},(\kappa^H_{\alpha})_{e} \in \Z$ in \eqref{eq:kappacoefhigherb}, and with twisted traces: 
	\begin{subequations}
	\begin{align}
		\hat{\chi}_{\alpha}^{\mC \left( \sQ \right)} &= \sum_{n \in \N^{\sQ_0}} \prod_{v\in\sQ_0} x_{v;\alpha}^{\frac{1}{2}} \left[(-1)^{\kappa_C^{v} } x_{v;\alpha} \right]^{n_v} , \\ 
		\hat{\chi}_{\alpha}^{\mH \left( \sQ \right)} &= \sum_{\nu \in \N^{\sQ_1 \setminus \alpha (\sQ_0)}} \prod_{e\in\sQ_1 \setminus \alpha (\sQ_0)} y_{e;\alpha}^{\frac{1}{2}} \left[(-1)^{\kappa_{H,e} } y_{e;\alpha} \right]^{\nu^e} ,
	\end{align}
	\end{subequations}
	The twist $e^{- \ii 2 \pi \Delta_{\alpha}} = (-1)^{R_A}(-1)^{R_B}$ and the half-integral shifts of the exponents are in perfect agreement with \cite{Bullimore:2020jdq}.
	\item For a given gauge theory, the rational numbers \eqref{eq:kappacoefhigher} may or may not be integers depending on the fixed point $p_{\alpha}$. Consider for instance a quiver with a single node and two edges:
		\begin{equation}
			\sQ^{\prime} = \fnode{1}\text{---}\gnode{1}\multe{\kappa}\fnode{1} .
		\end{equation}
		In this example, $\kk=(1,\kappa)$ and $\left( \widetilde{\mC}\left(\sQ^{\prime}\right)^{\ct_C}\right)_{\mathrm{red}} = \left\{ p_1, p_2 \right\}$, corresponding to the two maps \eqref{eq:alphaembed} with $\alpha (1)=1$ and $\alpha (1)=2$. Hence, 
		\begin{equation}
			\hat{\kappa}_1 = (1, \kappa) , \qquad \hat{\kappa}_2 = \left(\frac{1}{\kappa}, 1 \right) ,
		\end{equation}
		which, plugged into \eqref{eq:kappacoefhigher}, give 
		\begin{subequations}
		\begin{align}
			\tilde{\kappa}_1 &= \kappa , \qquad \kappa^C_{1} = 1 + \kappa , \qquad \kappa^H_{1} =1 + \kappa ; \\
			\tilde{\kappa}_2 &= \frac{1}{\kappa} , \qquad \kappa^C_{2} =\frac{ 1 + \kappa}{\kappa} , \quad \kappa^H_{2} = \frac{ 1 + \kappa}{\kappa} .
		\end{align}
		\end{subequations}
		Then, $e^{- \ii 2 \pi \Delta_{1}} = (-1)^{R_A}(-1)^{R_B}$ attains the form of \cite{Gaiotto:2023hda}, whereas $e^{- \ii 2 \pi \Delta_{2}}$ prevents the factorization of the trace.
	\item If $\sQ^{\prime}$ is obtained from $\sQred=\sQ$ by rescaling all the charges by $\kappa \in \Z_{\ne 0}$, then ${(\hat{\kappa}_{\alpha})_{e}}^{v} \in \{0, \pm 1 \}$ and $(\kappa_{\alpha}^C)^v,(\kappa_{\alpha}^H)_e \in \Z$ in \eqref{eq:kappacoefhigherb} are the same as those for $\sQ$. In this case, 
		\begin{equation}
			e^{- \ii 2 \pi \Delta_{\alpha}} = (-1)^{R_A (\sQ)}(-1)^{R_B (\sQ)}
		\end{equation}
		where on the right-hand side $R_A (\sQ)$ is the $\Z$-grading operator for $\scH_{\alpha}^{A}$ as a module over the quantization $\mA_{\hbar}^{\mC (\sQ)}$, not $\mA_{\hbar}^{\mC (\sQ^{\prime})}$; the same holds for $R_B (\sQ)$. In other words, the twist in this case depends on the R-charges of the quiver theory $\sQ$, not $\sQ^\prime$, and is insensitive to the rescaling of charges.
	\item\label{ex4higher} As a generalization of the previous example, let $\sQ^{\prime}$ be the quiver with higher charges obtained from $\sQ$ by gauging a finite Abelian subgroup $\Gamma_{\underline{\kappa}} = \prod_{v \in \sQ_0} \Z_{\kappa^v}$ of the topological symmetry $\ct_C$. The net effect is to rescale the charges of the hypermultiplets under the $U(1)$ gauge group at the node $v \in \sQ_0$ by $\kappa^{v}$. Then, 
		\begin{equation}\label{eq:chargesgauged}
			{\kk_e}^{v} =\begin{cases} \kappa^{v} &  v = \mathsf{h}(e) ,  \\ -\kappa^{v} &  v = \mathsf{t}(e) ,  \\ 0 & \text{ otherwise}, \end{cases}
		\end{equation}
		and also in this case ${(\hat{\kappa}_{\alpha})_{e}}^{v} \in \{0, \pm 1 \}$, $\tilde{\kappa}_{\alpha} \in \Z$ and $\kappa_{\alpha}^C,\kappa_{\alpha}^H $ are the same as those for $\sQ$. Hence, the twist is again given by the $\Z$-grading operators inherited from $\sQ$, different from the R-charges of $\sQ^{\prime}$:
		\begin{equation}
			e^{- \ii 2 \pi \Delta_{\alpha}} = (-1)^{R_A (\sQ)}(-1)^{R_B (\sQ)} .
		\end{equation}
\end{enumerate}

\subsubsection{Twisted traces and higher charges}

\begin{center}
\noindent\fbox{%
\parbox{0.94\linewidth}{
\begin{cor}\label{cor:hypertwtr}
Let $\sQ^{\prime}=(\sQ, \kk)$ describe an Abelian gauge theory with charge matrix given in \eqref{eq:chargesgauged}. Then
	\begin{equation}
		\mz_{\sQ^{\prime}} (\zeta, m) = \frac{1}{\lvert \Gamma_{\underline{\kappa}} \rvert} \sum_{p_{\alpha}\in  \mC \left( \sQ\right)^{\ct_C}} \ii^{-\tilde{\kappa}_{\alpha}} e^{\ii 2 \pi \pair{\zeta}{m}} \hat{\chi}_{\alpha}^{\mC \left( \sQ \right)} (x) \hat{\chi}_{\alpha}^{\mH \left( \sQ \right)} (y) ,
	\end{equation}
	where the twisted traces on the right-hand side are defined in \eqref{eq:tchiC}-\eqref{eq:tchiH}, and $\Gamma_{\underline{\kappa}} = \prod_{v \in \sQ_0} \Z_{\kappa^v}$.
\end{cor}
}}\end{center}
\begin{proof}
It follows immediately from the change of variables $\sigma_v^{\prime}=\kappa^{v}\sigma_v$ for all $v\in\sQ_0$, which yields 
	\begin{equation}\label{eq:ZQprimeisZQred}
		\mz_{\sQ^{\prime}} (\zeta, m) = \left( \prod_{v\in\sQ_0}\frac{1}{\kappa^{v}} \right) \mz_{\sQ} (\zeta^{\prime}, m) ,
	\end{equation}
	where $\zeta^{\prime,v}= \frac{\zeta^{v}}{\kappa^{v}}$. Then, $\lvert \Gamma_{\underline{\kappa}}\rvert = \prod_{v\in\sQ_0}\kappa^{v}$ together with the rewriting \eqref{eq:hypertoricnotfactor} applied to $\mz_{\sQ}$ shows the statement. Note that the rescaling $\zeta \mapsto \zeta^{\prime}$ is a manifestation of the quotient by $\Gamma_{\underline{\kappa}} \subset \ct_C$.
\end{proof}
\begin{rmk}
For given $\sQ^{\prime}$ as in Example (\ref{ex4higher}), there may exist a different quiver $\check{\sQ}$ with unit charges, such that $\mC \left( \check{\sQ}\right)= \mC \left( \sQ \right)/\Gamma_{\underline{\kappa}} = \mC \left(\sQ^{\prime} \right)$. The sphere quantization will endow $\mC \left( \sQ^\prime\right)$ and $\mC \left( \check{\sQ} \right)$ with \emph{different} Verma modules, and different twisted traces on them.
\end{rmk}

\subsection{Example: SQED with non-minimal highest weight representation}
\label{sec:higersqed}

The above discussion is illustrated in SQED, generalizing \S\ref{sec:SQED} to non-minimal charges:
\begin{equation}\label{eq:SQEDhigercharge}
	\sQ^{\prime} \ = \ \gnode{1} \multe{\kappa} \fnode{k} .
\end{equation}
The hypermultiplet corresponding to the edge is taken to transform in a $U(1)$-representation of highest weight $\kappa \ge 1$. By construction, $\sQred =\sQ$ is the usual SQED \eqref{eq:SQED}, and \eqref{eq:SQEDhigercharge} is obtained by gauging a $\Z_{\kappa}$ subgroup of the Coulomb branch isometry.
$\mathfrak{t}_C=\mathfrak{u}(1)$ and $\mathfrak{g}_H=\mathfrak{su}(k)$ exactly as in \S\ref{sec:SQED}.\par
\begin{itemize}
	\item $\mH (\sQ^\prime)$ is the symplectic reduction of $\C^{k} \oplus \C^{k}$ by a $\C^{\ast}$-action with weight $(\kappa,-\kappa)$. For $\zeta>0$, $\widetilde{\mH} (\sQ^\prime)$ is a $\Z_{\kappa}$-gerbe over $T^{\ast}\mathbb{P}^{k-1}$.\footnote{The gerbe structure of the Higgs branch of \eqref{eq:SQEDhigercharge} was already discussed by Mathew Bullimore in their talk at the 2022 workshop ``Mirrors, Moduli and M-theory in the Midlands''. I thank Marcus Sperling for pointing this out to me.}
	\item $\mC (\sQ^\prime)$ is a quotient stack with coarse moduli space a global $\Z_{\kappa}$-orbifold of the Coulomb branch of the theory with charge-1 hypermultiplets \cite{Nawata:2023rdx}. For generic mass parameters, 
	\begin{equation}
		\widetilde{\mC} (\sQ^\prime)= \left[ (\widetilde{\C^2 / \Z_k})/\Z_{\kappa} \right] .
	\end{equation}
	If $\mathrm{gcd}(\kappa,k)=1$, the Chinese remainder theorem implies that the coarse moduli space is
	\begin{equation}\label{eq:highersqeddoublequot}
		\mC \left( \sQ^{\prime}\right) \longrightarrow \C^2/\Z_{\kappa k} .
	\end{equation}
	If $\mathrm{gcd}(\kappa,k)\ne 1$, in principle $\mC \left( \sQ^{\prime} \right) =\C^2/\Gamma$ for an extension 
	\begin{equation}
		1 \longrightarrow \Z_{\kappa} \longrightarrow \Gamma \longrightarrow \Z_{k} \longrightarrow 1 .
	\end{equation}
	It was argued in \cite[Sec.6.1]{Hanany:2023uzn} that \eqref{eq:highersqeddoublequot} holds for any $\kappa$.
	Each fixed point $p_{\alpha} \in (\widetilde{\C^2 / \Z_k})^{\ct_C}$ descends to a length-$\kappa$, degree-0 subscheme of $\mC (\sQ^\prime)^{\ct_C}$.
\end{itemize}\par
\medskip
The sphere partition function satisfies
\begin{equation}
	\mz_{\gnode{1} \multe{\kappa} \fnode{k}} \left( \zeta, m \right) = \frac{1}{\kappa} \mz_{\gnode{1} \text{---} \fnode{k}}  \left( \frac{\zeta}{\kappa} , m \right)
\end{equation}
whereby, from the results reviewed in \S\ref{sec:SQED} and using $x=e^{-2\pi \frac{\zeta}{\kappa}}$ in the region $\zeta>0$,
\begin{align}
	\mz_{\gnode{1} \multe{\kappa} \fnode{k}} \left( \zeta, m \right) &= \frac{1}{\kappa} \sum_{\alpha=1}^{k} \ii^{-\tilde{\kappa}_{\alpha}} e^{\ii \frac{2 \pi}{\kappa} \pair{\zeta}{m}^{\text{\eqref{eq:SQED}}} } \hat{\chi}_{\alpha}^{\widetilde{\C^2/\Z_k}} \left( x\right) \hat{\chi}_{\alpha}^{T^{\ast}\mathbb{P}^{k-1}} \left( y \right) ,
\end{align}
with $y_{\alpha} = \left\{ y_{i;\alpha} , \ i =1 , \dots, k , \ i \ne \alpha \right\}$ being products of exponentials of the mass parameters, and $\tilde{\kappa}_{\alpha}=k-1 + d_{\alpha}$ where $d_{\alpha}$ depends on the Weyl chamber of the mass parameter space considered.\par 
This result matches the specialization of \eqref{eq:ZQprime} to the case \eqref{eq:SQEDhigercharge}, where $\hat{\kappa}_{\alpha}=1$ $\forall \alpha=1, \dots, k$. The twisted traces on the right-hand side can be identified with traces on modules of the quantized Coulomb and Higgs branch algebras of $\sQ^{\prime}$.\par
\medskip

The superconformal index provides an additional tool to study the quantized Higgs and Coulomb algebra. The conventions are collected in \S\ref{app:sci}.\par
The Coulomb limit of the index $\mI_{\gnode{1} \multe{\kappa} \fnode{k}}$ gives:
\begin{subequations}
\begin{align}
	\lim_{\tilde{q} \to 0 } \mI_{\gnode{1} \multe{\kappa} \fnode{k}} \left( q, \tilde{q};x,y\right) &= \left( \sum_{j=0}^{\kappa-1} q^{j} \right) \lim_{\tilde{q} \to 0 } \mI_{\gnode{1} \text{---}\fnode{k}} \left( q^{\kappa}, \tilde{q};x,y\right) \\
		& = \frac{\left(1-q^{\kappa k}\right) }{\left( 1-q\right) \left( 1-x q^{\kappa \frac{k}{2}}\right)\left( 1-x^{-1} q^{\kappa \frac{k}{2}}\right)} ,
\end{align}
\end{subequations}
where the latter expression is the Hilbert series of $\C[u_1,u_2,z]/(u_1 u_2-z^{\kappa k})$, indicating that $\mC (\sQ^\prime)$ is a quotient stack with coarse moduli space $\C^2 /\Z_{\kappa k}$. On the other hand, the Higgs limit is independent of $\kappa$,
\begin{align}
	\lim_{q\to 0 } \mI_{\gnode{1} \multe{\kappa} \fnode{k}} \left( q, \tilde{q};x,y\right) &=\lim_{q\to 0 } \mI_{\gnode{1} \text{---}\fnode{k}} \left( q, \tilde{q};x,y\right) .
\end{align}
The untwisted traces obtained from the Coulomb and Higgs branch Hilbert series equal the ones found in \cite{Bullimore:2020jdq} for SQED with $\kappa=1$, as predicted by the general discussion above.

\section{Twisted traces and quantized moduli stacks of Chern--Simons-matter theories}
\label{sec:CSquantized}

$\Th_{\underline{\kappa}}$ is a 3d $\mN=4$ Chern--Simons-matter theory as defined in \S\ref{sec:CSdef}.\par
The branches $\mM_A$ and $\mM_B$ of the moduli space of vacua of $\Th_{\underline{\kappa}}$ are treated here as Deligne--Mumford stacks, with coarse moduli spaces
\begin{equation}
	\mM_A \longrightarrow M_A , \qquad \mM_B  \longrightarrow M_B .
\end{equation}
$M_A$ and $M_B$ are expected to be symplectic singularities and, as such, to admit a deformation quantization. In words, there should exist two non-commutative algebras $\mA_{\hbar}^{A},\mA_{\hbar}^{B}$ together with isomorphisms of graded rings such that 
\begin{equation}
	\mA_{\hbar}^{A} / \left( \hbar \right) \cong  \C \left[ M_A  \right] , \qquad \mA_{\hbar}^{B}  / \left( \hbar \right) \cong \C \left[ M_B \right]  .
\end{equation}\par
This section initiates the investigation of the non-commutative algebras $\mA_{\hbar}^{A},\mA_{\hbar}^{B}$ and their modules. The sphere quantization is utilized to derive twisted traces on them.

\subsection{Quantized moduli stacks of Chern--Simons-matter theories}
\label{sec:VermaCS}

\subsubsection{Quantized A- and B-branches}
Typically, $\Th_{\underline{\kappa}}$ does not contain enough parameters to fully resolve $\mM_A$ and $\mM_B$. 
The lack of smoothings in the moduli spaces of Chern--Simons-matter theories prevents the analysis of \cite{Bullimore:2016nji} from carrying over directly to the present case, but does not prevent their quantization. 
\begin{rmk}[Quantization of symplectic singularities]
\hspace{1cm}
\begin{tenumerate}
	\item Bezrukavnikov--Kaledin \cite{BerKal} proved the existence of universal deformation quantizations of symplectic resolutions. 
	\item There exists a family of deformation quantizations of slices in the affine Grassmannian \cite{Kamnitzer:2012}.
	\item Losev \cite{Losev:2016} presents the quantization of any symplectic singularity.
\end{tenumerate}\par 
	Therefore, given a moduli stack, the quantization of its coarse moduli space is known, when it is a symplectic singularity \cite{Kamnitzer:2012,Losev:2016}. In \S\ref{sec:ex}, some examples will reproduce the explicit quantization of the du Val singularity $A_{\kappa -1}$ given in \cite[Sec.2.1]{Etingof:2020fls}.
\end{rmk}
The following is a compendium of the aforementioned results.
\begin{prop}\label{prop:quantumMAMB}\begin{enumerate}[(i)]
\item\label{quantYi} Let $\mM_A$ be a Deligne--Mumford stack, and assume its coarse moduli space $M_A$ is a symplectic singularity. There exists a family of non-commutative algebras $\mA_{\hbar}^{A}$ over $\C [\hbar]$, such that $\mA_{\hbar}^{A} /(\hbar ) \cong \C [M_A]$ is an isomorphism of $\Z$-graded rings.
\item\label{quantYii} Moreover, assume there exists a quiver $\sQ_A$ such that: (a) $\mC \left( \sQ_A \right) \cong M_A$; (b) $\sQ_A$ satisfies Hypothesis \ref{hyp:mava}.
Then, $\mA_{\hbar}^{A}$ is obtained specializing $\mA_{\hbar}^{\mC \left( \sQ_A \right)}$ at the origin of the parameter space.
\end{enumerate}
\end{prop}
\begin{proof}
(\ref{quantYi}) is a result of \cite[Sec.4.3]{Kamnitzer:2012} for slices in the affine Grassmannian, and \cite{Losev:2016} in general.\par
(\ref{quantYii}) is a direct consequence of the requirements on $\sQ_A$, and the analysis of \cite{Hilburn:2020aau}.
\end{proof}
\begin{rmk}
The map $\mA_{\hbar}^{\mC \left( \sQ_A \right)} \longrightarrow \mA_{\hbar}^{A} $ in Proposition \ref{prop:quantumMAMB}(\ref{quantYii}) need \emph{not} originate from a symplectic resolution of $M_A$.
\end{rmk}
The proposal of \cite{Marino:2025uub}, reviewed in \S\ref{sec:QAQB}, provides a constructive algorithm to derive $\sQ_A$ and $\sQ_B$ describing $M_A$ and $M_B$ when $\Th_{\underline{\kappa}}$ is based on an A-type quiver. Assuming its validity, combined with Proposition \ref{prop:quantumMAMB} it yields an explicit description of the quantized A- and B-branches. Note that Proposition \ref{prop:quantumMAMB} itself does not assume \cite{Marino:2025uub}.

\subsubsection{Verma modules}

Generic points of the parameter space of $\Th_{\underline{\kappa}}$ correspond to partial smoothings 
\begin{equation}
	\widetilde{M}_A \longrightarrow M_A, \qquad \widetilde{M}_B \longrightarrow M_B .
\end{equation}
By abuse of notation, let $\mA_{\hbar}^{A}$ also denote the quantization of $\widetilde{M}_A$. Under the assumptions of Proposition \ref{prop:quantumMAMB}(\ref{quantYii}), there exists a partial contraction 
\begin{equation}
\label{eq:CQAtoMA}
	\widetilde{\mC }\left( \sQ_A \right) \longrightarrow \widetilde{M}_A ,
\end{equation}
corresponding to specializing the resolution parameters to non-generic values.\par 
Assume that, under the map \eqref{eq:CQAtoMA}, some fixed points in $\widetilde{\mC}(\sQ_A)^{\ct_C}$ coalesce but there is no jump in dimension; whence the fixed locus $\widetilde{M}_A^{\ct_A}$ is a reducible, degree-0 subscheme of $\widetilde{M}_A$. Then, \eqref{eq:CQAtoMA} induces a surjective map
\begin{equation}\label{eq:equivariantCQAtoMA}
	\widetilde{\mC }\left( \sQ_A \right)^{\ct_C} \twoheadrightarrow \left( \widetilde{M}_A^{\ct_A} \right)_{\mathrm{red}}. 
\end{equation}
By Lemma \ref{lem:Vermatofp}(\ref{Vermatofp2}), the Verma modules over $\mA_{\hbar}^{A}$ are obtained from those of $\mA_{\hbar}^{ \widetilde{\mC} (\sQ_A)}$ by imposing the isomorphisms induced by \eqref{eq:equivariantCQAtoMA} \cite{Hilburn:2020aau}.

\subsection{Verma modules and twisted traces from Chern--Simons-matter theories}
\label{sec:sphereCS}
The notation continues as in \S\ref{sec:VermaCS}. It is convenient to recall the partition function of $U(N)_{\kappa}$ supersymmetric pure Chern--Simons theory:\footnote{The Witten--Reshetikhin--Turaev topological invariant of $\mathbb{S}^3$ is obtained by replacing $\kappa $ with $\kappa + N$.}
\begin{equation}
	\mz_{\csnode{N}{\kappa}}  = \frac{1}{\lvert \kappa\rvert ^{N/2}} \prod_{1\le a <b \le N} 2 \sin \left( \frac{\pi}{\kappa} (b-a)\right) .
\end{equation}\par
The main result of this work is the proposal:

\begin{center}
\noindent\fbox{%
\parbox{0.95\linewidth}{%
\begin{conj}\label{myconj1}
	Assume the fixed locus $\widetilde{M}_A^{\ct_A}$ is a degree-0 scheme. Then, for every $p_{\alpha} \in \left(\widetilde{M}_A^{\ct_A}\right)_{\mathrm{red}}$ there exist
	\begin{itemize}
		\item A Verma module $\scH_{\alpha}^{A}$ over $\mA_{\hbar}^{A}$, and a Verma module $\scH_{\alpha}^{B}$ over $\mA_{\hbar}^{B}$;
		\item A bilinear pairing $\pair{\cdot}{\cdot} : \mathfrak{t}_A \times \mathfrak{t}_B \to \C $;
		\item An integer $\kappa_{\alpha} \ne 0$;
		\item A subset $\sV_{\alpha} \subseteq \sQ_0$, determined solely by the local geometry around $p_{\alpha}$;
	\end{itemize}
	such that 
		\begin{equation}\label{eq:conj1Z}
		\mz_{\Th_{\underline{\kappa}}}  = \sum_{p_{\alpha} \in \left(\widetilde{M}_A^{\ct_A}\right)_{\mathrm{red}}} S_{\alpha}~e^{\ii 2 \pi \pair{ \zeta}{ m} } \tr_{\scH^{A}_{\alpha}\otimes\scH^{B}_{\alpha}} \left(  e^{-\ii \frac{2 \pi}{\kappa_{\alpha}} \left( \frac{1}{2}+ R_A\right) \otimes \left(\frac{1}{2} +R_B\right) } x^{J_A} y^{J_B} \right) ,
	\end{equation}
	where $R_A,R_B$ are, respectively, the $\Z$-grading operators on $\scH^{A}_{\alpha},\scH^{B}_{\alpha}$ induced by the contracting $\C^{\ast}$-action on $M_A,M_B$, and 
	\begin{equation}\label{eq:conj1S}
		S_{\alpha} := \prod_{v \in \sV_{\alpha}} \mz_{\csnode{N_v}{\kappa_v}}  .
	\end{equation}
\end{conj}
}}\end{center}
This is to be compared with the Gaiotto--Okazaki Conjecture \ref{c:GO}. For 3d $\mN=4$ Chern--Simons-matter theories, Conjecture \ref{myconj1} predicts that the partition function attains the form of a sum over twisted traces on the Verma modules $\scH^{A}_{\alpha}\otimes\scH^{B}_{\alpha}$, but such trace does not generically factorize into the product of two twisted traces. 
Conjecture \ref{myconj1} is tested in examples in \S\ref{sec:ex}.\par
\begin{rmk}
\begin{itemize}
\item For A-type quivers, which have $\kappa_v \in\{ 0, \pm \kappa \}$ for all $v \in \sQ_0$, $\kappa_{\alpha} \in \{ 1, \kappa \}$.
\item If $\lvert \kappa_{\alpha} \rvert =1$ for each $\alpha$, $e^{-\frac{\ii 2 \pi}{\kappa_{\alpha}} \left( \frac{1}{2}+R_A\right)\left( \frac{1}{2}+R_B\right) } = - \ii (-1)^{R_A}(-1)^{R_B}$, and the conjecture reduces to the form of Conjecture \ref{c:GO}. In particular, this is the case for A-type quivers with $\kappa=1$.
\item Isomorphic modules are identified as objects in the category $\mA_{\hbar}^{\bullet}\text{-mod}$, and they contribute only once to \eqref{eq:conj1Z}. In other words, the partition function is sensitive to the Verma modules associated with the reduced scheme $\left(\widetilde{M}_A^{\ct_A}\right)_{\mathrm{red}}$.
\end{itemize}
\end{rmk}

\subsection{Quantized hypertoric stacks and magnetic quivers of Chern--Simons-matter theories}
\label{sec:magQprime}

$\Th_{\underline{\kappa}}$ is now assumed to be an Abelian 3d $\mN=4$ quiver Chern--Simons-matter theory based on an A-type or affine A-type Dynkin diagram.
In this case, the Chern--Simons levels are $\kappa_v \in \{ 0, \pm \kappa \}$ and appear with alternating sign.\par
In the notation as above, recall that the proposal of \cite{Marino:2025uub} (reviewed in \S\ref{sec:QAQB}) provides two `magnetic' quivers $\sQ_A, \sQ_B$ with $\mC \left( \sQ_A \right) = M_A$ and $\mC \left( \sQ_B \right) = M_B$.\par
\begin{itemize}
	\item In general, $\mC \left( \sQ_A \right)$ admits desingularizations that are obstructed in $M_A$; likewise for $\mC \left(\sQ_B \right)$ and $M_B$. Simple examples of this fact are given in \S\ref{sec:ex}.
	\item In general, the twisted traces on $\mA^{A}_{\hbar}$-modules differ from the twisted traces on $\mA_{\hbar}^{\mC \left( \sQ_A \right)}$-modules read off from $\mz_{\sQ_A}$. The same holds for the B-branch. Again, this can be checked in all examples given in \S\ref{sec:ex}.
\end{itemize}

\begin{table}
\centering
\begin{tabular}{|c|c|}
	\hline
		local quiver $\subset \Th_{\underline{\kappa}}$ & local magnetic quiver $\subset \sQ^{\prime}$ \\
		\hline
		$\cdots \underbrace{\csnode{1}{\kappa} \text{---}\csnode{1}{0} \text{---}\cdots  \text{---}\csnode{1}{0} \text{---} \csnode{1}{-\kappa}}_{k} \cdots $ & \begin{tikzpicture}[baseline=-1] \node (g) at (0,0) {$\cdots \gnode{1} \cdots $}; \node (fl) at (-1,-0.75) {$\fnode{1}$};  \node (fc) at (0,-0.75) {$\cdots$}; \node (fr) at (1,-0.75) {$\fnode{1}$}; \path[anchor=east] (-0.8,-0.6) edge node {$\scriptstyle \kappa$} (-0.25,-0.25);\path[anchor=west] (0.8,-0.6) edge node {$\scriptstyle \kappa$} (0.25,-0.25); \node[anchor=north] (u) at (0,-0.85) {$\underbrace{\hspace{2.3cm}}_{k}$};\end{tikzpicture}\\
		\hline
		$\csnode{1}{-\kappa} \text{---} \cdots $ & $\fnode{1} \text{---} \cdots$ \\
		\hline
		$ \cdots \text{---}\csnode{1}{\kappa} $ & $ \cdots \text{---} \fnode{1}$ \\
		\hline
	\end{tabular}
\caption{Prescription to obtain the auxiliary quiver $\sQ^{\prime}$ for Abelian linear Chern--Simons-matter theories.}
\label{tab:myQprime}
\end{table}\par
Here a different, auxiliary quiver $\sQ^{\prime}$ is proposed, which includes edges of charge-$\kappa$. The replacement rule to derive $\sQ^{\prime}$ is given in Table \ref{tab:myQprime}. This is largely inspired by but different from \cite{Marino:2025uub}, and employs quiver gauge theories with non-minimal charges studied in \S\ref{sec:highercharge}.\par
\begin{conj}\label{myconjmag}
	Let $\Th_{\underline{\kappa}}$ be an Abelian 3d $\mN=4$ Chern--Simons-matter theory whose underlying quiver is an A-type or affine A-type Dynkin diagram, and assume that $\sum_{v \in \sQ_0} \kappa_v=0$. Consider the quiver $\sQ^{\prime}$ obtained by the prescription in Table \ref{tab:myQprime}. Then
	\begin{enumerate}[(i)]
		\item $\mM_A \cong \mC (\sQ^{\prime})$ as stacks, and  $\mM_B \cong \mH (\sQ^{\prime})$ as gerbes.
		\item $\mz_{\Th_{\underline{\kappa}}} = \mz_{\sQ^{\prime}}$.
	\end{enumerate}
\end{conj}
This conjecture is verified in families of examples in \S\ref{sec:ex}.
\begin{cor}\label{cor:magquiver}
	Let $\Th_{\underline{\kappa}}$ and $\sQ^{\prime}$ be as above, and assume they satisfy Conjecture \ref{myconjmag}. Then 
	\begin{enumerate}[(i)]
		\item\label{prop:aux1} $\Th_{\underline{\kappa}}$ satisfies Conjecture \ref{myconj1}.
		\item\label{prop:aux2} If $\sQ^{\prime}$ has charge matrix given in \eqref{eq:chargesgauged}, the twisted traces of the Verma modules $\scH_{\alpha}^{A}$ read off by the sphere partition function are twisted by the grading operator on $\C \left[ \mC \left( \sQ \right) \right]$. The same holds for the B-branch.
	\end{enumerate}
\end{cor}
\begin{proof}
(\ref{prop:aux1}) follows from $\mz_{\Th_{\underline{\kappa}}} = \mz_{\sQ^{\prime}}$ and Corollary \ref{cor:Zhighertrace}.\par
(\ref{prop:aux2}) follows from $\mz_{\Th_{\underline{\kappa}}} = \mz_{\sQ^{\prime}}$ and Corollary \ref{cor:hypertwtr}.
\end{proof}

\section{Examples}
\label{sec:ex}
This section contains several explicit examples demonstrating the main results and supporting the main conjectures. For concreteness in the exposition, it is assumed that 
\begin{equation}
	\kappa >0 .
\end{equation}

\subsection{Abelian two-node quiver}
\label{sec:ex2node}
A simple example is
\begin{equation}\label{eq:ex2node}
		\csnode{1}{\kappa}\text{---} \csnode{1}{-\kappa}\ .
\end{equation}
$\ct_B$ is trivial, and $\ct_A \cong \C^{\ast}$. The A-branch Hilbert series \cite{Li:2023ffx}
\begin{equation}\label{eq:ex2HS}
	\hs_{\C [ M_{A,\eqref{eq:ex2node}} ] } \left( q;x \right) = \frac{1-q^{\kappa}}{(1-q)\left(1-xq^{\frac{\kappa}{2}}\right)\left(1-x^{-1}q^{\frac{\kappa}{2}}\right)}  
\end{equation}
equals the Hilbert series of $\C[u_1,u_2,z]/(u_1u_2-z^{\kappa})$. Therefore, the coarse moduli space is
\begin{equation}
	M_{A,\eqref{eq:ex2node}} = \mathrm{Spec} \frac{\C \left[u_1,u_2,z\right]}{\left( u_1 u_2-z^{\kappa}\right)} \cong \C^2/\Z_{\kappa} ,
\end{equation}
while the B-branch is trivial as an algebraic variety \cite{Li:2023ffx}. In fact, a derivation analogous to the one for ABJM theory \cite{Bergman:2020ifi} shows that \eqref{eq:ex2node} possesses a $\Z_{\kappa}$ 1-form symmetry, and therefore the moduli spaces are best characterized as stacks:
\begin{itemize}
\item $\mM_{A,\text{\eqref{eq:ex2node}}} = \left[ \C^2/\Z_{\kappa} \right] $;
\item $ \mM_{B,\text{\eqref{eq:ex2node}}}$ is a $\Z_{\kappa}$-gerbe over a point.
\end{itemize}
The untwisted trace read off from the limit $q \to 1$ of \eqref{eq:ex2HS} is
\begin{equation}
	\chi^{A,\eqref{eq:ex2node}} =  \frac{x^{\frac{1}{2}}}{1-x} .
\end{equation}\par
Denoting $\zeta$ the free parameter in \eqref{eq:ex2node}, the partition function evaluates to 
\begin{equation}\label{eq:Z2node}
	\mz_{\text{\eqref{eq:ex2node}}} = \frac{1}{\kappa \ch (\zeta)} .
\end{equation}
\begin{thm}
	Conjecture \ref{myconj1} holds for \eqref{eq:ex2node}.
\end{thm}
\begin{proof}
	The theory \eqref{eq:ex2node} does not have parameters to resolve $M_{A,\eqref{eq:ex2node}}=\C^2/\Z_{\kappa}$, and the fixed point set is $\left( M_{A,\eqref{eq:ex2node}}^{\ct_A}\right)_{\text{red}} = \{0\}$.\par 
	The overall factor is $\frac{1}{\kappa} =\mz_{\csnode{1}{\kappa}}  \mz_{\csnode{1}{-\kappa}} $, and $\frac{1}{\ch (\zeta)}$ is the twisted trace on the Verma module associated with $0\in (\C^2)^{\ct_A}$. The operator $R_A$ in \eqref{eq:ex2node} acts on this Verma module with eigenvalues $\kappa \N$. These facts recast \eqref{eq:Z2node} in the form of Conjecture \ref{myconj1}, where $\kappa_\alpha=\kappa$.
\end{proof}
Moreover, the twisted trace in \eqref{eq:Z2node} is consistent with Corollary \ref{cor:magquiver}(\ref{prop:aux2}).

\subsubsection{Magnetic quiver analysis}
The prescription of \S\ref{sec:magQprime} associates to \eqref{eq:ex2node} the quiver
\begin{equation}\label{eq:aux2node}
	\sQ^{\prime}_{\eqref{eq:ex2node}} \ = \ \gnode{1} \multe{\kappa} \fnode{1}  .
\end{equation}
\begin{itemize}
\item The Coulomb branch is \cite{Nawata:2023rdx,Hanany:2023uzn},
	\begin{equation}
		\mC \left( \gnode{1} \multe{\kappa} \fnode{1}\right) =  \left[ \mC \left( \gnode{1} \text{---} \fnode{1}\right) / \Z_{\kappa}\right] =  \left[\C^2 / \Z_{\kappa} \right] .
	\end{equation}
\item The coarse moduli space of the Higgs branch is a point, but the theory possesses a 1-form symmetry. Hence $ \mH \left(\gnode{1} \multe{\kappa} \fnode{1}\right)$ is a $\Z_{\kappa}$-gerbe over a point.
\end{itemize}
They agree with, respectively, $\mM_{A,\eqref{eq:ex2node}}$ and $\mM_{B,\eqref{eq:ex2node}}$.\par
Furthermore, computing the partition function of \eqref{eq:aux2node} with FI parameter $\zeta^{\prime}$ gives 
\begin{equation}
		\mz_{\sQ^{\prime}_{\eqref{eq:ex2node}}} (\zeta^{\prime})= \frac{1}{\kappa \ch \left( \frac{\zeta^{\prime}}{\kappa}\right)}= \mz_{\eqref{eq:ex2node}} \left( \zeta = \frac{\zeta^{\prime}}{\kappa} \right) .
\end{equation}
Computing the superconformal index of both \eqref{eq:ex2node} and \eqref{eq:aux2node} for the first few values of $\kappa$ gives an additional match.\par
\begin{thm}
	Conjecture \ref{myconjmag} holds for \eqref{eq:ex2node}.
\end{thm}
\begin{rmk}[Comparison with {\cite{Marino:2025uub}}] The proposal of \S\ref{sec:QAQB} yields the pair of quivers
\begin{equation}\label{eq:QAisSQED}
	\sQ_{A,\eqref{eq:ex2node}} \ = \ \gnode{1}\text{---}\fnode{\kappa} , \qquad \sQ_{B,\eqref{eq:ex2node}} \ = \ \varnothing .
\end{equation}
The Coulomb branch of \eqref{eq:aux2node} does not admit a desingularization, exactly as in the original \eqref{eq:ex2node}, and contrary to \eqref{eq:QAisSQED}. Besides, \eqref{eq:QAisSQED} reproduces the symplectic singularity as the coarse moduli space, but neglects the gerbe structure.
\end{rmk}

\subsubsection{Field theory analysis}
\label{sec:QFTex2node}
\begin{table}
\centering
\begin{tabular}{|l|c c|c|}
	\hline
		 & $q$ & $\tilde{q}$ & $V_{(\lambda_1,\lambda_2)}$ \\
		\hline
		$U(1)_{\kappa}$ & $+1$ & $-1$ & $-\kappa \lambda_1$ \\
		$U(1)_{-\kappa}$ & $-1$ & $+1$ & $\kappa \lambda_2$ \\
		\hline
	\end{tabular}
\caption{Charges of the chiral multiplets and bare monopole operators in the two-node Chern--Simons-matter quiver \eqref{eq:ex2node}.}
\label{tab:ops2node}
\end{table}

The field theoretic study of $M_A$ is based on \cite{Assel:2017eun}, and special attention is paid here to the obstruction to desingularization.\par
The scalars in the chiral multiplets forming the bi-fundamental hypermultiplet of \eqref{eq:ex2node} are denoted $(q,\tilde{q})$, the complex scalars in the vector multiplet are denoted $\{ \varphi_v\}_{v=1,2}$; finally, $V_{\lambda}$ is the bare monopole operator with magnetic charge $\lambda =(\lambda_1, \lambda_2)$. The gauge charges are summarized in Table \ref{tab:ops2node}.\par
The complex moment map equations read:
\begin{equation}
	q \tilde{q} = \kappa \varphi_1, \qquad - q \tilde{q} = - \kappa \varphi_2 .
\end{equation}
The real moment map is dealt with analogously and is omitted here for brevity; see \cite[Sec.2]{Assel:2017eun} for the full set of equations. The canonical superpotential supplements these equations with
\begin{equation}
	q (\varphi_1-\varphi_2)=0 , \qquad \tilde{q} (\varphi_2-\varphi_1)=0 .
\end{equation}
There exists a solution parametrized by the gauge-invariant operator $z := q \tilde{q} = \kappa \varphi_1 = \kappa \varphi_2$, together with the dressed, gauge-invariant monopole operators 
\begin{equation}
	u_1 := V_{(1,1)} q^{\kappa} , \qquad u_2 := V_{(-1,-1)} \tilde{q}^{\kappa} .
\end{equation}
The generators $(z,u_1,u_2)$ of the ring $\C \left[M_{A,\eqref{eq:ex2node}} \right]$ have R-charges $(2,\kappa,\kappa)$. Observe the relation:
\begin{equation}\label{eq:2nodeOPE1}
	u_{1} u_{2} = z^{\kappa} , 
\end{equation}
which shows that 
\begin{equation}
	\C \left[ M_{A,\eqref{eq:ex2node}} \right] = \C \left[ u_1,u_2,z\right] /\left( u_1u_2-z^{\kappa} \right) ,
\end{equation}
reproducing \eqref{eq:ex2HS}. Contrarily to 3d $\mN=4$ SQED, the model \eqref{eq:ex2node} does not admit parameters to resolve \eqref{eq:2nodeOPE1}.

\subsubsection{Brane analysis}
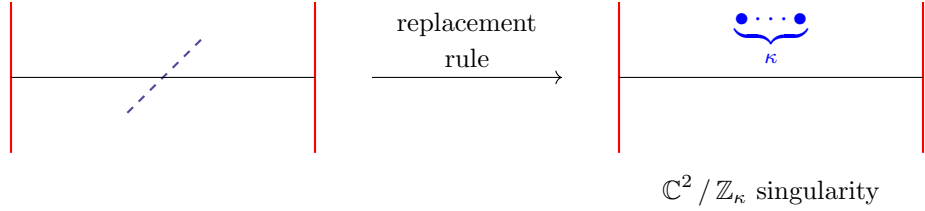
\begin{figure}
\centering
\begin{tikzpicture}
	\draw[red,thick] (-2,4) -- (-2,2);
	\draw[red,thick] (2,4) -- (2,2);
	\draw[Violet,dashed,thick] (.5,3.5) -- (-.5,2.5);
	\draw (-2,3) -- (2,3);
	
	\draw[red,thick] (6,4) -- (6,2);
	\draw[red,thick] (10,4) -- (10,2);
	\node at (8,3.5) {\begin{color}{blue}$\underbrace{\bullet \cdots\bullet}_{\kappa}$\end{color}};
	\draw (6,3) -- (10,3);
	\node at (8,1.5) {\small $\C^2/\Z_{\kappa}$ singularity};

	\path[->] (2.75,3) edge node[anchor=south,align=center] {\small replacement\\ \small rule} (5.25,3);
\end{tikzpicture}
\caption{Brane realization of the two-node Chern--Simons-matter quiver \eqref{eq:ex2node} (left) and the magnetic quiver for the A-branch (right).}
\label{fig:2node}
\end{figure}

The brane setup realizing \eqref{eq:ex2node} is shown in Figure \ref{fig:2node}, alongside the brane realization of $\sQ_{A,\eqref{eq:ex2node}}$ \eqref{eq:QAisSQED}.\par
The singularity of $M_{A,\eqref{eq:ex2node}}$ originates from tensionless strings when the D3-brane segment intersects the $(1,\kappa)$-brane. However, as opposed to the brane realization of \eqref{eq:QAisSQED}, there is no motion of the $(1,\kappa)$-brane that produces a (partial) smoothing of the singularity.

\subsection{Abelian A-type quiver with \texorpdfstring{($\kappa, 0, \dots, 0,-\kappa$)}{(k,0,...,0,k)}}
\label{sec:exk00k}
Consider the quiver
\begin{equation}\label{eq:ex3node0}
		\csnode{1}{\kappa}\text{---} \overbrace{\csnode{1}{0} \text{---} \cdots \text{---} \csnode{1}{0}}^{k-1} \text{---} \csnode{1}{-\kappa}\ ,
\end{equation}
modeled on the $A_{k+1}$ Dynkin diagram, $\lvert \sQ_0 \rvert=k+1$. The A- and B-branch of \eqref{eq:ex3node0} are derived from field theory in \S\ref{sec:QFTlin}. By the same argument as in \S\ref{sec:ex2node}, they are best described as stacks.
\begin{itemize}
\item $\mM_{A,\eqref{eq:ex3node0}}= \left[ \C^2 / \Z_{k\kappa} \right] $, which only admits a partial smoothing into $k$ singular points, each locally a du Val singularity $A_{\kappa-1}$.
\item The smoothing of $\mM_{B,\eqref{eq:ex3node0}}$ is a $\Z_{\kappa}$-gerbe over $T^{\ast}\mathbb{P}^{k-1}$.
\end{itemize}\par 
\medskip
Introducing one mass parameter $m$, associated for concreteness with the right-most edge of the quiver, and redundant FI parameters $\{\zeta^{v} \}_{v=0}^{k}$ subject to $\sum_{v=0}^{k} \zeta^{v} =0$, the sphere partition function is:
\begin{equation}
	\mz_{\text{\eqref{eq:ex3node0}}} = \int_{-\infty}^{+\infty} \dd \sigma_0\int_{-\infty}^{+\infty} \dd \sigma_1 \cdots \int_{-\infty}^{+\infty}\dd \sigma_k \frac{e^{\ii \pi \kappa \left( \sigma_0^2-\sigma_k^2\right) + 2 \pi \ii \sum_{v=0}^{k-1}\zeta^{v} \sigma_v }}{\ch \left( \sigma_{k-1}-\sigma_k +m\right)\prod_{v=1}^{k-1}\ch \left( \sigma_{v-1}-\sigma_{v}\right) } \label{eq:Z3node0a}.
\end{equation}
\begin{lem}
\begin{equation}\label{eq:AbelianCStoSQED}
	\mz_{\text{\eqref{eq:ex3node0}}} = \mz_{\csnode{1}{\kappa}} \mz_{\csnode{1}{-\kappa}} \mz_{\gnode{1}\text{---}\fnode{k}} ,
\end{equation}
where, on the right-hand side, the FI parameter is $\tilde{\zeta}= m$, and the masses are linear combinations of $\{\zeta^{v} \}_{v=1}^{k-1}$.
\end{lem}
\begin{proof}
It is shown in \S\ref{app:Abeliancalc}, from manipulations and direct integration, that
\begin{equation}
	\mz_{\text{\eqref{eq:ex3node0}}} = \frac{1}{\kappa} \int_{-\infty}^{+\infty} \dd \sigma \frac{e^{\ii 2 \pi m \sigma }}{ \prod_{i=1}^{k}\ch \left( \sigma - \tilde{m}_{i}\right)} , \label{eq:Z3node0b}
\end{equation}
where $\tilde{m}_{i}$ are linear combinations of $\{\zeta^{v} \}_{v=1}^{k-1}$ only, after using $\sum_{v=0}^{k}\zeta^{v} =0$. Recognizing the partition function of $U(1)$ gauge theory with $k$ hypermultiplets shows \eqref{eq:AbelianCStoSQED}.
\end{proof}

\begin{thm}
	Conjecture \ref{myconj1} holds for \eqref{eq:ex3node0}.
\end{thm}
\begin{proof}
	It follows from \eqref{eq:AbelianCStoSQED} and the fact that SQED satisfies Conjecture \ref{c:GO} \cite{Gaiotto:2019mmf,Bullimore:2020jdq}.
\end{proof}
Moreover, \eqref{eq:AbelianCStoSQED} is consistent with Corollary \ref{cor:magquiver}(\ref{prop:aux2}).

\subsubsection{Magnetic quiver analysis}
The prescription of \S\ref{sec:magQprime} for the auxiliary quiver $\sQ^{\prime}$ yields:
\begin{equation}\label{eq:aux3node}
	\sQ^{\prime}_{\eqref{eq:ex3node0}} \ = \ \gnode{1} \multe{\kappa} \fnode{k} ,
\end{equation}
namely $U(1)$ gauge theory with $k$ hypermultiplets of charge $\kappa$, already analyzed in \S\ref{sec:higersqed}. 
\begin{itemize}
	\item The coarse moduli space of the Coulomb branch of \eqref{eq:aux3node} is $\C^2 / \Z_{\kappa k} $ \cite{Nawata:2023rdx,Hanany:2023uzn}.
	\item The Higgs branch $\mH \left( \sQ^{\prime}_{\eqref{eq:ex3node0}} \right)$ is a $\Z_{\kappa}$-gerbe over the Higgs branch of SQED.
\end{itemize}
They match with, respectively, $\mM_{A,\eqref{eq:ex3node0}}$ and $\mM_{B,\eqref{eq:ex3node0}}$.\par
The partition functions of \eqref{eq:ex3node0} and \eqref{eq:aux3node} are also equal, by virtue of \eqref{eq:Z3node0b}:
\begin{equation}
	\mz_{\text{\eqref{eq:ex3node0}}} (m, \zeta) = \left.\mz_{\text{\eqref{eq:aux3node}}} (m^{\prime}, \zeta^{\prime}) \right\rvert_{m^{\prime}_i = \tilde{m}_i, \zeta^{\prime} = \kappa \tilde{\zeta} } .
\end{equation}
Furthermore, computing the superconformal indices of \eqref{eq:ex3node0} and \eqref{eq:aux3node} for the first few values of $k$ and $\kappa$ shows that they match.\par
\begin{thm}
	Conjecture \ref{myconjmag} holds for \eqref{eq:ex3node0}.
\end{thm}

\begin{rmk}[Comparison with {\cite{Marino:2025uub}}] The magnetic quivers of \S\ref{sec:QAQB} for \eqref{eq:ex3node0} are:
\begin{equation}\label{eq:QAQBex3}
	\sQ_A \ = \ \gnode{1}\text{---}\fnode{\kappa k} , \qquad \sQ_B \ = \ \gnode{1}\text{---}\fnode{k}  .
\end{equation}
$\mC\left( \sQ_A \right)$ of \eqref{eq:QAQBex3} matches the coarse moduli space of $\mC \left( \sQ^{\prime}_{\eqref{eq:ex3node0}} \right)$, and $\mC\left( \sQ_B \right)$ of \eqref{eq:QAQBex3} matches the base of the gerbe in $\mH \left( \sQ^{\prime}_{\eqref{eq:ex3node0}} \right)$.
\end{rmk}

\subsubsection{Field theory analysis}
\label{sec:QFTlin}

\begin{table}
\centering
\begin{tabular}{|l|c c|c|}
	\hline
		 & $q_e$ & $\tilde{q}_e$ & $V_{(\lambda_0,\dots, \lambda_k)}$ \\
		\hline
		$U(1)_{\kappa}$ & $\delta_{e, 0\to 1}$ & $-\delta_{e, 0\to 1}$ & $-\kappa \lambda_0$ \\
		$v^{\text{th}} \ U(1)_0$ & $\delta_{e,v \to v+1}- \delta_{e,v-1\to v}$ & $\delta_{e,v-1 \to v} -\delta_{e,v \to v+1}$ & 0 \\
		$U(1)_{-\kappa}$ & $-\delta_{e, k-1\to k}$ & $\delta_{e, k-1\to k}$ & $\kappa \lambda_k$ \\
		\hline
	\end{tabular}
\caption{Charges of the chiral multiplets and bare monopole operators in the A-type Chern--Simons-matter quiver \eqref{eq:ex3node0}.}
\label{tab:opslinear}
\end{table}

The analysis of the vacuum equations in field theory is a generalization of \cite[Sec.4.2]{Assel:2017eun}, and is very similar to \S\ref{sec:QFTex2node}.\par
The field content and charges are summarized in Table \ref{tab:opslinear}, where $q_e,\tilde{q}_e$ are the chiral multiplets corresponding to $e\in \sQ_1$ and $V_{\lambda}$ is the bare monopole operator of highest weight $\lambda$. The complex moment map and canonical superpotential give the equations 
\begin{subequations}
\begin{align}
	q_{0\to 1}\tilde{q}_{0\to 1} & = \kappa \varphi_0 \label{eq:qftlin1} \\
	q_{v \to v+1}\tilde{q}_{v \to v+1} - q_{v-1\to v}\tilde{q}_{v-1\to v} &=0 , \qquad \forall 1 \le v \le k-1 \label{eq:qftlin2} \\
	-q_{k-1\to k}\tilde{q}_{k-1\to k} & = -\kappa \varphi_k \label{eq:qftlin3} \\
	q_e \left( \varphi_{\mathsf{h}(e)} - \varphi_{\mathsf{t}(e)}  \right) &=0 , \qquad \forall e \in \sQ_1 \label{eq:qftlin4} \\
	\tilde{q}_e \left( -\varphi_{\mathsf{h}(e)} + \varphi_{\mathsf{t}(e)}  \right) &=0, \qquad \forall e \in \sQ_1 . \label{eq:qftlin5}
\end{align}
\end{subequations}
As monopole operators, the generators can be taken to be
\begin{equation}
	u_0 = V_{(1,\dots,1)} \prod_{e\in \sQ_1} q_e^{\kappa} , \qquad  u_{v}^{\pm}= V_{(0, \dots, \lambda_v=\pm 1, \dots , 0)} , \qquad u_k =V_{(-1,\dots,-1)} \prod_{e\in \sQ_1} \tilde{q}_e^{\kappa} 
\end{equation}
subject to the constraints \cite[Eq.(2.8)]{Assel:2017eun}
\begin{subequations}\label{eq:monoeqlin}
\begin{align}
	u_v^{\pm} q_{v-1\to v} &=0=u_{\pm} \tilde{q}_{v-1\to v} , \qquad \forall 1 \le v \le k-1 \\
	u_v^{\pm} q_{v\to v+1} &=0=u_{\pm} \tilde{q}_{v\to v+1} , \qquad \forall 1 \le v \le k-1 .
\end{align}
\end{subequations}\par
Denoting for short $z := q_{0\to 1}\tilde{q}_{0\to 1}$, \eqref{eq:qftlin1}-\eqref{eq:qftlin3} imply $q_e \tilde{q}_e =z$ for all $e\in \sQ_1$ and $\kappa \varphi_0 = \kappa \varphi_k = z$. There are two sets of solutions, corresponding to the A- and B-branch.
\begin{itemize}
\item The A-branch is obtained solving \eqref{eq:qftlin4}-\eqref{eq:qftlin5} by $\kappa \varphi_v=z$ for all $v$, and solving \eqref{eq:monoeqlin} by $u_v^{+}=0=u_v^{-}$ for all $1\le v \le k-1$. The remaining generators are $u_0, u_k,z$, subject to the relation 
	\begin{equation}
		u_0 u_k = \prod_{e \in \sQ_1} (q_e \tilde{q}_e)^{\kappa} = z^{k \kappa} .
	\end{equation}
	Therefore 
	\begin{equation}
		M_{A,\eqref{eq:ex3node0}} = \mathrm{Spec} \frac{\C \left[ u_0 ,u_k, z \right] }{\left( u_0 u_k - z^{k \kappa } \right)} = \C^2 / \Z_{k \kappa} .
	\end{equation}
	It only admits a partial smoothing into $k$ locally $A_{\kappa-1}$ singularities.
\item The B-branch arises from solving \eqref{eq:monoeqlin} by $q_e=0=\tilde{q}_e$ for every $e$, which implies $z=0$ and $u_0=0=u_k$. Since \eqref{eq:qftlin4}-\eqref{eq:qftlin5} are automatically satisfied, the remaining parameters are $\left\{ u_v , \varphi_v \right\}_{v=1}^{k-1}$. These are independent of $\kappa$. The case $k=2$ was analyzed in \cite[Sec.4.2]{Assel:2017eun}, where it is shown that $M_{B,\eqref{eq:ex3node0}}$ is equal to the Higgs branch of SQED with $k=2$. Assuming the derivation therein extends to any $k$, the final result is 
	\begin{equation}
		\widetilde{M}_{B,\eqref{eq:ex3node0}} = T^{\ast} \mathbb{P}^{k-1}.
	\end{equation}
\end{itemize}

\subsubsection{Brane analysis}
\begin{figure}
\centering
\begin{tikzpicture}
	\draw[red,thick] (-3,4) -- (-3,2);
	\draw[red,thick] (3,4) -- (3,2);
	\draw[Violet,dashed,thick] (-.5,3.5) -- (-1.5,2.5);
	\draw[Violet,dashed,thick] (1.5,3.5) -- (.5,2.5);
	\draw (-3,3) -- (3,3);
	\node[anchor=south][Violet] at (0.5,3.25) {$\overbrace{ \hspace{1cm} \cdots \hspace{1cm}  }^{k}$};
\end{tikzpicture}
\caption{Brane realization of the Chern--Simons-matter quiver \eqref{eq:ex3node0}.}
\label{fig:3node0}
\end{figure}
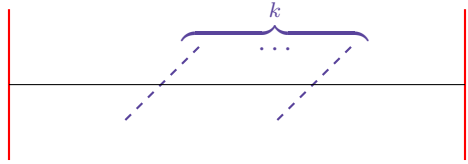

The brane setup is drawn in Figure \ref{fig:3node0}. The maximal branches of the moduli space of \eqref{eq:ex3node0} arise analogously to SQED: One branch corresponds to the motion of the D3-brane suspended between two NS5-branes; the other branch corresponds to the motion of the D3-brane segment suspended between the other type of 5-branes. The only difference from the SQED setup is that the 5-branes are now $(1,\kappa)$-branes, instead of D5-branes, giving rise to a different type of singularity.\par
Moreover, from Figure \ref{fig:3node0}, the A-branch can be partially smoothed misaligning the $(1,\kappa)$-branes, yielding $k$ singularities locally of type $A_{\kappa -1}$. The brane setup cannot be desingularized further, in agreement with the Coulomb branch of \eqref{eq:aux3node}.

\subsection{Abelian three-node quiver with \texorpdfstring{$(\kappa,-\kappa,\kappa)$}{(k,-k,k)}}
Another simple model with $\mN=4$ supersymmetry, first studied in \cite{Jafferis:2008em}, is 
\begin{equation}\label{eq:ex3nodeJY}
		\csnode{1}{\kappa}\text{---} \csnode{1}{-\kappa}\text{---} \csnode{1}{\kappa}\ .
\end{equation}
This theory has $\ct_A \cong \C^{\ast}\cong \ct_B$, thus admitting one FI parameter $\zeta$ and one mass parameter $m$. As usual, $x=e^{-2\pi \zeta}, y=e^{-2\pi m}$.\par
The Hilbert series computation gives \cite{Li:2023ffx}
\begin{equation}
	\hs_{\C [ M_{A,\eqref{eq:ex3nodeJY}} ] } \left( q;x \right) = \frac{1-q^{\kappa+1}}{(1-q)\left(1-xq^{\frac{\kappa+1}{2}}\right)\left(1-x^{-1}q^{\frac{\kappa+1}{2}}\right)} = \hs_{\frac{\C \left[u_1,u_2,z\right]}{\left( u_1 u_2-z^{\kappa+1}\right)}} \left( q;x\right) ,
\end{equation}
and $\hs_{\C [ M_{B,,\eqref{eq:ex3nodeJY}}] } =\hs_{\C [ M_{A,\eqref{eq:ex3nodeJY}}] }\vert_{(q;x) \mapsto (\tilde{q};y)}  $, leading to 
\begin{equation}
	M_{A,\eqref{eq:ex3nodeJY}} \cong \C^2/\Z_{\kappa+1} , \qquad M_{B,\eqref{eq:ex3nodeJY}} \cong \C^2/\Z_{\kappa+1} .
\end{equation}
In this example, no gerbe structure is expected, thus the moduli spaces are identified with the underlying schemes, given by the $A_{\kappa}$ singularity.\par
\medskip
The partition function is manipulated into 
\begin{equation}\label{eq:ZJY}
	\mz_{\text{\eqref{eq:ex3nodeJY}}} = \frac{1}{\sqrt{\kappa}} \int_{-\infty}^{\infty} \dd \sigma \frac{e^{\ii 2 \pi \zeta \sigma }}{\ch (\kappa \sigma) \ch \left(\sigma - m \right)} ,
\end{equation}
in the conventions \eqref{eq:phaseZ} for the overall phase. The derivation is almost identical to \S\ref{app:altcalc}.

\begin{thm}
	Conjecture \ref{myconj1} holds for \eqref{eq:ex3nodeJY}.
\end{thm}
\begin{proof}
The partial smoothing $\widetilde{M}_A \to M_A$ contains an $A_{\kappa-1} $ singularity. The reduced scheme fixed locus $\left( \widetilde{M}_A^{\ct_A} \right)_{\text{red}}$ consists of two points $\{p_{\alpha}\}_{\alpha=1,2}$. $p_2$ is a smooth point in $\widetilde{M}_A$, while $p_1$ is an $A_{\kappa-1} $ singularity. The same holds for the B-branch.\par
In the region of parameters $\zeta \in (0,\infty) , m\in (0,\infty)$, it is possible to close the integration contour in the upper half-plane. The partition function \eqref{eq:ZJY} receives contributions from two towers of poles:
\begin{equation}
	\sigma = \frac{\ii}{\kappa} \left( \frac{1}{2} + n \right) , \qquad \sigma = m  + \ii\left( \frac{1}{2} + n \right) , \qquad \forall n \in \N .
\end{equation}
Taking the residues and rearranging the expressions, one arrives at:
\begin{subequations}
\begin{align}
	\mz_{\text{\eqref{eq:ex3nodeJY}}} & = \frac{1}{\kappa^{\frac{3}{2}}} \mz_{\text{\eqref{eq:ex3nodeJY}}}^{(1)}+ \frac{e^{\ii 2 \pi \zeta m }}{\sqrt{\kappa}} \mz_{\text{\eqref{eq:ex3nodeJY}}}^{(2)} , \\
	\mz_{\text{\eqref{eq:ex3nodeJY}}}^{(1)} &= \sum_{n =0}^{\infty} \left( x^{1/\kappa}\right)^{\frac{1}{2}} \left(-x^{1/\kappa}\right)^{n} \sum_{\nu=0}^{\infty} y^{\frac{1}{2}} (-y)^{\nu} ~e^{\ii \frac{2 \pi}{\kappa} \left( \frac{1}{2} +n \right)\left( \frac{1}{2} +\nu \right) }  \label{eq:ZJYa}\\
	\mz_{\text{\eqref{eq:ex3nodeJY}}}^{(2)} &= \sum_{n=0}^{\infty} x^{\frac{1}{2}} \left(-x\right)^{n} \sum_{\nu=0}^{\infty} y^{\frac{\kappa}{2}} \left(-y^{\kappa}\right)^{\nu}~(-1)^{\kappa (n+\nu)} \ii^{\kappa  } . \label{eq:ZJYb}
\end{align}
\end{subequations}
Note that $\mz_{\text{\eqref{eq:ex3nodeJY}}}^{(1)}$ is obtained from $\mz_{\text{\eqref{eq:ex3nodeJY}}}^{(2)}$ by replacing $\kappa \mapsto 1/\kappa$ followed by $(x,y)\mapsto (x^{1/\kappa},y^{1/\kappa})$. Analogous expressions are obtained in other chambers of the parameter space.\par
This expression is of the form conjectured in \ref{myconj1}, with $\sV_1 = \sQ_0$, $\sV_2=\{3\}$, and $\kappa_{\alpha=1}=\kappa$, $\kappa_{\alpha=2}=1$. $R_A$ and $R_B$ act on the respective Verma modules with eigenvalues $(\kappa +1)n$ and $(\kappa +1)\nu$, whence $\frac{1}{\kappa_{\alpha}}\left( \frac{1}{2} + R_A\right) \otimes \left( \frac{1}{2} + R_B\right) $ has eigenvalues $\frac{1}{\kappa_\alpha}\left( \frac{1}{2} +n \right)\left( \frac{1}{2} +\nu \right) \mod 1$, in perfect agreement with \eqref{eq:ZJYa}-\eqref{eq:ZJYb}.\par
The contribution at $p_2$,
\begin{equation}
	\frac{e^{\ii 2 \pi \zeta m}}{\sqrt{\kappa}} \mz_{\text{\eqref{eq:ex3nodeJY}}}^{(2)} = \frac{\ii^{\kappa}}{\sqrt{\kappa}} e^{\ii 2 \pi \zeta m } \hat{\chi}^{A}\left( e^{-2\pi \zeta} \right)\hat{\chi}^{B}\left( e^{-2\pi \kappa m} \right) ,
\end{equation}
has the factorized form of Conjecture \ref{c:GO}, but the contribution at $p_1$ does not. Nevertheless, it is a trace of the form predicted by Conjecture \ref{myconj1}.
\end{proof}

\subsubsection{Magnetic quiver analysis}

The prescription in \S\ref{sec:magQprime} associates to \eqref{eq:ex3nodeJY} the quiver
\begin{equation}\label{eq:aux3JY}
	\sQ^{\prime}_{\eqref{eq:ex3nodeJY}} \ = \ \fnode{1} \multe{\kappa}  \gnode{1} \multe{1} \fnode{1} ,
\end{equation}
namely $U(1)$ gauge theory coupled to one hypermultiplet of charge $1$ and one hypermultiplet of charge $\kappa$. The partition function of \eqref{eq:aux3JY} is proportional to \eqref{eq:ZJY}:
	\begin{equation}\label{eq:ZSQED2}
		\mz_{\text{\eqref{eq:ex3nodeJY}}}= \frac{1}{\sqrt{\kappa}} \mz_{\text{\eqref{eq:aux3JY}}} ,
	\end{equation}
where, on the right-hand side, the charge-1 hypermultiplet is given mass $m$ and the charge-$\kappa$ one is massless. This can be put in a more symmetric fashion with a shift of variables. The normalization \eqref{eq:phaseZ} avoids a mismatch of overall roots of unity.\par
Furthermore, the Coulomb and Higgs branch Hilbert series for \eqref{eq:aux3JY} match, respectively, the results of \cite{Li:2023ffx} for the A- and B-branch of \eqref{eq:ex3nodeJY}.\par
\begin{thm}
	Conjecture \ref{myconjmag} holds for \eqref{eq:ex3nodeJY}.
\end{thm}
\begin{rmk}
	If $\kappa=1$, the Chern--Simons-matter theory \eqref{eq:ex3nodeJY} is dual to $\gnode{1}\text{---}\fnode{2}$ \cite{Jafferis:2008em}. As first shown in \cite[Sec.4.2]{Gaiotto:2019mmf}, \eqref{eq:ZSQED2} at $\kappa=1$ satisfies Conjecture \ref{c:GO}.
\end{rmk}
\begin{rmk}[Comparison with {\cite{Marino:2025uub}}] The prescription reviewed in \S\ref{sec:QAQB} yields:
\begin{equation}
	\sQ_A \ = \ \gnode{1}\text{---}\fnode{\kappa +1} , \qquad  \sQ_B \ = \ \gnode{1}\text{---}\fnode{\kappa +1}.
\end{equation}
They reproduce the singular schemes, $\mC (\sQ_A) =M_{A,\eqref{eq:ex3nodeJY}}$ and $ \mC (\sQ_B) =M_{B,\eqref{eq:ex3nodeJY}}$, but, contrary to \eqref{eq:aux3JY}, are not suited for studying partial smoothings.
\end{rmk}

\subsubsection{Field theory analysis and brane analysis}
The emergence of the coarse moduli spaces from the vacuum equations, as well as the brane interpretation, are detailed in \cite[Sec.4.3]{Assel:2017eun}. It suffices here to observe that the brane setup only allows the partial desingularization $A_{\kappa} \longrightarrow A_{\kappa +1}$. This agrees with the sphere partition function result and with the magnetic quiver analysis.\par

\subsection{Abelian A-type quiver with alternating \texorpdfstring{$\pm \kappa$}{k}}
\label{sec:exaltpmk}
Consider an Abelian quiver modeled on the A-type Dynkin diagram, with Chern--Simons levels with alternating sign, $\kappa_{v} = (-1)^{v-1} \kappa$. For concreteness, this section deals with the case $\lvert \sQ_0 \rvert=2 \ell$ of an even number of vertices; the case $\lvert \sQ_0 \rvert=2 \ell+1$ is almost identical.\par
The theory of interest is:
\begin{equation}\label{eq:ex4node}
		\csnode{1}{\kappa}\text{---} \csnode{1}{-\kappa}\text{---} \cdots \text{---}  \csnode{1}{\kappa}\text{---} \csnode{1}{-\kappa}\ .
\end{equation}\par
The partition function can be cast in the form:
\begin{align}
	 \mz_{\text{\eqref{eq:ex4node}}} &= \int_{\R^{\ell}} \left( \prod_{v=1}^{\ell}\dd \tilde{\sigma}_v  \frac{e^{\ii 2 \pi \zeta^{v}\tilde{\sigma}_v }}{\ch (\kappa \tilde{\sigma}_v ) }\right)  \prod_{v=1}^{\ell-1}\frac{1}{\ch (\tilde{\sigma}_v-\tilde{\sigma}_{v+1} -m_v )}.  \label{eq:ZaltAb}
\end{align}
This claim is proven in \S\ref{app:altcalc}.\par
\begin{thm}
	Conjecture \ref{myconj1} holds for \eqref{eq:ex4node}.
\end{thm}
\begin{proof}
To reduce clutter, the proof is first spelled out explicitly for $\ell=2$. The argument is later generalized to arbitrary $\ell$.\par
The partial smoothing $\widetilde{M}_A \to M_A$ contains two singular points, which are locally $A_{\kappa-1} $ singularities, while $\left( \widetilde{M}_A^{\ct_A}\right)_{\mathrm{red}}=\{p_{\alpha}\}_{\alpha=1,2,3}$.\par
The partition function specialized to $\ell=2$ reads
\begin{align}
	 \mz_{\text{\eqref{eq:ex4node}}} &= \int_{\R^2} \dd \tilde{\sigma}_1 \dd \tilde{\sigma}_2 \frac{ e^{\ii 2 \pi \kappa  \left( m_1\tilde{\sigma}_1 +m_3\tilde{\sigma}_2\right)} }{\ch (\kappa \tilde{\sigma}_1 ) \ch (\tilde{\sigma}_1-\tilde{\sigma}_2 +m_2 ) \ch (\kappa \tilde{\sigma}_2 )} ,  \label{eq:Z4nodeb}
\end{align}
where $\zeta^{1}=\kappa m_1$ and $\zeta^{2}=\kappa m_3$, cf. \S\ref{app:altcalc}. In the region $(m_1,m_2,m_3) \in (0,\infty)^3$, denoting 
\begin{equation}
	x_1=e^{-2\pi m_1}  \qquad y=e^{-2\pi m_2} , \qquad x_2=e^{-2\pi m_3} ,
\end{equation}
\eqref{eq:Z4nodeb} is evaluated by closing the integration contours in the upper half-plane and picking the residues, and gives:
\begin{subequations}
\begin{align}\label{eq:Zex4nodesum}
	\mz_{\text{\eqref{eq:ex4node}}} & = \frac{1}{\kappa^{2}} \mz_{\text{\eqref{eq:ex4node}}}^{(1)} (x_1,x_2,y)+ \frac{e^{\ii 2 \pi \kappa m_2 m_3 }}{\kappa} \mz_{\text{\eqref{eq:ex4node}}}^{(2)} (x_1,x_2,y) + \frac{e^{-\ii 2 \pi \kappa m_2 m_1 }}{\kappa} \mz_{\text{\eqref{eq:ex4node}}} ^{(2)} (x_2,x_1,y^{-1}), 
\end{align}
\begin{align}
	\mz_{\text{\eqref{eq:ex4node}}}^{(1)} (x_1,x_2,y) &= \sum_{n_1,n_2 =0}^{\infty} (x_1x_2)^{\frac{1}{2}} (-x_1)^{n_1}(-x_2)^{n_2} \sum_{\nu=0}^{\infty} y^{\frac{1}{2}} (-y)^{\nu} ~e^{\ii \frac{2 \pi}{\kappa} \left( n_1-n_2 \right)\left( \frac{1}{2} +\nu \right) }  \label{eq:Zex4a}\\
	\mz_{\text{\eqref{eq:ex4node}}}^{(2)} (x_1,x_2,y)&= \sum_{n_1,n_2 =0}^{\infty} x_1^{\frac{1}{2}}x_2^{\frac{\kappa+1}{2}} (x_1x_2)^{n_1} \left(-(-x_2)^{\kappa}\right)^{n_2} \sum_{\nu=0}^{\infty} y^{\frac{\kappa}{2}} \left(-y\right)^{\kappa\nu}\ii^{\kappa+1 } . \label{eq:Zex4b}
\end{align}
\end{subequations}
The second and third contributions to \eqref{eq:Zex4nodesum} are given by the same function, and are related via $(m_1,m_2,m_3) \mapsto (m_3,-m_2,m_1)$. Expanding in the chamber $m_2>0$, the resulting expressions only differ by exchanging $m_1\leftrightarrow m_3$.\par
For arbitrary $\ell$, the partition function splits into the sum of as many terms as ways of picking poles in \eqref{eq:ZaltAb}. Expanding each summand gives expressions akin to \eqref{eq:Zex4a} or \eqref{eq:Zex4b}. The inverse power of $\kappa$ in front of each summand equals the number of variables $\sigma_v$ picking the poles of $1/\ch (\kappa \sigma_v)$.
\end{proof}

\subsubsection{Magnetic quiver analysis}
To \eqref{eq:ex4node}, the prescription of \S\ref{sec:magQprime} associates the auxiliary quiver 
\begin{equation}\label{eq:aux4node}
	\sQ^{\prime}_{\text{\eqref{eq:ex4node}}} \ = \ \begin{tikzpicture}[baseline=-1] 
	\node (gl) at (-2,0.2) {$\gnode{1}$}; 
	\node (fl) at (-2,-1) {$\underset{1}{\Box}$};  
	\node (gc) at (-1,0.2) {$\gnode{1}$}; 
	\node (fc) at (-1,-1) {$\underset{1}{\Box}$}; 
	\node (d) at (0,0) {$\cdots$}; 
	\node (gr) at (1,0.2) {$\gnode{1}$}; 
	\node (fr) at (1,-1) {$\underset{1}{\Box}$}; 
	\draw (-1.75,0)--(-1.25,0); \draw (-0.75,0)--(-0.25,0); \draw (0.25,0)--(0.75,0);
	\path (gl) edge node[anchor=east] {$\scriptstyle \kappa$} (fl);
	\path (gc) edge node[anchor=east] {$\scriptstyle \kappa$} (fc);
	\path (gr) edge node[anchor=east] {$\scriptstyle \kappa$} (fr);
	\node[anchor=south] (u) at (-0.5,0.5) {$\overbrace{\hspace{3.3cm}}^{\ell}$};
	\end{tikzpicture} .
\end{equation}
The partition function of \eqref{eq:aux4node} equals \eqref{eq:ZaltAb}, supporting Conjecture \ref{myconjmag}.

\subsubsection{Brane analysis}
\begin{figure}
\centering
\begin{tikzpicture}
	\draw[red,thick] (-6,4) -- (-6,2);
	\draw[red,thick] (-2,4) -- (-2,2);
	\draw[red,thick] (6,4) -- (6,2);
	\draw[red,thick] (2,4) -- (2,2);
	\draw[Violet,dashed,thick] (-3.5,3.5) -- (-4.5,2.5);
	\draw[Violet,dashed,thick] (4.5,3.5) -- (3.5,2.5);
	\draw (-6,3) -- (6,3);
	\node at (0,3.5) {$\cdots$};
\end{tikzpicture}
\caption{Brane realization of Chern--Simons-matter quiver with alternating Chern--Simons levels.}
\label{fig:4alt}
\end{figure}
The brane system realizing \eqref{eq:ex4node} is shown in Figure \ref{fig:4alt}. The A-branch is described by the motion of D3-brane segments suspended between two consecutive NS5-branes. The B-branch corresponds to reconnecting the D3-segments, and sliding the resulting D3 between two $(1,\kappa)$-branes.\par
Figure \ref{fig:4alt} shows that, for generic mass parameters, $M_A$ contains $\ell$ singular points, locally $A_{\kappa-1}$ singularities that cannot be resolved further. This is in agreement with the sphere partition function analysis from \eqref{eq:ZaltAb}.

\subsection{Abelian affine A-type quiver: ABJM}
Abelian ABJM theory is modeled on the affine $A_1$ quiver, with gauge group $U(1)_{\kappa} \times U(1)_{-\kappa}$:
\begin{equation}\label{eq:abABJM}
	\begin{tikzpicture}[auto,node distance=1cm,baseline=-1]
		\node (gauge1) at (-1,0) {$\csnode{1}{\kappa}$};
		\node (gauge2) at (1,0) {$\csnode{1}{-\kappa}$};
		\path (gauge1) edge [bend left] node {$\scriptstyle -m$} (gauge2);
		\path (gauge2) edge [bend left] node {$\scriptstyle +m$} (gauge1);
		
		\node[anchor=east] at (-2,0) {Abelian ABJM};
	\end{tikzpicture}
\end{equation}\par
This theory has enhanced $\mN=6$ supersymmetry. The full moduli space has underlying algebraic variety $\C^4/\Z_{\kappa}$; fixing an $\mN=4$ subalgebra of the supersymmetry algebra singles out two branches, with coarse moduli spaces 
\begin{equation}\label{eq:MAMBABJM}
	M_{A,\eqref{eq:abABJM}} \cong \C^2/\Z_{\kappa} , \qquad M_{B,\eqref{eq:abABJM}} \cong \C^2/\Z_{\kappa} .
\end{equation}
ABJM theory has a $\Z_{\kappa}$ 1-form symmetry \cite{Bergman:2020ifi}, whence it is more appropriate to formulate $\mM_{A,\eqref{eq:abABJM}}, \mM_{B,\eqref{eq:abABJM}}$ as moduli stacks, rather than algebraic varieties.\par

The tori acting on these branches are $\ct_A \cong \C^{\ast}$, exactly as in \S\ref{sec:ex2node}, and $\ct_B= (\C^{\ast})^2 / \C^{\ast} \cong \C^{\ast}$, due to the presence of one additional hypermultiplet compared to \S\ref{sec:ex2node}.
The FI parameter is $\zeta$, the masses of the bi-fundamental hypermultiplets are $(+m,-m)$, and it is convenient to define
\begin{equation}
	m_{\pm} :=  m \pm \frac{2 \zeta}{\kappa} .
\end{equation}\par
\medskip
The partition function is given by \cite{Russo:2017qmw}
\begin{equation}\label{eq:ZABJM}
	\mz_{\text{\eqref{eq:abABJM}}} = \frac{1}{\kappa \ch \left(m_{+} \right)\ch \left( m_{-} \right)} ,
\end{equation}
\begin{thm}
	Conjecture \ref{myconj1} holds for Abelian ABJM theory.
\end{thm}
\begin{proof}
It suffices to note that
\begin{equation}
	\mz_{\text{\eqref{eq:abABJM}}} = \mz_{\csnode{1}{\kappa}}  \mz_{\csnode{1}{-\kappa}} \cdot  \hat{\chi}^{A,\eqref{eq:abABJM}} \left( e^{-2\pi m_{+}}\right)\hat{\chi}^{B,\eqref{eq:abABJM}} \left( e^{-2\pi m_{-}}\right) ,
\end{equation}
where the twisted traces associated with the stack $\left[ \C^2/\Z_{\kappa} \right] $ are the same as in \S\ref{sec:ex2node}. 
\end{proof}
Furthermore, \eqref{eq:ZABJM} is consistent with Corollary \ref{cor:magquiver}(\ref{prop:aux2}).

\subsubsection{Magnetic quiver analysis}
The prescription for the auxiliary quiver of \S\ref{sec:magQprime} applied to \eqref{eq:abABJM} gives 
\begin{equation}\label{eq:auxABJM}
	\sQ^{\prime}_{\text{\eqref{eq:abABJM}}} = \begin{tikzpicture}[baseline=-10pt]  \node (f1) at (4,0) {$\fnode{1}$};
		\node (m1) at (2.8,0) {$ \gnode{1}$};
		\path (3,-0.15) edge node[anchor=south] {$\scriptstyle \kappa$} (3.8,-0.15);
		\node (phantom) at (3,-0.15) { \hspace{16pt} };
		\draw (phantom) to[loop left] (phantom);\end{tikzpicture} .
\end{equation}
\begin{itemize}
\item This auxiliary quiver has a Coulomb branch isomorphic, as a quotient stack, to that of $\sQ^{\prime}_{\eqref{eq:ex2node}}$ in \eqref{eq:aux2node}:
\begin{equation}
	\mC \left( \sQ^{\prime}_{\text{\eqref{eq:abABJM}}}  \right) \cong  \mC \left( \sQ^{\prime}_{\eqref{eq:ex2node}}\right) ,
\end{equation}
in agreement with $\mM_{A, \eqref{eq:abABJM}}$.
\item The adjoint hypermultiplet of \eqref{eq:auxABJM} parameterizes a $\C^2$. However, there is an unbroken $\Z_{\kappa}$ gauge symmetry at every point, so the Higgs branch of \eqref{eq:auxABJM} is expected to be $\left[ \C^2 /\Z_{\kappa} \right]$, in agreement with $\mM_{B, \eqref{eq:abABJM}}$.
\end{itemize}\par
The partition functions are also equal, 
\begin{equation}
	\mz_{\text{\eqref{eq:abABJM}}} \left( m_+, m_- \right)= \left. \mz_{\sQ^{\prime}_{\text{\eqref{eq:abABJM}}} }  \left( \zeta^{\prime}, m^{\prime} \right)\right\rvert_{\zeta^{\prime}=\kappa m_+ , m^{\prime}=m_-} .
\end{equation}
\begin{thm}
	Conjecture \ref{myconjmag} holds for Abelian ABJM theory.
\end{thm}

\begin{rmk}[Comparison with {\cite{Marino:2025uub}}] Both magnetic quivers $\sQ_A$ and $\sQ_B$ prescribed by \S\ref{sec:QAQB} are given by the Abelian ADHM with $\kappa$ hypermultiplets, 
\begin{equation}
	\sQ_{A,\text{\eqref{eq:abABJM}}} =\  \begin{tikzpicture}[baseline=-8pt]\node (f1) at (4,0) {$\fnode{\kappa}$};
		\node (m1) at (2.8,0) {$ \gnode{1}$};
		\draw (3,-0.15) -- (3.8,-0.15);
		\node (phantom) at (3,-0.15) { \hspace{16pt} };
		\draw (phantom) to[loop left] (phantom);\end{tikzpicture} \ = \sQ_{B,\text{\eqref{eq:abABJM}}} .
\end{equation}
As in the example of \S\ref{sec:ex2node}, they reproduce the coarse moduli spaces \eqref{eq:MAMBABJM} as singular schemes, but the resolution parameters differ.
\end{rmk}

\subsubsection{Field theory analysis and brane analysis}
The analysis of the vacuum equations in field theory goes along the same lines of \S\ref{sec:QFTex2node}.
ABJM theory has been largely studied in the literature, and the field theoretic and brane analyses can be found in \cite[Sec.6.1]{Assel:2017eun}. For the present discussion, it suffices to observe that it is manifest in the brane setup that the $\Z_{\kappa}$-orbifold singularity cannot be resolved.

\subsection{Abelian three-node affine A-type quiver}
The next Abelian Chern--Simons-matter theory based on an affine A-type Dynkin diagram is:
\begin{equation}\label{eq:circular3}
	\begin{tikzpicture}[baseline=0]
			\node (a) at (1,0) {$\csnode{1}{\kappa}$};
			\node (b) at (0,1) {$\csnode{1}{0}$};
			\node (c) at (-1,0) {$\csnode{1}{-\kappa}$};
			\draw (c) -- (a);
			\draw (a) -- (b);
			\draw (b) -- (c);
			\node[anchor=west,Violet] (n1) at (1.1,0) {$\scriptstyle v=1$};
			\node[anchor=west,Violet] (n2) at (0.1,1) {$\scriptstyle v=2$};			
			\node[anchor=east,Violet] (n3) at (-1.1,0) {$\scriptstyle v=3$};
		\end{tikzpicture}
\end{equation}
with the colored label indicating the conventions for the vertices. The analysis of the A- and B-branch is deferred to \S\ref{sec:QFTcircular3}-\S\ref{sec:circular3brane}. It suffices to anticipate here that the A- and B-branches are analogous to the specialization $k=2$ of the example \S\ref{sec:exk00k}, but in addition there is an overall, factorized component due to the presence of one additional hypermultiplet compared to \S\ref{sec:exk00k}.\par
The parameters of \eqref{eq:circular3} are: One FI parameter $\zeta$ from the sub-quiver \eqref{eq:ex2node} of \S\ref{sec:ex2node}; one FI parameter $\zeta^{2}$, from the topological symmetry of the node $v=2$; one mass parameter from the anti-diagonal combination of the $\C^{\ast}$-actions scaling the edges $2\to 1$ and $2\to 3$.\par
\medskip
Up to a shift of variables, the partition function is
\begin{equation}
	\mz_{\eqref{eq:circular3}} = \int_{\R^3} \dd \sigma_1 \dd \sigma_2 \dd \sigma_3 \frac{ e^{\ii 2 \pi \zeta^2 \sigma_2 + \ii \pi \kappa \left( \sigma_1^2 - \sigma_3^2 \right) } }{\ch (\sigma_2 - \sigma_1 + m_1)\ch (\sigma_2 - \sigma_3 - m_1) \ch (\sigma_1-\sigma_3 +\zeta)} .
\end{equation}
A computation analogous to those in \S\ref{app:evaluation} gives:
\begin{equation}\label{eq:Zcirc3}
	\mz_{\eqref{eq:circular3}} = \frac{1}{\kappa} \frac{1}{\ch (m_+)}\int_{\R} \dd \sigma \frac{ e^{\ii 2 \pi \zeta^{\prime} \sigma } }{\ch (\sigma+ m^{\prime})\ch (\sigma_1 - m^{\prime})} ,
\end{equation}
where 
\begin{equation}
	m_+ := \zeta + \frac{\zeta^2}{\kappa} , \qquad \zeta^{\prime} := 2m_1 -\frac{\zeta^2}{\kappa}  \qquad m^{\prime} := \frac{\zeta^2}{2} .
\end{equation}
\begin{thm}
	Conjecture \ref{myconj1} holds for \eqref{eq:circular3}.
\end{thm}

\subsubsection{Magnetic quiver analysis}
The auxiliary quiver for \eqref{eq:circular3} is 
\begin{equation}\label{eq:auxcirc3}
	\sQ^{\prime}_{\text{\eqref{eq:circular3}}} = \begin{tikzpicture}[baseline=-10pt]  \node (f1) at (4,0) {$\fnode{2}$};
		\node (m1) at (2.8,0) {$ \gnode{1}$};
		\path (3,-0.15) edge node[anchor=south] {$\scriptstyle \kappa$} (3.8,-0.15);
		\node (phantom) at (3,-0.15) { \hspace{16pt} };
		\draw (phantom) to[loop left] (phantom);\end{tikzpicture}
\end{equation}
The number of FI and mass parameters in \eqref{eq:auxcirc3} matches exactly with \eqref{eq:circular3}. The adjoint hypermultiplet decouples, thus the moduli space of vacua of \eqref{eq:auxcirc3} will be a global product with a $\C^2$ parameterized by this field.
\begin{itemize}
\item The Coulomb branch is $\mC \left( \sQ^{\prime}_{\text{\eqref{eq:circular3}}} \right) =\left[ \C^2 / \Z_{2\kappa} \right]$, and can be partially resolved into two locally $A_{\kappa-1}$ singularities.
\item Higgs branch $\mH \left( \sQ^{\prime}_{\text{\eqref{eq:circular3}}} \right)$ is a $\Z_{\kappa}$-gerbe over $\mH \left( \gnode{1} \text{---} \fnode{2}\right)= \C^2/\Z_2$.
\end{itemize}
It is shown in \S\ref{sec:QFTcircular3} that the underlying algebraic varieties match $M_A,M_B$.\par
The partition function of \eqref{eq:auxcirc3} equals \eqref{eq:Zcirc3}, identifying $m_+$ with the mass of the adjoint hypermultiplet. The result is consistent with Corollary \ref{cor:magquiver}(\ref{prop:aux2}), where for this case $\sQred_{\text{\eqref{eq:circular3}}}$ is the ADHM quiver with two fundamental hypermultiplets.
\begin{thm}
	Conjecture \ref{myconjmag} holds for \eqref{eq:circular3}.
\end{thm}

\subsubsection{Field theory analysis}
\label{sec:QFTcircular3}
\begin{table}
\centering
\begin{tabular}{|l|c c|c|}
	\hline
		 & $q_{e}$ & $\tilde{q}_{e}$ & $V_{(\lambda_1,\lambda_2,\lambda_3)}$ \\
		\hline
		$U(1)_{\kappa}$ & $\delta_{e, 1 \to 2} - \delta_{e, 3 \to 1}$ & $-\delta_{e, 1 \to 2}+\delta_{e, 3 \to 1}$ & $-\kappa \lambda_1$ \\
		$U(1)_{0}$ & $\delta_{e, 2 \to 3}-\delta_{e, 1 \to 2}$ & $- \delta_{e, 2 \to 3}+\delta_{e, 1 \to 2}$ & $0$ \\
		$U(1)_{-\kappa}$ & $\delta_{e, 3 \to 1} -\delta_{e, 2 \to 3}$ & $-\delta_{e, 3 \to 1}+\delta_{e, 2 \to 3}$ & $\kappa \lambda_3$ \\
		\hline
	\end{tabular}
\caption{Charges of the chiral multiplets and bare monopole operators in the three-node Chern--Simons-matter quiver \eqref{eq:circular3}.}
\label{tab:ops3nodecirc}
\end{table}

The field theoretic study of $M_A, M_B$ is in part similar to \cite[Sec.4.2]{Assel:2017eun}. To lighten the exposition, only the complex moment map is discussed; the real moment map is dealt with in exactly the same way.\par
The scalars in the chiral multiplets corresponding to the edges in \eqref{eq:circular3} are denoted $(q_e,\tilde{q}_e)$, the complex scalars in the vector multiplet $v \in \{1,2,3\}$ are denoted $\varphi_v$; finally, $V_{\lambda}$ denotes the bare monopole operator with magnetic charge $\lambda =(\lambda_1, \lambda_2, \lambda_3)$. The gauge charges are summarized in Table \ref{tab:ops3nodecirc}.\par
The complex moment map equations read:
\begin{subequations}\label{eq:monentcircular3}
\begin{align}
	q_{1\to 2} \tilde{q}_{1\to 2} - q_{3\to 1} \tilde{q}_{3\to 1} &= \kappa \varphi_1, \\
	q_{2\to 3} \tilde{q}_{2\to 3} - q_{1\to 2} \tilde{q}_{1\to 2} &= 0 \\
	q_{3\to 1} \tilde{q}_{3\to 1} - q_{2\to 3} \tilde{q}_{2\to 3} &= - \kappa \varphi_3 ,
\end{align}
\end{subequations}
supplemented by the critical locus of the canonical superpotential, 
\begin{equation}\label{eq:dWcircular3}
	q_e \left( \varphi_{\mathsf{h}(e)} - \varphi_{\mathsf{t}(e)}\right) = 0 , \qquad \tilde{q}_e \left( \varphi_{\mathsf{t}(e)} - \varphi_{\mathsf{h}(e)}\right) = 0 ,
\end{equation}
for all $e\in \sQ_1$. The discussion of monopole operators is very similar to \cite[Sec.4.2]{Assel:2017eun}, and one can check that 
\begin{equation}
	u_{\pm} := V_{(0, \pm 1 , 0)} , \qquad u_1 := V_{(1,1,1)} (q_{1\to 2}q_{2\to 3} )^{\kappa} , \qquad u_2 := V_{(-1,-1,-1)} (\tilde{q}_{1\to 2}\tilde{q}_{2\to 3} )^{\kappa} 
\end{equation}
constitute a set of gauge-invariant generators. There are additional constraints (from \cite[Eq.(2.8)]{Assel:2017eun})
\begin{subequations}\label{eq:monoeqcirc3}
\begin{align}
	u_{\pm} q_{1\to 2} &=0=u_{\pm} \tilde{q}_{1\to 2} \\
	u_{\pm} q_{2\to 3} &=0=u_{\pm} \tilde{q}_{2\to 3} .
\end{align}
\end{subequations}\par
It is convenient to define the shorthand notation 
\begin{equation}
	z := q_{1\to 2} \tilde{q}_{1\to 2} = q_{2\to 3} \tilde{q}_{2\to 3} , \qquad w := q_{3\to 1} \tilde{q}_{3\to 1} .
\end{equation}
Then \eqref{eq:monentcircular3} implies
\begin{equation}
	\kappa \varphi_1 = \kappa \varphi_3  = z-w .
\end{equation}
This condition inserted in \eqref{eq:dWcircular3} automatically solves the equations involving $q_{3\to 1},\tilde{q}_{3\to 1}$, leaving them unconstrained. Since $q_{3\to 1},\tilde{q}_{3\to 1}$ do not appear in \eqref{eq:monoeqcirc3}, the moduli space of vacua is a global product with a $\C^2$ parameterized by these variables. For the rest, there are two solutions.
\begin{itemize}
\item One solution has $u_{\pm}=0$, and is parametrized by $u_1, u_2, z $, where the remaining equations of \eqref{eq:dWcircular3} are solved setting $\kappa \varphi_2= z$. Observe that there is one relation, 
	\begin{equation}
		u_1 u_2 = z^{2\kappa} .
	\end{equation}
	This relation can be partly resolved by turning on the parameter $m_1$, which modifies it into 
	\begin{equation}
		u_1 u_2 = \left( z^{\kappa} -m_1\right)\left( z^{\kappa} +m_1\right).
	\end{equation}
	Therefore 
	\begin{equation}
		\C \left[ M_{A,\eqref{eq:circular3}} \right] = \C \left[ u_1,u_2,z \right] / \left( u_1 u_2 - z^{2\kappa} \right) \quad \Longrightarrow \quad M_{A,\eqref{eq:circular3}} = \C^2/\Z_{2\kappa} ,
	\end{equation}
	and it only admits a partial smoothing, with two locally $A_{\kappa-1}$ singularities.
\item Alternatively, one solves \eqref{eq:monoeqcirc3} by 
	\begin{equation}
		q_{1\to 2}=\tilde{q}_{1\to 2}=q_{2\to 3}=\tilde{q}_{2\to 3} =0 ,
	\end{equation}
	which forces $u_1=0=u_2$ and $z=0$. \eqref{eq:dWcircular3} is then satisfied, leaving $u_{\pm}, \varphi_2$ as generators. This branch is identical to the one analyzed in \cite[Sec.4.2]{Assel:2017eun}, where it is argued that there is one relation \cite[Eq.(4.8)]{Assel:2017eun}
	\begin{equation}
		u_+u_-= \varphi_2^2 ,
	\end{equation}
	whence $M_{B,\eqref{eq:circular3}}=\C^2/\Z_2$.
\end{itemize}

\subsubsection{Brane analysis}
\label{sec:circular3brane}
There are two ways to realize \eqref{eq:circular3} using brane configurations in Type IIB string theory. One uses a circular D3-brane intersecting one NS5-brane and two $(1,\kappa)$-branes; alternatively, the circular D3-brane intersects two NS5-branes and one $(1,\kappa)$-brane. The two configurations differ in the exchange of A-branch and B-branch. Here the former setup is considered.\par
The A-branch is one-quaternionic dimensional, associated with the sliding of the D3-brane along the NS5-brane. The location at which the D3-brane intersects the two $(1,\kappa)$-branes produces a singularity $\C^2 / \Z_{2\kappa}$, which can be partly smoothed by misaligning the two $(1,\kappa)$-branes, producing two singular points of type $A_{\kappa-1}$.\par
One component of the B-branch is analogous to the Higgs branch of SQED, arising from sliding a D3-brane segment suspended between the two $(1,\kappa)$-branes. Additionally, there is a universal factor for circular configurations, analogous to ABJM theory (see \cite{Assel:2017eun}).

\subsection{Abelian affine A-type quiver with \texorpdfstring{($\kappa, 0, \dots, 0,-\kappa,0$)}{(k,0,...,0,k,0)}}
Consider next the Abelian Chern--Simons-matter theory based on an affine A-type Dynkin diagram:
\begin{equation}\label{eq:circular4}
	\begin{tikzpicture}[baseline=0]
			\node (a) at (-2,0) {$\csnode{1}{\kappa}$};
			\node (b) at (-1,0) {$\csnode{1}{0}$};
			\node (o) at (0,1) {$\csnode{1}{0}$};
			\node (e) at (0,0) {$\cdots$};
			\node (c) at (1,0) {$\csnode{1}{0}$};
			\node (d) at (2,0) {$\csnode{1}{-\kappa}$};
			\draw (b) -- (a);
			\draw (e) -- (b);
			\draw (e) -- (c);
			\draw (d) -- (c);
			\draw (o) -- (a);
			\draw (o) -- (d);
			\node[Violet] (n1) at (a.south) {$\scriptstyle v=0$};
			\node[Violet] (n2) at (b.south) {$\scriptstyle v=1$};	
			\node[Violet] (n3) at (c.south) {$\scriptstyle v=k-1$};
			\node[Violet] (n4) at (d.south) {$\scriptstyle v=k$};	
			\node[anchor=west, Violet] (n0) at (0.1,1) {$\scriptstyle v=\infty$};
		\end{tikzpicture}
\end{equation}
The colored label indicates the conventions for the vertices. This is an affine version of the quiver \eqref{eq:ex3node0} from \S\ref{sec:exk00k}, obtained by adding one node with vanishing Chern--Simons levels and two edges to close the necklace.\par
The discussion of A- and B-branches is deferred to \S\ref{sec:circular4brane}, with the aid of the brane setup. The sub-quiver \eqref{eq:ex3node0} provides $k$ independent FI parameters, parametrized as in \S\ref{sec:exk00k} by $\{\zeta^{v} \}_{v=0}^{k}$ subject to $\sum_{v=0}^{k} \zeta^{v}=0$, and one mass parameter $m$ which is associated for definiteness with the right-most edge. In addition, \eqref{eq:circular4} has one FI parameter $\zeta^{\infty}$ from the additional node, and one mass parameter $m_{\infty}$ such that the hypermultiplets corresponding to the two edges $\infty \to 0$ and $\infty \to k$ have masses $\pm m_{\infty}$.\par
\medskip
It is shown in \S\ref{app:4nodecirc} that the partition function of \eqref{eq:circular4} equals
\begin{equation}
\label{eq:Zcirc4}
	\mz_{\eqref{eq:circular4}} = \mz_{\csnode{1}{\kappa}} \mz_{\csnode{1}{-\kappa}}\mz_{\gnode{1} \text{---} \fnode{2} } (\zeta^{\prime} , m^{\prime})  \mz_{\gnode{1} \text{---} \fnode{k}} (\zeta^{\prime\prime}, m^{\prime\prime}) 
\end{equation}
with, on the right-hand side, parameters given by:
\begin{equation}
\label{eq:circ4param}
	 \zeta^{\prime}= 2m_{\infty}, \quad \zeta^{\prime \prime} =m, \quad m^{\prime}=\zeta^{\infty}, \quad m^{\prime\prime}_i = \sum_{v=i}^{k-1} \zeta^{v} .
\end{equation}
\begin{thm}
	Conjecture \ref{myconj1} holds for \eqref{eq:circular4}.
\end{thm}
\begin{proof}
	It follows immediately from \eqref{eq:Zcirc4} and the fact that SQED satisfies Conjecture \ref{c:GO}. In this case, not only the twisted trace on $\scH_{\alpha}^{A} \otimes \scH_{\alpha}^{B}$ given in \eqref{eq:conj1Z} factorizes into the product of twisted traces $\hat{\chi}_{\alpha}^{A} \hat{\chi}_{\alpha}^{B}$, but moreover each twisted trace itself factorizes:
	\begin{subequations}
	\begin{align}
		\hat{\chi}_{\alpha}^{A,\eqref{eq:circular4}} (e^{-2\pi \zeta}) &= \hat{\chi}_{\alpha}^{\widetilde{\C^2 / \Z_2}} (e^{-2\pi \zeta^{\prime}}) \hat{\chi}_{\alpha}^{\widetilde{\C^2 / \Z_k}} (e^{-2\pi \zeta^{\prime\prime}}) , \\ 
		\hat{\chi}_{\alpha}^{B,\eqref{eq:circular4}} (e^{-2\pi m}) &= \hat{\chi}_{\alpha}^{\widetilde{\C^2 / \Z_2}} (e^{-2\pi m^{\prime}}) \hat{\chi}_{\alpha}^{T^{\ast}P^{k-1}} (e^{-2\pi m^{\prime\prime}}) .
	\end{align}
	\end{subequations}
\end{proof}

\subsubsection{Magnetic quiver analysis}
The auxiliary quiver for \eqref{eq:circular4} is 
\begin{equation}\label{eq:auxcirc4}
\sQ^{\prime}_{\text{\eqref{eq:circular4}}} = \begin{tikzpicture}[baseline=0]  \node (f1) at (2,0) {$\underset{ \ }{\fnode{k}}$};
	\node (g1) at (1,0) {$ \csnode{1}{}$};
	\node (g2) at (-0.5,0) {$ \csnode{1}{}$};
	\path (g1) edge node[anchor=south] {$\scriptstyle \kappa$} (f1);
	\path (g1) edge[bend left] (g2);
	\path (g1) edge[bend right] (g2);
	\end{tikzpicture}
\end{equation}
This theory has two FI parameters, say $\eta^{\prime,1}, \eta^{\prime,2}$; $(k-1)$ independent mass parameters for the charge-$\kappa$ hypermultiplets of charge $\kappa$, and one mass parameter $m^{\prime}$ for the bi-fundamental hypermultiplets.
\begin{itemize}
\item The Coulomb branch contains locally an $A_{k\kappa-1}$ singularity, which can be partially resolved into $k$ locally $A_{\kappa-1}$ singularities with resolution parameters $m^{\prime\prime}_i$. Moreover, even after Higgsing the gauge node connected to the framing, there is an $A_1$ singularity which can be smoothed by turning on the mass parameters $m^{\prime}_2$.
\item The Higgs branch factorizes into the Higgs branch of SQED with charge-$\kappa$ hypermultiplets, and the contribution by the bi-fundamental hypermultiplets.
\end{itemize}\par
The partition function of \eqref{eq:auxcirc4} is:
\begin{subequations}
\begin{align}
	\mz_{\eqref{eq:auxcirc4}} &= \int_{\R^2} \frac{ \dd \sigma^{\prime}_1 \dd \sigma^{\prime}_2 ~e^{\ii 2 \pi \sum_{v=1}^{2}\eta^{\prime,v} \sigma^{\prime}_v } }{\ch \left(\sigma^{\prime}_2 - \sigma^{\prime}_1 + m^{\prime} \right)\ch \left(\sigma^{\prime}_2 - \sigma^{\prime}_1  \right)} \prod_{i=1}^{k} \frac{1}{ \ch \left(\kappa\sigma^{\prime}_1 + m^{\prime\prime}_i \right)} \\
		&= \frac{1}{\kappa} \int_{-\infty}^{+\infty} \dd \sigma_1 e^{\ii 2 \pi \left( \frac{\eta^{\prime ,1} +\eta^{\prime, 2}}{\kappa} \right) \sigma_1} \prod_{i=1}^{k}\frac{ 1}{\ch \left(\sigma_1 + m^{\prime\prime}_i \right)}  \int_{-\infty}^{+\infty} \dd \sigma_2 \frac{e^{\ii 2 \pi \eta^{\prime, 2} \sigma_2} }{\ch \left(\sigma_2 + m^{\prime} \right)\ch \left(\sigma_2 \right)} ,
\end{align}
\end{subequations}
where the conventions here are chosen for comparison with \S\ref{app:4nodecirc}; more symmetric assignments of masses would differ just by an overall factor $e^{\ii 2 \pi \eta^{\prime,1} m^{\prime}}$.
The second line stems from the change of variables $\sigma_2 = \sigma_2^{\prime}-\sigma^{\prime}_1 $ followed by $\sigma_1= \kappa \sigma_1^{\prime}$. This equals \eqref{eq:Zcirc4} with identification of the FI parameters 
\begin{equation}
	\zeta^{\prime\prime} =  \frac{\eta^{\prime}_1+\eta^{\prime}_2}{\kappa} , \qquad \zeta^{\prime} = \eta^{\prime}_2 .
\end{equation}\par
\begin{thm}
	Conjecture \ref{myconjmag} holds for \eqref{eq:circular4}.
\end{thm}
The result agrees with Corollary \ref{cor:magquiver}(\ref{prop:aux2}), where now the quiver $\sQred$ is disconnected and given by the union of $\gnode{1} \text{---} \fnode{2}$ and $\gnode{1} \text{---} \fnode{k}$.

\subsubsection{Brane analysis}
\label{sec:circular4brane}
The theory \eqref{eq:circular4} is realized by a circular D3-brane intersecting two consecutive NS5-branes and $k$ consecutive $(1,\kappa)$-branes. The A-brach is two-quaternionic dimensional, associated with the sliding of the two D3-brane segments on either side of the interval between consecutive NS5-branes. One of the two sides has a singularity due to the presence of the $(1,\kappa)$-branes, which can be partially smoothed into $k$ singular points by misaligning them.\par
One can perform a partial Higgsing by sliding the D3-brane segments suspended between pairs of $(1,\kappa)$-branes. The motion of the remaining D3-brane is analogous to the Coulomb branch of SQED with two fundamental hypermultiplets.

\subsection{Abelian affine A-type quiver with alternating \texorpdfstring{$\pm \kappa$}{k}}
The next example is the affine version of the case studied in \S\ref{sec:exaltpmk}. It consists of an affine A-type Dynkin diagram with a total of $\lvert \sQ_0\rvert =2\ell$ nodes, and with Chern--Simons levels with alternating sign, $\kappa_{v} = (-1)^{v-1} \kappa$:
\begin{equation}\label{eq:circalt}
		\begin{tikzpicture}[baseline=0]
			\node (a) at (-2,0) {$\csnode{1}{\kappa}$};
			\node (b) at (-1,0) {$\csnode{1}{-\kappa}$};
			\node (o) at (0,1) {$\csnode{1}{-\kappa}$};
			\node (e) at (0,0) {$\cdots$};
			\node (c) at (1,0) {$\csnode{1}{-\kappa}$};
			\node (d) at (2,0) {$\csnode{1}{\kappa}$};
			\draw (b) -- (a);
			\draw (e) -- (b);
			\draw (e) -- (c);
			\draw (d) -- (c);
			\draw (o) -- (a);
			\draw (o) -- (d);
			\node[Violet] (n1) at (a.south) {$\scriptstyle v=1$};
			\node[Violet] (n2) at (b.south) {$\scriptstyle v=2$};	
			\node[Violet] (n4) at (d.south) {$\scriptstyle v=2\ell -1$};	
			\node[anchor=west, Violet] (n0) at (0.1,1) {$\scriptstyle v=2\ell$};
		\end{tikzpicture}
\end{equation}\par
A computation almost identical to that in \S\ref{app:altcalc} shows that the partition function equals:
\begin{align} \label{eq:Zcircalt}
	 \mz_{\eqref{eq:circalt}} &= \int_{\R^{\ell}} \left( \prod_{v=1}^{\ell}\dd \tilde{\sigma}_v  \frac{e^{\ii 2 \pi \zeta^{v}\tilde{\sigma}_v }}{\ch (\kappa \tilde{\sigma}_v ) }\right) \frac{1}{\ch (\tilde{\sigma}_{2\ell}-\tilde{\sigma}_{1} -m_{2\ell} )} \prod_{v=1}^{\ell-1}\frac{1}{\ch (\tilde{\sigma}_v-\tilde{\sigma}_{v+1} -m_v )}. 
\end{align}
\begin{thm}
	Conjecture \ref{myconj1} holds for \eqref{eq:circalt}.
\end{thm}
\begin{proof}Analogous to \S\ref{sec:exaltpmk}.
\end{proof}

\subsubsection{Magnetic quiver analysis}
To \eqref{eq:circalt}, the prescription of \S\ref{sec:magQprime} associates the auxiliary affine A-type quiver 
\begin{equation}\label{eq:auxcircalt}
	\sQ^{\prime}_{\text{\eqref{eq:circalt}}} \ = \ \begin{tikzpicture}[baseline=-1] 
	\node (gl) at (-2,0.2) {$\gnode{1}$}; 
	\node (fl) at (-2,-1) {$\underset{1}{\Box}$};  
	\node (gc) at (-1,0.2) {$\gnode{1}$}; 
	\node (fc) at (-1,-1) {$\underset{1}{\Box}$}; 
	\node (d) at (0,0) {$\cdots$}; 
	\node (gr) at (1,0.2) {$\gnode{1}$}; 
	\node (fr) at (1,-1) {$\underset{1}{\Box}$}; 
	\draw (-1.75,0)--(-1.25,0); \draw (-0.75,0)--(-0.25,0); \draw (0.25,0)--(0.75,0);
	\path (gl) edge node[anchor=east] {$\scriptstyle \kappa$} (fl);
	\path (gc) edge node[anchor=east] {$\scriptstyle \kappa$} (fc);
	\path (gr) edge node[anchor=east] {$\scriptstyle \kappa$} (fr);
	\path (gr) edge[bend right] (gl);
	\node[anchor=south] (u) at (-0.5,0.75) {$\overbrace{\hspace{3.3cm}}^{\ell}$};
	\end{tikzpicture} .
\end{equation}
The partition function of \eqref{eq:auxcircalt} equals \eqref{eq:Zcircalt}, supporting Conjecture \ref{myconjmag}.

\subsection{Abelian four-node quiver not of Dynkin type}
\label{sec:nonlin1}

\begin{table}
\centering
\begin{tabular}{|c|c|c|c|c|c|}
\hline
$e$ label & 0 & 1 & 2 & 3 & 4 \\
\hline
$e$ in \eqref{eq:nonlinear1} & $\csnode{1}{\kappa} \text{---} \csnode{1}{-\kappa}$ & $\csnode{1}{\kappa} \text{---} \csnode{1}{0}$ & $\csnode{1}{0} \text{---} \csnode{1}{-\kappa}$  & $\csnode{1}{-\kappa} \text{---} \csnode{1}{0}$ & $\csnode{1}{0} \text{---} \csnode{1}{\kappa}$ \\
$\mathsf{t}(e) \to \mathsf{h}(e)$ & $1 \to 3$ & $1 \to 2$ & $2 \to 3$ & $3 \to 4$ & $4 \to 1$ \\
\hline
\end{tabular}
\caption{Labelling of the edges in the quiver \eqref{eq:nonlinear1}.}
\label{tab:nonlinedge}
\end{table}

To test Conjecture \ref{myconj1} beyond A-type quivers, consider the Abelian theory
\begin{equation}\label{eq:nonlinear1}
	\begin{tikzpicture}[baseline=0]
			\node (a) at (1,0) {$\csnode{1}{\kappa}$};
			\node (b) at (0,1) {$\csnode{1}{0}$};
			\node (c) at (-1,0) {$\csnode{1}{-\kappa}$};
			\node (d) at (0,-1) {$\csnode{1}{0}$};
			\draw (c) -- (a);
			\draw (a) -- (b);
			\draw (b) -- (c);
			\draw (c) -- (d);
			\draw (d) -- (a);
			
			\node[anchor=east,Violet] (nl) at (-1.1,0) {$\scriptstyle v=3$};
			\node[anchor=west,Violet] (nr) at (1.1,0) {$\scriptstyle v=1$};
			\node[anchor=west,Violet] (u) at (.1,1) {$\scriptstyle v=2$};
			\node[anchor=west,Violet] (d) at (.1,-1) {$\scriptstyle v=4$};
		\end{tikzpicture}
\end{equation}
The vertices are labeled according to the colored index next to them, and the edges are labeled according to Table \ref{tab:nonlinedge}.\par
The theory has $\ct_A \cong \C^{\ast}$, given by the anti-diagonal combination of the topological symmetries from the two $U(1)$ gauge nodes without Chern--Simons couplings; the corresponding FI parameters are $(+\zeta, -\zeta)$. Additionally, $\ct_B \cong \C^{\ast}_0 \times \C^{\ast}_2\times \C^{\ast}_4$, where $\C^{\ast}_0$ acts by scaling $e=0$, $\C^{\ast}_2$ acts by scaling $e=1$ and $e=2$, and $\C^{\ast}_4$ acts by scaling $e=3$ and $e=4$.\par
\begin{thm}
	Conjecture \ref{myconj1} holds for \eqref{eq:nonlinear1}.
\end{thm}
\begin{proof}
Assigning redundant masses $\{m_e\}_{e\in \sQ_1}$, the partition function is 
{\small \begin{equation}
	\mz_{\eqref{eq:nonlinear1}} = \int_{\R^4}  \frac{\dd \sigma_1 \dd \sigma_2\dd \sigma_3\dd \sigma_4 ~e^{\ii \pi \kappa \left(\sigma_1^2 -\sigma_3^2 \right) + \ii 2 \pi \zeta (\sigma_2 - \sigma_4) }}{\ch (\sigma_3-\sigma_1+m_0) \ch (\sigma_2-\sigma_1+m_1) \ch (\sigma_3-\sigma_2+m_2)  \ch (\sigma_4-\sigma_3+m_3) \ch (\sigma_1-\sigma_4+m_4)} .
\end{equation}}
A computation entirely analogous to those in \S\ref{app:evaluation} shows that
\begin{equation}
	\mz_{\eqref{eq:nonlinear1}} =  \mz_{\csnode{1}{\kappa}} \mz_{\csnode{1}{-\kappa}} \cdot \left. \frac{e^{\ii 2 \pi \zeta (m_4-m_2)}}{\ch (m_0)} \mz_{\gnode{1}\text{---}\fnode{2}} (\zeta^{\prime}_1 , m^{\prime}) \mz_{\gnode{1}\text{---}\fnode{2}} (\zeta^{\prime}_3 , -m^{\prime}) \right\rvert_{\scriptscriptstyle \zeta^{\prime}_1 = m_1+m_2, \zeta^{\prime}_3 = m_3+m_4, m^{\prime} = \zeta} .
\end{equation}
Thus, the partition function is a product of pure Chern--Simons partition functions, multiplied by the partition function of a disconnected quiver. The latter factorizes into the product of three partition functions, one for each connected component of the quiver, and each component satisfies Conjecture \ref{c:GO}.
\end{proof}

\subsection{Abelian quiver not of Dynkin type with \texorpdfstring{($\kappa, 0, \dots, 0,-\kappa$)}{(k,0,...,0,k)}}
\label{sec:nonlin2}
Another example of a Chern--Simons-matter quiver not of Dynkin type is:
\begin{equation}\label{eq:nonlinear2}
	\begin{tikzpicture}[baseline=0]
			\node (t) at (0,1.25) {$\csnode{1}{-\kappa}$};
			\node (c2) at (1.5,0) {$\csnode{1}{0}$};
			\node (c1) at (-1.5,0) {$\csnode{1}{0}$};
			\node (cr) at (3,0) {$\csnode{1}{0}$};
			\node (cl) at (-3,0) {$\csnode{1}{0}$};
			\node (b) at (0,-1.25) {$\csnode{1}{\kappa}$};
			\node (d) at (0,0) {$\cdots$};
			\draw (c1) -- (t);
			\draw (c1) -- (b);
			\draw (c2) -- (t);
			\draw (c2) -- (b);
			\draw (cl) -- (t);
			\draw (cl) -- (b);
			\draw (cr) -- (t);
			\draw (cr) -- (b);
			
			\node[anchor=south west,Violet] (nt) at (0.1,1.25) {$\scriptstyle v=\tilde{0}$};
			\node[anchor=north west,Violet] (nb) at (0.1,-1.25) {$\scriptstyle v=0$};			
			\node[anchor=east,Violet] (nl) at (-3.1,0) {$\scriptstyle v=1$};
			\node[anchor=west,Violet] (nr) at (3.1,0) {$\scriptstyle v=2\ell$};
			\node[anchor=west,Violet] (nl) at (-1.4,0) {$\scriptstyle v=2$};
		\end{tikzpicture}
\end{equation}
with gauge group of even rank $U(1)_{\kappa} \times U(1)^{2\ell} \times U(1)_{-\kappa}$. The vertices are labeled according to the colored index next to them. This theory admits $2\ell$ independent mass parameters $\{m_i\}_{i=1,\dots,2\ell}$, which, in terms of the redundant mass parameters $\{ m_e \}_{e\in \sQ_1}$ associated with the scaling action on the edges, read 
\begin{equation}\label{eq:massnonlin2}
	m_i = \sum_{v=1}^{2\ell} \left( m_{0 \to v}  - m_{\tilde{0} \to v} \right)\delta_{v,i} .
\end{equation}
The FI parameters at the vertices $v=1,\dots, 2\ell$ are redundant, and are subject to the constraint $\sum_{v=1}^{2\ell} \zeta^{v} =0$.\par
\medskip
The partition function is shown in \S\ref{app:nonlin2} to be 
\begin{equation}\label{eq:Znonlinear2}
	\mz_{\eqref{eq:nonlinear2}} = \mz_{\csnode{1}{\kappa}} \mz_{\csnode{1}{-\kappa}} \prod_{i=1}^{2\ell} \left. \mz_{\gnode{1}\text{---}\fnode{2}} (\zeta^{\prime}_i , m^{\prime}_i) \right\rvert_{\zeta^{\prime}_i= m_i , m^{\prime}_i= \zeta^{i}} .
\end{equation}
This form is consistent with Conjecture \ref{myconj1}.

\subsection{Abelian Chern--Simons trinion}
\label{sec:trinion}
Conjecture \ref{myconj1} is based on general features of $\mN=4$ Chern--Simons-matter theories, and does not necessarily require a description in terms of a quiver.\par
As one such tractable example, consider $U(1)_{\kappa_1} \times U(1)_{\kappa_2} \times U(1)_{\kappa_3}$ gauge theory with one hypermultiplet in the tri-fundamental representation, with the Chern--Simons levels subject to 
\begin{equation}\label{eq:ATT1}
	\frac{1}{\kappa_1} +\frac{1}{\kappa_2} +\frac{1}{\kappa_3} =0.
\end{equation}
This is a toy version of the models in \cite{Assel:2022row}, obtained by gauging an Abelian version of the `trinion' theory, depicted as: 
\begin{equation}\label{eq:Abtrinion}
		\begin{tikzpicture}[baseline=0]
			\node (a) at (1,0) {$\csnode{1}{\kappa_1}$};
			\node (b) at (-0.75,0.75) {$\csnode{1}{\kappa_2}$};
			\node (c) at (-0.75,-0.75) {$\csnode{1}{\kappa_3}$};
			\draw (0,0) -- (a);
			\draw (0,0) -- (b);
			\draw (0,0) -- (c);
		\end{tikzpicture}
\end{equation}\par
In this model, $M_{A,\eqref{eq:Abtrinion}}$ is a point, $M_{B,\eqref{eq:Abtrinion}}$ is one-quaternionic dimensional, and there is a flavor symmetry $\ct_B\cong \C^{\ast}$ acting on it with a single fixed point $\left( M_{B,\eqref{eq:Abtrinion}}^{\ct_B}\right)_{\mathrm{red}} = \{ 0\}$.\par
\begin{thm}
	Conjecture \ref{myconj1} holds for \eqref{eq:Abtrinion}.
\end{thm}
\begin{proof}
Let $m$ denote the mass of the hypermultiplet, equivalently the equivariant parameter for the action of $\ct_B$. The partition function is evaluated as:
\begin{subequations}
\begin{align}
	\mz_{\eqref{eq:Abtrinion}} &=\sqrt{-\ii} \int_{\R^3} \left( \prod_{v=1}^{3}\dd \sigma_v e^{\ii \pi \kappa_v \sigma_v^2 } \right) \frac{1}{\ch (\sigma_1 + \sigma_2 + \sigma_3 + m)} \\
	& =\left( \prod_{v=1}^{3}\frac{1}{ \sqrt{\kappa_v}} \right) \frac{1}{\ch (m)} ,
\end{align}
\end{subequations}
where \eqref{eq:ATT1} and the conventions \eqref{eq:phaseZ} have been used after integrating over the $\sigma_v$ variables. The partition function consists of just one term, which equals a twisted trace on the Verma module for $M_B^{\ct_B}$, and with a coefficient given by the product of three pure $U(1)_{\kappa_v}$ Chern--Simons partition functions. 
\end{proof}\par
\medskip
The computation generalizes straightforwardly.
\begin{thm}
	Conjecture \ref{myconj1} holds for the Chern--Simons-matter theory with gauge group $\prod_{v=1}^{\ell} U(1)_{\kappa_v}$, subject to 
	\begin{equation}\label{eq:ATTgen}
		\sum_{v=1}^{\ell} \frac{1}{\kappa_v} =0 ,
	\end{equation}
	 and one hypermultiplet in the representation of highest weight $( 1,\dots ,1)$.
\end{thm}
\begin{proof}
Also in this case, $M_A$ is a point, $M_B$ is one-quaternionic dimensional, and $\ct_B\cong \C^{\ast}$ with $\left( M_B^{\ct_B}\right)_{\mathrm{red}} = \{ 0\}$. Using \eqref{eq:ATTgen}, the partition function evaluates to 
\begin{equation}
	\mz = \left( \prod_{v=1}^{\ell} \frac{1}{\sqrt{\kappa_v}} \right) \frac{1}{\ch (m)} .
\end{equation}
\end{proof}

\subsection{Non-Abelian two-node quiver}
\label{sec:ex2nodeUN}
A tractable non-Abelian example is provided by the non-Abelian version of \eqref{eq:ex2node}, that is,
\begin{equation}\label{eq:ex2nodeUN}
		\csnode{N}{\kappa}\text{---} \csnode{N}{-\kappa}\ .
\end{equation}
$\ct_A \cong \C^{\ast}$ as in the Abelian case, and let $\zeta$ denote the equivariant parameter, taken for definiteness in the chamber $\zeta \in (0,\infty)$. It is further assumed that $\kappa \ge 2N-1$. The Hilbert series of the A-branch is \cite[Eq.(3.16)]{Li:2023ffx}
\begin{equation}\label{eq:ex2HSUN}
	\hs_{\C [ M_{A,\eqref{eq:ex2nodeUN}}] } \left( q;x\right) = \prod_{a=1}^{N}\frac{1-q^{\kappa-a+1}}{(1-q^{a})\left(1-xq^{\frac{\kappa}{2}-a+1}\right)\left(1-x^{-1}q^{\frac{\kappa}{2}-a+1}\right)}  ,
\end{equation}
which implies \cite{Li:2023ffx,Marino:2025uub}
\begin{equation}
	M_{A, \eqref{eq:ex2nodeUN}} = \mC \left( \gnode{N}\text{---}\fnode{\kappa} \right) 
\end{equation}
as singular varieties. In the limit $q \to 1$, \eqref{eq:ex2HSUN} becomes
\begin{equation}
	\lim_{q\to 1} \hs_{\C [ M_{A,\eqref{eq:ex2nodeUN}}] } \left( q;x\right) = (-1)^ N \frac{\kappa !}{N! (\kappa - N)!}  \left( \frac{x^{\frac{1}{2}}}{1-x} \right)^{2N} .
\end{equation}
The combinatorial factor is an integer, and one reads off the (untwisted) trace 
\begin{equation}
	\chi^{A,\eqref{eq:ex2nodeUN}} (x) =  \left( \frac{x^{\frac{1}{2}}}{1-x} \right)^N .
\end{equation}\par
\medskip
Under the assumption $\kappa \ge 2N-1$, the partition function of \eqref{eq:ex2nodeUN} is given by \cite{Nosaka:2017ohr}:
\begin{equation}\label{eq:Z2nodeUN}
	\mz_{\eqref{eq:ex2nodeUN}} (\zeta) = \mz_{\csnode{N}{\kappa}}  \mz_{\csnode{N}{-\kappa}} \prod_{a,b=1}^{N} \ch \left(\frac{ \zeta - \ii (a-b)}{\kappa}\right)^{-1} .
\end{equation}
Using the Cauchy identity, this expression yields
\begin{subequations}
\begin{align}
	\mz_{\eqref{eq:ex2nodeUN}} (\zeta) & =  \mz_{\csnode{N}{\kappa}}  \mz_{\csnode{N}{-\kappa}} ~ x^{\frac{N^2}{2}} \sum_{\lambda} \left( -x\right)^{\lvert \lambda \rvert} s_{\lambda} \left( q^{\rho}\right)^{2} \label{eq:Z2nodeqdima}\\
		&=  \mz_{\csnode{N}{\kappa}}  \mz_{\csnode{N}{-\kappa}}  ~ \sum_{n_1 > n_2 > \cdots > n_N \ge 0} (-1)^{\sum_{a=1}^{N} n_a} x^{\sum_{a=1}^{N} \left( n_a+ \frac{1}{2} \right)} \left( \prod_{1 \le a < b \le N} \frac{ \sin \left( \frac{\pi}{\kappa} \left( n_a - n_b\right)\right) }{ \sin \left( \frac{\pi}{\kappa} \left( a - b\right)\right)} \right)^2 .  \label{eq:Z2nodeqdimb}
\end{align}
\end{subequations}
In \eqref{eq:Z2nodeqdima}, the sum is over all partitions $\lambda$, $s_{\lambda}$ is the Schur polynomial, $x=e^{-2 \pi \frac{\zeta}{\kappa}}$, $q=e^{\ii \frac{2 \pi}{\kappa} }$, and $\rho$ is the Weyl vector of $U(N)$ with $\rho_a = \frac{N+1}{2}-a$. The expression $s_{\lambda} \left( q^{\rho}\right)$ is the quantum dimension of the irreducible $U_q(\mathfrak{gl}_N)$-module labeled by $\lambda$. Changing the summation variables as $n_a = \lambda_a +N-a$ gives \eqref{eq:Z2nodeqdimb}.\par
Formula \eqref{eq:Z2nodeqdimb} supports Conjecture \ref{myconj1}.

\subsection{Non-Abelian two-node quiver with unequal ranks}
\label{sec:ex2nodediff}
The non-Abelian example in \S\ref{sec:ex2nodeUN} generalizes to unequal gauge ranks:
\begin{equation}\label{eq:ex2nodediff}
		\csnode{N_1}{\kappa}\text{---} \csnode{N_2}{-\kappa}\ .
\end{equation}
It is convenient to define $N := \min (N_1, N_2) $, $\nu := \lvert N_1-N_2\rvert $, and to impose $\kappa \ge 2N+\nu -1$.\par
As in \S\ref{sec:ex2nodeUN}, $\ct_A = \C^{\ast}$ and $M_{B, \eqref{eq:ex2nodediff}}$ is a point, while it is shown in \cite[Eq.(3.21)]{Li:2023ffx} that:
\begin{equation}
	\C \left[ M_{A, \eqref{eq:ex2nodediff}} \right]= \C \left[ \mC \left( \gnode{N}\text{---}\fnode{\kappa-\nu} \right) \right] .
\end{equation}\par
The partition function is computed in \cite{Nosaka:2017ohr} and gives:
\begin{equation}\label{eq:Zex2nodeNA}
	\mz_{\eqref{eq:ex2nodediff}} (\zeta) = e^{- \ii \pi \frac{ \nu}{6\kappa}\eta_{\text{fr}} } \mz_{\csnode{N_1}{\kappa}}  \mz_{\csnode{N_2}{-\kappa}} \prod_{a=1}^{N_1} \prod_{b=1}^{N_2}\ch \left( \frac{\zeta}{\kappa} - \frac{\ii}{\kappa} \left( a-b + \frac{\nu}{2} \right)\right)^{-1} .
\end{equation}
The overall phase vanishes when $\nu=0$, and the framing anomaly $\eta_{\text{fr}}$, given explicitly in \cite[Eq.(3.13)]{Nosaka:2017ohr}, is not relevant here.
Expanding in the fugacity $x=e^{-2\pi \frac{\zeta}{\kappa}}$ as in \eqref{eq:Z2nodeqdima}, \eqref{eq:Zex2nodeNA} attains a form consistent with Conjecture \ref{myconj1}. The twist is affected by $\nu \ne 0$.

\subsection{Non-Abelian A-type quiver with \texorpdfstring{($\kappa, 0, \dots, 0,-\kappa$)}{(k,0,...,0,k)}}
\label{sec:NAlin}
The non-Abelian Chern--Simons-matter quiver
\begin{equation}\label{eq:exNAlin}
		\csnode{N}{\kappa}\text{---} \overbrace{\csnode{N}{0} \text{---} \cdots \text{---} \csnode{N}{0}}^{k-1} \text{---} \csnode{N}{-\kappa}\ ,
\end{equation}
generalizes \S\ref{sec:exk00k} to non-Abelian gauge groups and \S\ref{sec:ex2nodeUN} to $A_{k+1}$-type Dynkin quivers.
$\kappa k \ge 2N-1$ is assumed. In this case $\ct_A\cong \C^{\ast}$, with associated FI parameter $\zeta$, and $\ct_B \cong( \C^{\ast})^{k-1}$, with associated redundant mass parameters $\{m_i\}_{i=1,\dots, k}$.\par
\medskip
\begin{lem}\label{lem:ZNAlin}
	The sphere partition function of \eqref{eq:exNAlin} equals 
	\begin{equation}\label{eq:ZNAlin}
		\mz_{\eqref{eq:exNAlin}} = \frac{1}{N!} \int_{\R^N} \prod_{1 \le a < b \le N} \sh (\sigma_a - \sigma_b)^2  \prod_{a=1}^{N} \left( \prod_{i=1}^{k}\frac{1}{\ch (\kappa \sigma_a - m_i)} \right) e^{\ii 2 \pi \zeta \sigma_a} \dd \sigma_a .
\end{equation}
\end{lem}
This result, proven in \S\ref{app:NAblin}, is utilized to evaluate the partition function.\par
For concreteness, the parameters are taken in the principal Weyl chamber, $\zeta>0$ and $m_1 > m_2 > \cdots > m_k$. From \eqref{eq:ZNAlin}, the vacua are in one-to-one correspondence with maps 
\begin{equation}
\label{eq:alphaNAlin}
	\alpha \ : \ \{ 1, \dots N \} \longrightarrow \{1, \dots, k \}
\end{equation}
modulo permutations. Closing the integration contour for each variable in the upper half-plane, \eqref{eq:ZNAlin} receives contributions from the poles 
\begin{equation}
	\left\{ \sigma_a = \frac{1}{\kappa} \left[  m_{\alpha (a)} + \ii \left( \frac{1}{2} + n_a \right) \right] , \ n_a \in \N \right\}_{a=1, \dots, N}
\end{equation}
for every $\alpha$ as in \eqref{eq:alphaNAlin}. Weyl-equivalent configurations contribute equally. One can thus restrict to $\alpha (a) \le \alpha (b)$ if $a<b$, and moreover restrict to $n_a \ge n_b$ if $a<b$ when $\alpha (a)= \alpha (b)$; summing over these restricted contributions cancels the overall $1/N!$.\par
Computing \eqref{eq:ZNAlin} explicitly by residues gives:
\begin{equation}\label{eq:ZfullNAlin}
\begin{aligned}
	\mz_{\eqref{eq:exNAlin}} &=  \mz_{\csnode{N}{\kappa}}  \mz_{\csnode{N}{-\kappa}} \sum_{\alpha} \ii^{- N(k-1)} e^{\ii 2 \pi \pair{\frac{\zeta}{\kappa} }{m} } \sum_{n \in \mathfrak{N}_{\alpha}} (-1)^{\sum_{a=1}^{N} n_a} x^{\sum_{a=1}^{N} \left( n_a+\frac{1}{2}\right) } \\
	 	& \times \prod_{a=1}^{N}\prod_{\substack{i=1 , \dots, k \\ i \ne \alpha (a)}} \frac{1}{\sh \left( m_{\alpha (a)} - m_i \right)} ~\left( \prod_{1 \le a <b \le N} \frac{ \sh \left( \frac{m_{\alpha (a)} - m_{\alpha (b)} +\ii (n_a - n_b) }{\kappa} \right)}{\sin \left( \frac{\pi (b-a)}{\kappa}\right)}  \right)^2 ,
\end{aligned}
\end{equation}
where $x = e^{- 2 \pi \frac{\zeta}{\kappa}}$, the pairing is 
\begin{equation}
	\pair{\zeta}{m} = \sum_{a=1}^{N} \zeta m_{\alpha (a)} ,
\end{equation}
and the summation is over the set $\mathfrak{N}_{\alpha} \subseteq \N^N$ defined as
\begin{equation}
	\mathfrak{N}_{\alpha} = \left\{  n \in \N^N \ : \  n_a > n_b \text{ if } a<b \text{ and } \alpha (a) = \alpha (b)\right\} .
\end{equation}
The last term in \eqref{eq:ZfullNAlin} is a Laurent polynomial in the variables $y_{\alpha(a)\alpha(b)}^{\frac{1}{2}} :=e^{\frac{\pi}{\kappa} (m_{\alpha(a)} - m_{\alpha(b)})}$. Each monomial carries a product of $y_{\alpha(a)\alpha(b)}^{\frac{1}{2}}$, independent of $n$, which contributes to the twisted trace on $\scH_{\alpha}^{B}$, while the $n$-dependent coefficient of each monomial contributes to the twisted trace on $\scH_{\alpha}^{A}$.\par
\begin{rmk}
\begin{itemize}
	\item If $k=1$, there is only one possible $\alpha$, and \eqref{eq:ZfullNAlin} reproduces \eqref{eq:Z2nodeqdimb}.
	\item When $\kappa=1$, \eqref{eq:exNAlin} is dual to $U(N)$ gauge theory with $k$ hypermultiplets in the fundamental representation. The partition function of the latter is indeed given by \eqref{eq:ZNAlin} at $\kappa=1$. In this case, only vacua with $\alpha (a) \ne \alpha (b)$ contribute non-trivially, giving 
	\begin{equation}
	\begin{aligned}
		\mz_{\gnode{N}\text{---}\fnode{k}} = \sum_{1 \le \alpha_1 < \alpha_2 < \cdots <\alpha_N \le k} & \ii^{- N(k+1)} e^{\ii 2 \pi \pair{\zeta}{m} } \sum_{n \in \N^N} (-1)^{\sum_{a=1}^{N} n_a} x^{\sum_{a=1}^{N} \left( n_a+\frac{1}{2}\right) } \\
		 & \ \times \prod_{a=1}^{N} \prod_{i \in \{1, \dots, k\} \setminus \{ \alpha_1, \dots, \alpha_N \} } \frac{1}{\sh \left( m_{\alpha_a} - m_i \right)}  .
	\end{aligned}
	\end{equation}
	It satisfies Conjecture \ref{c:GO}.
	\item If $\kappa >1$, there is always one set of vacua, labeled by an edge $e=1, \dots, k$, with $\alpha(a)=e$ for all $a=1, \dots, N$. In this case the summand factorizes neatly, the A-branch twisted trace is the same as in \S\ref{sec:ex2nodeUN}, and the B-branch twisted trace is the $N$-th power of the twisted trace in the Abelian theory in \S\ref{sec:exk00k}.
\end{itemize}
\end{rmk}
In general, the summand factorizes into products of twisted traces, one on a Verma module over the A-branch and one on a Verma module over the B-branch. The explicit form of the twisted traces depends on the fixed point $\alpha$.

\subsection{Non-Abelian three-node bad quiver}
Another non-Abelian example with unequal ranks is
\begin{equation}\label{eq:ex3node121}
		\csnode{1}{\kappa}\text{---} \csnode{2}{0}\text{---} \csnode{1}{-\kappa}\ .
\end{equation}
The central node violates Hypothesis \ref{hyp:goodQ}. Nevertheless, isolating the central node would give a $U(2)$ gauge theory with two hypermultiplets; this belongs to the class of `bad' theories redeemed in \cite{Yaakov:2013fza}. It is thus possible to study moduli spaces of vacua and sphere partition functions of \eqref{eq:ex3node121}, provided the FI parameter is non-vanishing.\par
\medskip
The partition function of \eqref{eq:ex3node121} is
\begin{equation}\label{eq:Z3node121a}
	\mz_{\text{\eqref{eq:ex3node121}}} = \frac{1}{2}\int_{\R^4} \dd \sigma_1 \dd \sigma_{2,1}\dd \sigma_{2,2} \dd \sigma_3 \frac{e^{\ii \pi \kappa \left( \sigma_1^2 - \sigma_3^2 \right) + 2 \pi \ii \zeta  \left( \sigma_{2,1} +\sigma_{2,2}\right)}\sh \left( \sigma_{2,1} - \sigma_{2,2}\right)^2  }{ \prod_{a=1}^{2} \ch (\sigma_1 - \sigma_{2,a} +m) \ch (\sigma_{2,a}-\sigma_3 - m) }.
\end{equation}
For generic parameters, the integral can be evaluated explicitly:
\begin{equation}\label{eq:Z3node121b}
	\mz_{\text{\eqref{eq:ex3node121}}} (\zeta, m)= \mz_{\csnode{1}{\kappa}} \mz_{\csnode{1}{-\kappa}} ~\frac{ e^{4 \pi \ii \zeta m } }{\sh \left(\zeta \right)^2} .
\end{equation}
The proof of this equality is a direct computation, crucially using $\zeta \ne 0$, and is deferred to \S\ref{app:121}. Expanding the denominator in power series casts \eqref{eq:Z3node121b} in the form of a trace, compatible with Conjecture \ref{myconj1}.\par
In fact, using $\sh \left(\zeta \right)^2 = \ch \left( \zeta + \frac{\ii}{2} \right)\ch \left( \zeta - \frac{\ii}{2} \right)$, \eqref{eq:Z3node121b} equals 
\begin{equation}
	\mz_{\text{\eqref{eq:ex3node121}}} (\zeta, m) = \mz_{\csnode{1}{\kappa}} \mz_{\csnode{1}{-\kappa}} ~\mz_{\gnode{2}\text{---}\fnode{2}}(\zeta, m)  ,
\end{equation}
with the partition function of $\gnode{2}\text{---}\fnode{2}$ as computed in \cite[Eq.(3.34)]{Yaakov:2013fza}.

\section{Summary and outlook}
This work investigates the sphere quantization formalism \cite{Gaiotto:2023hda} in 3d $\mN=4$ theories for which Hypothesis \ref{hyp:massive} fails. 
\begin{itemize}
	\item The quantization of Coulomb and Higgs stacks of quiver gauge theories with non-minimal charges has been discussed in \S\ref{sec:higherquant}. Corollary \ref{cor:Zhighertrace} shows that the partition function is expressed as a sum of as many terms as the number of vacua, and each term is proportional to a twisted trace over a tensor product of an $\mA_{\hbar}^{\mC}$-module and an $\mA_{\hbar}^{\mH}$-module.\par
	By Corollary \ref{cor:hypertwtr}, at least for some theories with a 1-form symmetry, the twisted trace factorizes into a Coulomb and a Higgs part, thereby attaining a form analogous to the Gaiotto--Okazaki Conjecture \ref{c:GO}.
	\item The quantization of the moduli stacks of vacua $\mM_A$ and $\mM_B$ of 3d $\mN=4$ Chern--Simons-matter theories has been considered in \S\ref{sec:VermaCS}. The analysis of the sphere partition function led to Conjecture \ref{myconj1}. It states that the partition function is expressed as a sum of twisted traces over tensor products of two Verma modules, one over the quantization of the A-branch and one over the quantization of the B-branch. The coefficient of each summand is a supersymmetric pure Chern--Simons partition function on the three-sphere. \par
		The twist depends on the Chern--Simons levels, as well as on the choice of vacuum. In general, it prevents a factorization of the trace into A- and B-part; yet, in some cases, the twisted trace in each summand factorizes, paralleling Conjecture \ref{c:GO}.\par
	\item An interplay between the two main results is highlighted in \S\ref{sec:magQprime}. A prescription is given to associate to any Abelian Chern--Simons-matter A-type quiver an Abelian A-type quiver with higher charges. It is proposed (Conjecture \ref{myconjmag}) that the two theories have the same moduli spaces of vacua, not just as symplectic singularities but as stacks, and moreover that they have equal partition functions.
\end{itemize}
The latter result resonates with work of Chan--Leung \cite{Chan:2024oya}. For the familiar 3d $\mN=4$ gauge theories, the authors of \cite{Chan:2024oya} produce the Higgs branch given the Coulomb branch \emph{as a stack}. Here this observation is promoted to a relation between Chern--Simons-matter theories and `ordinary' 3d $\mN=4$ theories, where not only the moduli spaces of vacua, but also the partition functions are uniquely fixed by the stacky structure of the A-branch.\par
\medskip
The rest of this outlook section lists open avenues for future investigation.

\subsubsection*{Symplectic duality with gerbes and stacks}
The quivers $\sQ^{\prime}$ can be taken independently of their relation to Chern--Simons-matter theories. Their Higgs branches are gerbes over hypertoric varieties, and their Coulomb branches are smooth hypertoric Deligne--Mumford stacks.\par
\S\ref{sec:highercharge} provides a proof of concept that treating the moduli spaces of vacua as stacks is the appropriate formulation. Indeed, the explicit examples show that theories with the same Coulomb branch \emph{as a stack} also have the same Higgs branch. It is well-known that this statement would fail if one only looks at the Coulomb branch as an algebraic variety. Moreover, it was shown in \S\ref{sec:highercharge} how the sphere quantization distinguishes among Coulomb branches with the same coarse moduli space. These observations are consistent with the work \cite{Chan:2024oya}, which, given different stacky Coulomb branches with the same coarse moduli space, outputs different Higgs branches.\par
\medskip
The results of the present work thus raise the question of how to enhance the formulation of symplectic duality to encode these structures. In physics terms, it would be interesting to incorporate higher-form symmetries into the symplectic duality. Gerbes over hypertoric GIT quotients, as the one studied herein, provide a concrete benchmark to test these ideas.

\subsubsection*{Dualities}
It has been extensively checked in \S\ref{sec:ex} that there exist pairs of Abelian theories, namely a Chern--Simons-matter theory and the associated quiver theory $\sQ^{\prime}$, having equal moduli stacks of vacua and equal sphere partition functions.\par
It would be desirable to derive these putative dualities from first principles. 
Moreover, it would be extremely interesting to find candidate duals to non-Abelian Chern--Simons-matter theories. The obstacle is identifying a suitable notion of non-minimal charges for representations of $U(N)$.\par

\subsubsection*{Quantized moduli spaces of \texorpdfstring{$\mathcal{N}=3$}{N=3} Chern--Simons-matter theories}
Finally, it would be worthwhile to extend the current analysis, in which the sphere partition function is used to read off modules over non-commutative algebras, and twisted traces on them, to the moduli spaces of vacua of 3d $\mN=3$ Chern--Simons-matter theories. These theories have several maximal branches, which are nonetheless believed to be symplectic singularities. Explicit calculations \cite{Santilli:2020snh} show that the sphere partition function attains a similar form, with a sum of terms that appear amenable to being reinterpreted as twisted traces.

\begin{appendices}
\section{Supersymmetric partition functions}
\label{app:SUSYZ}
\subsection{Sphere partition function}
\label{app:Zreview}

In the conventions of the main text, let $\Th_{\underline{\kappa}}$ be a 3d $\mN \ge 3$ Chern--Simons-matter theory, specified by a (representation of a) quiver $\sQ=(\sQ_0 \sqcup \sF_0, \sQ_1)$, where $\sQ_0$ and $\sF_0$ are the gauge and framing nodes, together with Chern--Simons levels $\{\kappa_v\}_{v\in \sQ_0} \subset \Z$.\par
A 3d $\mN \ge 3$ gauge theory can be placed on the round three-sphere $\cs^3$ preserving a fraction of the supersymmetry. In the conventions of the present paper, $\cs^3$ is taken to have unit radius. To lighten the expressions, the shorthand notation (often used in the physics literature) is adopted:
\begin{equation}
	\sh (\sigma) := 2 \sinh (\pi \sigma) , \qquad \ch (\sigma) := 2 \cosh (\pi \sigma) .
\end{equation}
Localization reduces the partition function to an ordinary integral \cite{Kapustin:2009kz}:
\begin{equation}\label{eq:Z3d}
\begin{aligned}
	\mz_{\Th_{\underline{\kappa}}}= \frac{e^{- \ii \frac{\pi}{4}\eta_{\Th_{\underline{\kappa}}} }}{\prod_{v \in \sQ_0} N_v !} \int_{\mathfrak{t}_{\text{\tiny gauge}}} & \prod_{v \in \sQ_0} \prod_{a=1}^{N_v} \dd \sigma_{v,a} e^{\ii \pi \kappa_v \sigma_{v,a}^2} \prod_{1 \le a < b \le N_v} \sh (\sigma_{v,a} - \sigma_{v,b})^2 \\
	\times & \prod_{e \in \sQ_1} \prod_{a=1}^{N_{\mathsf{h} (e)}}\prod_{b=1}^{N_{\mathsf{t} (e)}} \frac{1}{\ch (\sigma_{\mathsf{h} (e),a} - \sigma_{\mathsf{t} (e),b})} ,
\end{aligned}
\end{equation}
with integration domain the Cartan subalgebra of the gauge algebra, $\mathfrak{t}_{\text{\tiny gauge}} \cong \R^{\sum_{v \in \sQ_0} N_v}$.\par 
The partition function \eqref{eq:Z3d} is specified up to an overall eighth root of unity $e^{- \ii \frac{\pi}{4}\eta_{\Th_{\underline{\kappa}}} }$. It is proposed that, for 3d $\mN \ge 4$ Chern--Simons-matter theories, 
\begin{equation}\label{eq:phaseZ}\boxed{
	\qquad \eta_{\Th_{\underline{\kappa}}}= \sum_{ v \in \sQ_0 \ : \ \kappa_{v} \ne 0 }  N_v^2 \mathrm{sign} (\kappa_v)  . \qquad 
}\end{equation}
This choice agrees with the explicit derivation in \cite[Sec.6]{Marino:2011nm} in the particular cases of pure Chern--Simons and ABJM theories.\par

\begin{rmk}
For the purposes of the partition function, an edge $e \in \sQ_1$ with $\mathsf{t} (e) \in \sQ_0$ and $\mathsf{h} (e) \in \sF_0$ can be replaced by $N_{\mathsf{h} (e)}$ edges, each ending on a framing node $w \in \sF_0$ with $N_w=1$. This replacement is assumed in \S\ref{sec:HigherChargeSphere} for ease of exposition.
\end{rmk}

\subsection{Sphere partition function: Evaluation}
\label{app:evaluation}

The Fourier transform of the function $1/\ch(\sigma)$ will be repeatedly used:
\begin{equation}
\label{eq:FTcosh}
	\frac{1}{\ch (\sigma)} = \int_{-\infty}^{+\infty}\dd \tau \frac{e^{2 \pi \ii \tau \sigma }}{\ch (\tau)} .
\end{equation}

\subsubsection{Explicit calculations: Abelian A-type quiver}
\label{app:Abeliancalc}
This appendix contains the derivation of \eqref{eq:Z3node0b}.\par
Starting with  \eqref{eq:Z3node0a} and using \eqref{eq:FTcosh} on each term in the denominator gives:
\begin{subequations}
\begin{align}
	\mz_{\text{\eqref{eq:ex3node0}}} &= \int_{-\infty}^{+\infty} \dd \sigma_0\int_{-\infty}^{+\infty} \dd \sigma_1 \cdots \int_{-\infty}^{+\infty}\dd \sigma_k \frac{e^{\ii \pi \kappa \left( \sigma_0^2-\sigma_k^2\right) + 2 \pi \ii \sum_{v=0}^{k}\zeta^{v} \sigma_v }}{ \ch \left( \sigma_{k-1}-\sigma_k +m\right)\prod_{v=1}^{k-1}\ch \left( \sigma_{v-1}-\sigma_{v}\right) } \\
	&=\int_{-\infty}^{+\infty} \dd \sigma_0\int_{-\infty}^{+\infty} \dd \sigma_1 \cdots \int_{-\infty}^{+\infty}\dd \sigma_k e^{\ii \pi \kappa \left( \sigma_0^2-\sigma_k^2\right) }\int_{-\infty}^{+\infty} \frac{\dd \tau^1}{\ch (\tau^1)} \cdots \int_{-\infty}^{+\infty} \frac{\dd \tau^{k}}{\ch (\tau^{k})} \notag \\
	& \times \exp \left\{2 \pi \ii \left[ \sum_{v=1}^{k-1} \left(\tau^{v+1} - \tau^{v} + \zeta^{v}\right)\sigma_v + \tau^{k} m + \left(\tau^{1} + \zeta^{0}\right) \sigma_0 + \left(- \tau^{k} +\zeta^{k}\right)\sigma_k \right] \right\} .
\end{align}
\end{subequations}
Integrating over $\sigma_v$, for $1 \le v \le k-1$, produces $\delta  \left(\tau^{v+1} - \tau^{v} + \zeta^{v}\right)$, and therefore 
\begin{equation}
	\tau^{i} = \tau^k + \sum_{v=i}^{k-1} \zeta^{v} , \qquad i =1, \dots, k-1.
\end{equation}
For shortness, it is convenient to define 
\begin{equation}
	\tilde{m}_{i} := \sum_{v=i}^{k-1} \zeta^{v} , \qquad i =1, \dots, k-1, 
\end{equation}
and also let $\tilde{m}_k=0$. Then, 
\begin{subequations}
\begin{align}
	\mz_{\text{\eqref{eq:ex3node0}}}  &=\int_{-\infty}^{+\infty} \dd \sigma_0\int_{-\infty}^{+\infty}\dd \sigma_k e^{\ii \pi \kappa \left( \sigma_0^2-\sigma_k^2\right) }\int_{-\infty}^{+\infty} \dd \tau^{k}   \prod_{i=1}^{k} \frac{1}{\ch (\tau^k + \tilde{m}_{i} )}  \\
		& \times \int_{-\infty}^{+\infty} \dd \tau^1  e^{\ii 2\pi \left[ \tau^{k} m + \left(\tau^{1} + \zeta^{0}\right) \sigma_0 + \left(- \tau^{k} +\zeta^{k}\right)\sigma_k \right]} \delta \left( \tau^{1} - \tau^{k} - \tilde{m}_{1} \right) \notag \\
	&=\sqrt{ \frac{\ii }{\kappa } } \cdot \sqrt{ - \frac{\ii }{\kappa} }  \int_{-\infty}^{+\infty} \dd \tau^k  e^{2 \pi \ii m \tau^{k}} \prod_{i=1}^{k-1} \frac{1}{\ch (\tau^k + \tilde{m}_{i} )} ,  \label{eq:Z3node0app}
\end{align}
\end{subequations}
where the constraint $\zeta^{0} + \zeta^{k} +\tilde{m}_1 =\sum_{v=0}^{k} \zeta^{v} =0$ has been used to simplify the exponent after integration over $\sigma_0, \sigma_k$. Renaming $\tau^k \mapsto \sigma$ gives \eqref{eq:Z3node0b}.\par
At this stage, \eqref{eq:Z3node0app} is given in terms of $(k-1)$ generic mass parameters and one FI parameter. It is possible to parametrize the Cartan subalgebra of the flavor symmetry with $k$ masses $\{ m^{\prime}_i\}_{i=1}^{k}$ subject to $\sum_{i=1}^{k}m^{\prime}_i=0$. It suffices to let $m^{\prime}_k:= - \frac{1}{k-1}\sum_{v=1}^{k-1} v \zeta^{v}$, shift variables $\tau^{k}=\sigma -m^{\prime}_k$ and rename $\zeta^{\prime}:=m$, $m^{\prime}_i:= -\tilde{m}_i + m^{\prime}_k $. The difference from \eqref{eq:Z3node0b} is in the background mixed Chern--Simons term $e^{-\ii 2\pi m m^{\prime}_k}$, which does not affect the analysis.

\subsubsection{Explicit calculations: Alternating Abelian A-type quiver}
\label{app:altcalc}

This appendix contains the proof of \eqref{eq:ZaltAb}.
Shifting variables, all the parameter dependence appears as giving arbitrary masses to the hypermultiplets. The partition function is:
\begin{equation}
	\mz_{\text{\eqref{eq:ex4node}}}  = \int_{\R^{2\ell}} \prod_{v=1}^{2\ell} \dd \sigma_v e^{\ii \pi \kappa \sum_{v=1}^{2\ell} (-1)^{v-1} \sigma_v^2 }\prod_{v=1}^{2 \ell -1} \frac{1}{\ch (\sigma_v - \sigma_{v+1} +m_v )} .
\end{equation}
Using \eqref{eq:FTcosh} on each term in the denominator gives
\begin{subequations}
\begin{align}
	\mz_{\text{\eqref{eq:ex4node}}}  &= \int_{\R^{2\ell}} \prod_{v=1}^{2\ell} \dd \sigma_v \int_{\R^{2 \ell -1}} \prod_{i=1}^{2\ell -1} \frac{\dd \tau^{i} }{\ch (\tau^{i})} e^{\ii \pi \kappa \sum_{v=1}^{2\ell} (-1)^{v-1} \sigma_v^2 + 2 \pi \ii \sum_{i=1}^{2\ell-1} \tau^{i} \left( m_i +\sigma_i - \sigma_{i+1}\right)} \\
		& = \left( \prod_{v=1}^{2\ell} \sqrt{(-1)^{v-1} \frac{\ii}{\kappa}} \right) \int_{\R^{2 \ell -1}} \prod_{i=1}^{2\ell -1} \frac{\dd \tau^{i} }{\ch (\tau^{i})} \exp \left\{  2 \pi \ii\sum_{i=1}^{2\ell -1} \left( m_i \tau^{i}  + (-1)^{i} \frac{\tau^{i} \tau^{i+1}}{\kappa}\right) \right\} ,
\end{align}
\end{subequations}
where the second line stems from direct integration over the variables $\sigma_v$. Then, for $i=2\nu -1$ odd, one defines $\tilde{\sigma}_{\nu} =\frac{\tau^{2\nu-1}}{\kappa}$, $\nu=1, \dots, \ell$; while the integration over $\tau^{i}$ for $i \in 2 \Z $ is performed using \eqref{eq:FTcosh}. The result is
\begin{equation}
	\mz_{\text{\eqref{eq:ex4node}}}  = \int_{\R^{\ell}}\prod_{\nu=1}^{\ell} \frac{\dd \tilde{\sigma}_\nu}{\ch (\kappa \tilde{\sigma}_{\nu})} ~e^{\ii 2\pi\sum_{\nu=1}^{\ell} m_{2\nu-1} \tilde{\sigma}_{\nu} } \prod_{\nu=1}^{\ell -1} \frac{1}{\ch (\tilde{\sigma}_{\nu} - \tilde{\sigma}_{\nu+1} - m_{2\nu})} .
\end{equation}
Renaming the dummy variable $\nu\mapsto v$ and defining $\zeta^{v}=m_{2v-1},\tilde{m}_{v}=m_{2v}$ proves \eqref{eq:ZaltAb}.

\subsubsection{Explicit calculations: Abelian affine A-type quiver}
\label{app:4nodecirc}

This appendix contains the derivation of \eqref{eq:Zcirc4}.
The partition function of \eqref{eq:circular4} is
\begin{equation}
	\mz_{\eqref{eq:circular4}} = \int_{\R^{k+2}} \frac{e^{\ii \pi \kappa \left( \sigma_0^2-\sigma_k^2\right) + 2 \pi \ii \left[ \sum_{v=0}^{k}\zeta^{v} \sigma_v  + \zeta^{\infty} \sigma_{\infty}  \right]} \dd \sigma_0 \cdots \dd \sigma_k \dd \sigma_{\infty} }{ \ch \left( \sigma_{\infty}-\sigma_0+ m\right)\ch \left( \sigma_{\infty}-\sigma_k -m\right) \ch \left( \sigma_{k-1}-\sigma_k +m\right) \prod_{v=1}^{k-1}\ch \left( \sigma_{v-1}-\sigma_{v}\right) }
\end{equation}
Using \eqref{eq:FTcosh} on each term in the denominator, the integration over $\{ \sigma_v \}_{v=1}^{k-1}$ proceeds exactly as in \S\ref{app:Abeliancalc}, arriving at: 
\begin{equation}
\begin{aligned}
	\mz_{\eqref{eq:circular4}} &=\int_{\R^3} \dd \sigma_0 \dd \sigma_k \dd \sigma_{\infty} e^{\ii \pi \kappa \left( \sigma_0^2-\sigma_k^2\right) } \int_{-\infty}^{+\infty} \dd \tau^{k}   \prod_{i=1}^{k} \frac{1}{\ch (\tau^k + \tilde{m}_{i} )} \int_{\R^2} \frac{\dd\tau^0 \dd \tau^{\infty} }{ \ch (\tau^0) \ch (\tau^{\infty})} \\
		& \times e^{\ii 2\pi \left[ \left(\tau^{k} + \tilde{m}_{0} \right) \sigma_0 + \left(- \tau^{k} +\zeta^{k}\right)\sigma_k  + \left( \tau^0 + \tau^{\infty} + \zeta^{\infty} \right) \sigma_{\infty} + \left( \tau^0 + \tau^{\infty}\right) m_{\infty} + \tau^{k} m \right]} ,
\end{aligned}
\end{equation}
where the masses $\tilde{m}_i= \sum_{v=i}^{k-1} \zeta^{v}$ are defined as in \S\ref{app:Abeliancalc}. Integrating over $\sigma_{\infty}$ at this stage sets $\tau^{0}=-\tau^{\infty}-\zeta^{\infty}$. Integrating over $\sigma_0, \sigma_k$ next leads to
\begin{equation}
	\mz_{\eqref{eq:circular4}} = \sqrt{\frac{\ii}{\kappa}} \sqrt{-\frac{\ii}{\kappa}}\int_{-\infty}^{+\infty} \dd \tau^{\infty}\frac{e^{\ii 4\pi \tau^{\infty}m_{\infty} }}{\ch (\tau^{\infty}) \ch (\tau^{\infty}+\zeta^{\infty})}\int_{-\infty}^{+\infty} \dd \tau^k e^{\ii 2 \pi \tau^{k} m} \prod_{i=1}^{k} \frac{1}{\ch (\tau^k + \tilde{m}_i)} .
\end{equation}
Renaming the variables $\sigma_1 = \tau^{\infty} $, $\sigma_2 = \tau^{k} $ and the parameters as in \eqref{eq:circ4param} gives \eqref{eq:Zcirc4}.

\subsubsection{Explicit calculations: Abelian quiver not of Dynkin type}
\label{app:nonlin2}

This appendix derives the partition function \eqref{eq:Znonlinear2} in \S\ref{sec:nonlin2}.\par
The partition function of \eqref{eq:nonlinear2} in terms of redundant FI and mass parameters is
\begin{equation}
	\mz_{\eqref{eq:nonlinear2}} = \int_{\R^2} \dd \sigma_0 \dd \tilde{\sigma}_0 e^{\ii \pi \kappa \left( \sigma_0^2 - \tilde{\sigma}_0^2\right)}\int_{\R^{2\ell}} \prod_{v=1}^{2\ell} \dd \sigma_v \frac{e^{\ii 2 \pi \zeta^{v} \sigma_v} }{\ch (\sigma_0 - \sigma_v + m_{0v})\ch (\tilde{\sigma}_0 - \sigma_v + \tilde{m}_{0v})}
\end{equation}
where $m_{0v},\tilde{m}_{0v}$ are shorthand notations for the masses $m_e$ of the hypermultiplets corresponding to the edges $e$ with $\mathsf{t}(e)=v$ and, respectively, $\mathsf{h}(e)=0$ and $\mathsf{h}(e)=\tilde{0}$.\par
Applying \eqref{eq:FTcosh} to each $1/\ch$ gives:
\begin{equation}
\begin{aligned}
	\mz_{\eqref{eq:nonlinear2}} &= \int_{\R^2} \dd \sigma_0 \dd \tilde{\sigma}_0 e^{\ii \pi \kappa \left( \sigma_0^2 - \tilde{\sigma}_0^2\right)} \int_{\R^{2\ell}} \prod_{v=1}^{2\ell} \dd \sigma_v \int_{\R^{4\ell}} \prod_{i=1}^{2 \ell} \frac{\dd \tau^{i} \dd \tilde{\tau}^{i}}{\ch (\tau^{i}) \ch (\tilde{\tau}^{i})} \\
		& \times \exp \left\{ \ii 2 \pi \sum_{v=1}^{2\ell}  \left[ \sigma_v (- \tau^{v} - \tilde{\tau}^{v} + \zeta^{v} )  + \tau^{v}m_{0v} + \tilde{\tau}^{v} \tilde{m}_{0v} + \sigma_0 \sum_{i=1}^{2\ell} \tau^{i} + \tilde{\sigma}_0 \sum_{i=1}^{2\ell} \tilde{\tau}^{i}  \right]\right\} .
\end{aligned}
\end{equation}
Integrating over $\sigma_v$ produces a delta function, which is used to remove the integration over $\tilde{\tau}^{i}$ setting 
\begin{equation}
	\tilde{\tau}^{i} = -\tau^{i} + \zeta^{i} .
\end{equation}
It is also convenient to define $m_{i}= m_{0i} - \tilde{m}_{0i}$ in \eqref{eq:massnonlin2}. Integrating over $\sigma_0, \tilde{\sigma}_0$ one arrives at:
\begin{equation}
\begin{aligned}
	\mz_{\eqref{eq:nonlinear2}} &= \sqrt{\frac{\ii}{\kappa}}\sqrt{-\frac{\ii}{\kappa}}\int_{\R^{2\ell}}\prod_{i=1}^{2 \ell} \dd \tau^{i} \frac{e^{\ii 2 \pi \tau^{i} m_i}}{\ch (\tau^{i}) \ch (\tau^{i} - \zeta^{i})}  \\
		& \times \exp \left\{ \ii \frac{\pi}{\kappa} \left( \sum_{v=1}^{2\ell} \zeta^{v} \right)^2 -  \ii \frac{\pi}{\kappa}\sum_{i=1}^{2\ell}  \tau^{i} \sum_{v=1}^{2\ell} \zeta^{v} +\ii 2 \pi \sum_{v=1}^{2\ell} \zeta^{v} \tilde{m}_{0v}\right\} .
\end{aligned}
\end{equation}
Imposing the condition $\sum_{v=1}^{2\ell} \zeta^{v} =0$, the partition function completely factorizes:
\begin{equation}
	\mz_{\eqref{eq:nonlinear2}} =e^{\ii 2 \pi \langle \zeta, \tilde{m} \rangle }\sqrt{\frac{\ii}{\kappa}}\sqrt{-\frac{\ii}{\kappa}} \prod_{i=1}^{2 \ell} \int_{\R}\dd \tau \frac{e^{\ii 2 \pi m_i  \tau}}{\ch (\tau) \ch (\tau - \zeta^{i})}  ,
\end{equation}
with each term equal to the partition function of $U(1)$ gauge theory with two hypermultiplets of mass $0$ and $m^{\prime}_i=\zeta^{i}$, and FI parameter $\zeta^{\prime, i}=m_i$.

\subsubsection{Explicit calculations: Non-Abelian A-type quiver}
\label{app:NAblin}
The goal of this appendix is to derive \eqref{eq:ZNAlin}. The proof makes repeated use of classical identities:
\begin{itemize}
\item The Cauchy determinant formula: 
	\begin{equation}
	\label{eq:Cauchydet}
		\frac{\prod_{1 \le a < b \le N} \sh (\sigma_{1,a} - \sigma_{1,b} ) \prod_{1 \le a < b \le N} \sh (\sigma_{2,a} - \sigma_{2,b} ) }{\prod_{a=1}^{N} \prod_{b=1}^{N}\ch (\sigma_{1,a} - \sigma_{2,b})} = \det_{1 \le a,b \le N} \frac{1}{\ch (\sigma_{1,a} - \sigma_{2,b})} .
	\end{equation} 
\item The Vandermonde determinant formula:
	\begin{equation}
	\label{eq:Vdmdet}
		\prod_{1 \le a< b \le N} \sh (\sigma_{a} - \sigma_{b} ) = (-1)^{N(N-1)/2} \det_{1 \le a,b \le N} \left[ e^{2\pi \sigma_b \left( a - \frac{N+1}{2}\right)} \right].
	\end{equation} 
\item Andr\'eief's identity \cite{Andreief}:
	\begin{equation}
	\label{eq:Andreief}
		\frac{1}{N!} \int_{\R^N} \det_{1 \le a,b \le N} \left[ \Psi_b (\sigma_a) \right] \det_{1 \le a,b \le N} \left[ \Phi_b (\sigma_a)\right] \prod_{a=1}^{N} \dd \sigma_a = \det_{1 \le a,b \le N} \left[ \int_{\R} \dd \tilde{\sigma} \Psi_a (\tilde{\sigma}) \Phi_b (\tilde{\sigma}) \right] .
	\end{equation}
	for arbitrary integrable functions $\{\Psi_b, \Phi_b \}_{b=1, \dots, N}$.
\end{itemize}
Lemma \ref{lem:ZNAlin} will be a corollary of a more general statement.
\begin{lem}\label{lem:genericlinearCS}
	The partition function of 
	\begin{equation}\label{eq:CSlinear}
		\csnode{N}{\kappa_0}\text{---} \overbrace{\csnode{N}{\kappa_1} \text{---} \cdots \text{---} \csnode{N}{\kappa_{k-1}}}^{k-1}\text{---} \csnode{N}{\kappa_{k}} 
	\end{equation}
	is given by 
	\begin{equation}\label{eq:ZCSlinear}
	\begin{aligned}
		\mz_{\eqref{eq:CSlinear}} & =  \det_{1 \le a,b \le N} \left[  \int_{\R^{k+1}} \dd \tilde{\sigma}_{0}\cdots \dd \tilde{\sigma}_{k} e^{\ii \pi \sum_{v=0}^{k}\left[ \kappa_v \tilde{\sigma}_v^2 +2\zeta_{v} \tilde{\sigma}_{v} \right]} \right. \\
			& \left. \times \prod_{v=1}^{k}\frac{1}{\ch (\tilde{\sigma}_{v-1} - \tilde{\sigma}_{v})} \times e^{2\pi \left[ \tilde{\sigma}_{k} \left( b - \frac{N+1}{2}\right) +\tilde{\sigma}_{0} \left( a - \frac{N+1}{2}\right) \right] } \right] ,
	\end{aligned}
	\end{equation}
	where $\zeta^v$ are $\Z$-linear combinations of the FI parameters and of $\kappa_v m_i$, where $m_i$ are the mass parameters.
\end{lem}
\begin{proof}
	Labeling the nodes by $v=0,1\dots, k$, the partition function of \eqref{eq:CSlinear} is
	\begin{equation}
	\begin{aligned}
		\mz_{\eqref{eq:CSlinear}} = \frac{1}{(N!)^{k+1}} \int_{\R^{(k+1)N}} &\prod_{v=0}^{k} \left( \prod_{1 \le a < b \le N} \sh (\sigma_{v,a} - \sigma_{v,b})^2 \prod_{a=1}^{N}  e^{\ii \pi \kappa_v \sigma_{v,a}^2 + \ii 2 \pi \zeta^{v} \sigma_{v,a}} \dd \sigma_{v,a} \right) \\
		& \times \prod_{v=1}^{k} \prod_{a,b=1}^{N} \frac{1}{\ch \left( \sigma_{v-1,a} - \sigma_{v,b} \right)} .
	\end{aligned}
	\end{equation}
	Here a shift of variables has been used to absorb the dependence on the mass parameters into the redundant FI parameters $\zeta^v$. Starting from the last integral, applying \eqref{eq:Cauchydet} and \eqref{eq:Vdmdet} once shows that 
	\begin{align}
		 &\prod_{1 \le a < b \le N} \sh (\sigma_{k-1,a} - \sigma_{k-1,b})^2  \prod_{a,b=1}^{N} \frac{1}{\ch (\sigma_{k-1,a} - \sigma_{k,b} )} \cdot  \prod_{1 \le a < b \le N} \sh (\sigma_{k,a} - \sigma_{k,b})^2 \\
		 = &(-1)^{\frac{N(N-1)}{2}} \prod_{1 \le a < b \le N} \sh (\sigma_{k-1,a} - \sigma_{k-1,b}) \det_{1 \le a,b \le N}  \frac{1}{\ch (\sigma_{k-1,a} - \sigma_{k,b} )} \det_{1 \le a,b \le N} \left[ e^{2\pi \sigma_{k,b} \left( a - \frac{N+1}{2}\right)} \right] . \notag
	\end{align}
	The product of exponentials $\prod_{a=1}^{N}  e^{\ii \pi \kappa_k \sigma_{k,a}^2 + \ii 2 \pi \zeta^{k} \sigma_{k,a}}$ can be written as a determinant and combined with the Vandermonde determinant. Then, one applies Andr\'eief's identity \eqref{eq:Andreief} for the integral over the variables $\sigma_{k,a}$, which gives 
	\begin{align}
		& \prod_{1 \le a < b \le N} \sh (\sigma_{k-1,a} - \sigma_{k-1,b})^2  \frac{1}{N!} \int_{\R^N} \prod_{a=1}^{N} \dd \sigma_{k,a} e^{\ii \pi \kappa_{k} \sum_{a=1}^{N} \sigma_{k,a}^2 + \ii 2 \pi \zeta_{k} \sum_{a=1}^{N}\sigma_{k,a} }  \notag \\
		& \hspace{128pt} \times  \prod_{1 \le a < b \le N} \sh (\sigma_{k,a} - \sigma_{k,b})^2 \prod_{a,b=1}^{N} \frac{1}{\ch (\sigma_{k-1,a} - \sigma_{k,b} )} \\
	 	&= (-1)^{\frac{N(N-1)}{2}} \prod_{1 \le a < b \le N} \sh (\sigma_{k-1,a} - \sigma_{k-1,b}) \det_{1 \le a,b \le N} \left[ \int_{\R} \dd \tilde{\sigma}_{\ell} \frac{e^{\ii \pi \kappa_{k} \tilde{\sigma}_{k}^2 +\ii 2 \pi \zeta^{k} \tilde{\sigma}_{k}}}{\ch (\sigma^{(\ell-1)}_a - \tilde{\sigma}_{k})}e^{2\pi \tilde{\sigma}_{\ell} \left( b - \frac{N+1}{2}\right)} \right] .  \notag
	\end{align}
	From here, one iteratively applies \eqref{eq:Cauchydet} followed by \eqref{eq:Andreief} for every $v=k-1,\dots, 1$, arriving at the last integral 
	\begin{equation}
	\begin{aligned}
		&\mz_{\eqref{eq:CSlinear}} = \frac{1}{N!}\int_{\R^N}  \left( \prod_{1 \le a < b \le N} \sh (\sigma^{(0)}_a - \sigma^{(0)}_b)\right) (-1)^{\frac{N(N-1)}{2}}  \prod_{a=1}^{N} e^{\ii \pi \kappa_0 \sigma_{0,a}^2 + \ii 2 \pi \zeta_0 \sigma_{0,a} } \\
		&\times \det_{1 \le a , b \le N} \left[  \int_{\R^{k}} \dd \tilde{\sigma}_{1}\cdots \dd \tilde{\sigma}_{k} \frac{e^{2\pi \tilde{\sigma}_{k} \left( b - \frac{N+1}{2}\right)}}{\ch (\sigma_{0,a} - \tilde{\sigma_1})} \times  e^{\ii\pi \sum_{v=1}^{k}\left[ \kappa_v \tilde{\sigma}_{v}^2 +2\zeta^{v} \tilde{\sigma}_{v} \right]} \prod_{v=2}^{k}\frac{1}{\ch (\tilde{\sigma}_{v-1} - \tilde{\sigma}_{v})} \right] .
	\end{aligned}
	\end{equation}
	Using \eqref{eq:Vdmdet}, writing the integrand in the first line as a determinant, and applying \eqref{eq:Andreief} one last time, gives \eqref{eq:ZCSlinear}.
\end{proof}

\begin{proof}[Proof of Lemma \ref{lem:ZNAlin}]
	Substituting 
	\begin{equation}
	\label{eq:CSk00k}
		\kappa_0=\kappa, \qquad \kappa_{1 \le v \le k-1}=0, \qquad  \kappa_k=-\kappa ,
	\end{equation}
	into \eqref{eq:ZCSlinear} gives the determinant of an $N \times N$ matrix, in which each entry is almost identical to the partition function studied in \S\ref{app:Abeliancalc}. The calculation in \S\ref{app:Abeliancalc} shows that 
	\begin{equation}
		\mz_{\eqref{eq:exNAlin}} = \det_{1 \le a , b \le N} \left[  \int_{\R} \dd \tilde{\sigma} e^{- \ii 2 \pi \zeta_+ \tilde{\sigma}} e^{2 \pi \tilde{\sigma} \left(N+1-a-b \right)} \prod_{i=1}^{k} \frac{1}{\ch \left( \kappa \tilde{\sigma} +\tilde{m}_i \right) } \right] ,
	\end{equation}
	where $\tilde{m}_i := \sum_{v=i}^{k-1} \zeta^v$, $\zeta_+ := \sum_{v=0}^{k} \zeta^v$, and the redundancy of the FI parameters has been used to check that, for \eqref{eq:exNAlin}, it is possible to set $\zeta^{k}=\frac{\zeta_+}{2}$. Other choices of conventions for the mass parameters would simply differ by an overall integral power of $e^{\ii \frac{\pi}{\kappa} \left(2 \zeta^{k} \zeta_+ - \zeta_+^2\right)}$.\par
	Finally, renaming $\zeta_+ \mapsto - \zeta, \tilde{m}_i \mapsto - m_i$, and applying \eqref{eq:Andreief} in the reverse direction, followed by \eqref{eq:Vdmdet}, proves \eqref{eq:ZNAlin}.
\end{proof}

\subsubsection{Explicit calculations: Non-Abelian three-node quiver}
\label{app:121}

This appendix contains the proof of \eqref{eq:Z3node121b}.\par
The starting point is \eqref{eq:Z3node121a}. Shifting variables $\sigma_1=\sigma_1^{\text{\tiny new}} -m$ and $\sigma_3=\sigma_3^{\text{\tiny new}} -m$, the $m$-dependence is removed from the denominator at the expense of introducing $e^{2 \pi \ii \kappa m( -\sigma_1 + \sigma_3)}$ in the integrand. The next step consists of rewriting 
\begin{subequations}
\begin{align}
	\frac{ \sh \left( \sigma_{2,1} - \sigma_{2,2}\right)^2}{\prod_{a=1}^{2} \ch (\sigma_1 - \sigma_{2,a}) \ch (\sigma_3 - \sigma_{2,a})} & = \sum_{a,b=1}^{2} (-1)^{a+b} f_{ab} (\sigma_1,\sigma_{2,1},\sigma_{2,2},\sigma_3 ) , \\
	f_{ab} (\sigma_1,\sigma_{2,1},\sigma_{2,2},\sigma_3 ) &:= e^{\pi (\sigma_1 + \sigma_3)} \frac{ e^{-\pi \left( \sigma_{2,a} +\sigma_{2,b} \right)}}{\ch (\sigma_1 - \sigma_{2,a}) \ch (\sigma_3 - \sigma_{2,b})} .
\end{align}
\end{subequations}
Then, 
\begin{subequations}
\begin{align}
	\mz_{\text{\eqref{eq:ex3node121}}} &= \frac{1}{2}\sum_{a,b=1}^{2} (-1)^{a+b} \mz_{ab} , \\
	\mz_{ab}&:=  \int_{\R^4} \dd \sigma_1 \dd \sigma_{2,1}\dd \sigma_{2,2} \dd \sigma_3  e^{\ii \pi \kappa \left( \sigma_1^2 - \sigma_3^2 \right) + 2 \pi \ii \left[  \zeta  \left( \sigma_{2,1} +\sigma_{2,2}\right) -\kappa m \sigma_1 + \kappa m \sigma_3 \right] } e^{\pi (\sigma_1 + \sigma_3 - \sigma_{2,a} - \sigma_{2,b})} \notag \\
		& \times \int_{\R} \frac{\dd \tau^{1} }{\ch (\tau^{1})}\int_{\R} \frac{\dd \tau^{2} }{\ch (\tau^{2})} e^{2 \pi \ii \left[ \tau^1 (\sigma_1 - \sigma_{2,a}) - \tau^2 (\sigma_3 - \sigma_{2,b})\right]} \label{eq:Zbadintermediate}
\end{align}
\end{subequations}
where \eqref{eq:FTcosh} was applied to $f_{ab}$. The integration over $\sigma_{2,1},\sigma_{2,2}$ can now be performed, obtaining
\begin{equation}
\begin{aligned}
	\mz_{ab} & =  \int_{\R^2} \dd \sigma_1\dd \sigma_3 e^{\ii \pi \kappa \left( \sigma_1^2 - \sigma_3^2 \right) }\int_{\R^2} \frac{\dd \tau^{1} }{\ch (\tau^{1})}\frac{\dd \tau^{2} }{\ch (\tau^{2})}  e^{ 2 \pi \ii \left[ \sigma_1 \left( \tau^{1} -\kappa m- \frac{\ii}{2} \right) + \sigma_3 \left( -\tau^{2} +\kappa m- \frac{\ii}{2} \right) \right] } \\
	& \times \delta \left( \tau^{2} \delta_{1,b}-\tau^{1} \delta_{1,a} + \zeta + \frac{\ii}{2} \left(  \delta_{1,a} +  \delta_{1,b} \right)  \right) \delta \left( \tau^{2} \delta_{2,b}-\tau^{1} \delta_{2,a} + \zeta + \frac{\ii}{2} \left(  \delta_{2,a} +  \delta_{2,b} \right)  \right) .
\end{aligned}
\end{equation}
At this stage, if $a=b$, the integral contains $\delta (\zeta)$, and vanishes for generic $\zeta \ne 0$. Moreover, it is straightforward to check that $\mz_{12}=\mz_{21}$, which is a direct consequence of the Weyl invariance of \eqref{eq:Z3node121a}. Therefore, 
\begin{subequations}
\begin{align}
	\mz_{\text{\eqref{eq:ex3node121}}} = -\mz_{12} &= \frac{1}{\sh \left(\zeta \right)^2} \int_{\R} \dd \sigma_1 e^{\ii \pi \kappa\sigma_1^2 +\ii 2\pi \sigma_1 \left( \zeta-\kappa m\right) } \int_{\R}\dd \sigma_3 e^{-\ii \pi \kappa \sigma_3^2 +\ii 2\pi \sigma_3 \left( \zeta + \kappa m\right) } \\
	&=\frac{1}{\kappa \sh \left(\zeta \right)^2} e^{4 \pi \ii \zeta m } .
\end{align}
\end{subequations}
This proves \eqref{eq:Z3node121b}.

\subsection{Superconformal index}
\label{app:sci}

The superconformal index is a trace over the BPS Hilbert space $\mathscr{H}_{\text{BPS}}$ of the theory quantized on a sphere, refined by characters of the isometries acting on the moduli space. Explicitly: 
\begin{equation}
\mI_{\Th_{\underline{\kappa}}} \left( q, \tilde{q};x,y\right) = \tr_{\mathscr{H}_{\text{BPS}}} \left(  (-1)^{2J_3} q^{J_3 + R_A} \tilde{q}^{J_3 + R_B} x^{J_C} y^{J_H}\right) ,
\end{equation}
where $R_A$ (respectively $R_B$) is the generator of the $\C^{\ast}$-action on the A-branch (respectively B-branch) by scaling. $\mI_{\Th_{\underline{\kappa}}}$ is equivalently presented as a contour integral, see \cite[App.A]{Li:2023ffx} for a review. The contour integral approach is used for the explicit calculations.\par
The index recovers the Hilbert series of the A- and B-branch rings: 
\begin{equation}
	\lim_{\tilde{q} \to 0} \mI_{\Th_{\underline{\kappa}}} \left( q, \tilde{q};x,y\right)  = \hs_{\C \left[ M_A \right]} (q;x) , \qquad \lim_{q \to 0} \mI_{\Th_{\underline{\kappa}}} \left( q, \tilde{q};x,y\right)  = \hs_{\C \left[ M_B \right]} (\tilde{q};y) .
\end{equation}\par
For the theories with non-minimal charges discussed in \S\ref{sec:highercharge}, the superconformal index can be computed in two ways. A direct approach is the residue integral including hypermultiplet contributions in a representation of non-minimal highest weight. Alternatively, it is obtained using a discrete Molien--Weyl projection formula, see in particular \cite[App.B]{Nawata:2023rdx}.\par

\end{appendices}
{\footnotesize 
\bibliography{CBandHB}
}
\end{document}